\documentclass[11pt,letterpaper]{article}

\usepackage{mathtools} \mathtoolsset{showonlyrefs=true}
\usepackage{accents}
\usepackage{amsfonts,amsmath,amssymb,amsthm,bm}
\usepackage{array}
\usepackage[english]{babel}
\usepackage{bbm}
\usepackage[font=sl,labelfont=bf]{caption}
\usepackage[usenames,dvipsnames]{color} 
\usepackage{colortbl}
\usepackage{enumerate}
\usepackage{float} \restylefloat{table}
\usepackage[left=1in,top=1in,bottom=1in,right=1in]{geometry}
\usepackage{graphicx}
\usepackage[colorlinks=true,pdfstartview=FitV,linkcolor=blue,citecolor=blue,urlcolor=blue]{hyperref}
\usepackage{mathrsfs}
\usepackage[framemethod=default]{mdframed}
\usepackage{subcaption}
\usepackage{url}
\usepackage{verbatim}
\usepackage{multirow}

\newtheorem{theorem}{Theorem}

\newtheorem{corollary}[theorem]{Corollary}
\newtheorem{assumption}[theorem]{Assumption}
\theoremstyle{definition}

\theoremstyle{remark}
\newtheorem{remark}[theorem]{Remark}

\numberwithin{equation}{section}
\numberwithin{theorem}{section}

\def\cF{\mathcal{F}}

\def\bE{\mathbb{E}}
\def\bF{\mathbb{F}}

\def\bN{\mathbb{N}}

\def\bP{\mathbb{P}}

\def\bR{\mathbb{R}}

\def\bZ{\mathbb{Z}}

\def\sF{\mathscr{F}}

\newcommand{\sgn}{{\rm sgn}}
\newcommand{\wh}{\widehat}
\newcommand{\wt}{\widetilde}

\definecolor{Red}{rgb}{0,0,0}
\newcommand{\Red}{\color{Red}}
\definecolor{DR}{rgb}{0,0,0}

\definecolor{Green}{rgb}{0.2,0.5,0.2} 
\definecolor{Blue}{rgb}{0,0,0}
\newcommand{\Blue}{\color{Blue}}
\definecolor{PaleGrey}{rgb}{0.6,0.6,0.6} 
\definecolor{Purple}{rgb}{0.5,0.00,1} 

\title{Estimation of Tempered Stable L\'{e}vy Models of Infinite Variation}

\author{Jos\'{e} E. Figueroa-L\'{o}pez\thanks{Department of Mathematics \& Statistics, Washington University in St. Louis, St. Louis, MO 63130, U.S.A. (Email: {\tt figueroa-lopez@wustl.edu}). Research supported in part by the NSF Grants: DMS-2015323, DMS-1613016.}
\and Ruoting Gong\thanks{Department of Applied Mathematics, Illinois Institute of Technology, Chicago, IL 60616, U.S.A. (Email: {\tt rgong2@iit.edu}).}
\and Yuchen Han\thanks{Department of Mathematics \& Statistics, Washington University in St. Louis, St. Louis, MO 63130, U.S.A. (Email: {\tt y.han@wustl.edu}).}}

\date{\today}

\begin{document}

\maketitle

\begin{abstract}
Truncated realized quadratic variations (TRQV) are among the most widely used high-frequency-based nonparametric methods to estimate the volatility of a process in the presence of jumps. Nevertheless, the {\Red truncation level is known to critically affect} its performance, especially in the presence of infinite variation jumps. In this paper, we study the optimal truncation level, in the mean-square error sense, for a semiparametric tempered stable L\'evy model. We obtain a novel closed-form 2nd-order approximation of the optimal threshold in a high-frequency setting. As an application, we propose a new estimation method, which combines iteratively an approximate semiparametric method of moment estimator and TRQVs with the newly found small-time approximation for the optimal threshold. The method is tested via simulations to estimate the volatility and the Blumenthal-Getoor index of a generalized CGMY model and, via a localization technique, {\Red to estimate} the integrated volatility of a Heston type model with CGMY jumps. {\Red Our method is found to} outperform other alternatives proposed in the literature {\Blue when working with a L\'evy process  (i.e., the volatility is constant), or when the index of jump intensity $Y$ is larger than $3/2$ in the presence of stochastic volatility}.

\vspace{0.3 cm}


\vspace{0.3 cm}

\end{abstract}

\section{Introduction}

L\'{e}vy processes have experienced a revival in the past $20$ years, propelled by the need for more realistic modeling of irregular behavior in many phenomena of nature and society. These fundamental building blocks of stochastic modeling have been widely applied in many fields, including statistical physics, meteorology, seismology, insurance, finance, and telecommunication.
While, in principle, L\'{e}vy models offer ideal conditions for estimation purposes, two main bottlenecks complicate their estimation. Firstly, their marginal distributions often lack tractable or closed-form representations. In those situations, the marginal distributions must be approximated by Fourier, Monte Carlo, or other numerical methods, which makes the estimation slower and noisier. The second issue comes from the need to handle high-frequency sampling data of the process. This type of data has been widely available in finance during the last $15$ years and is increasingly {more common} in other fields. The two just-mentioned issues have rendered traditional statistical methods such as likelihood and Bayesian estimation unfeasible. {\Blue We refer the reader to \cite{ContTankov:2004} for more information about L\'evy processes and their application in finance, \cite{Masuda} for a survey on frequentist parametric estimation of L\'evy process, and \cite{FigueroaQiKuffner} for more information about Bayesian estimation methods.}

In this paper, we study {a new method} for the estimation of the parameters of a L\'{e}vy model. A semiparametric model is considered in which the jump component is assumed to exhibit small jumps that behave like those of {\Blue a $Y$-stable} L\'{e}vy process. Specifically, the class of tempered stable processes introduced in \cite{FigueroaLopezGongHoudre:2016}\footnote{The term ``tempered stable" is understood here in {\Red a more} general sense than in several classical sources of financial mathematics (e.g., \cite{Applebaum:2004}, \cite{ContTankov:2004}, \cite{KyprianouSchoutensWilmott:2005}) and even more general than in \cite{Rosinski:2007}. In fact, such class of L\'{e}vy processes is called {\Red the} \emph{tempered-stable-like} L\'{e}vy processes in \cite{FigueroaLopezGongHoudre:2016}.} and \cite{FigueroaLopezOlafsson:2019(1)} is considered. We focus on models of infinite variation (i.e., {$Y\in(1,2)$}), which are arguably the most relevant for financial applications (see \cite{AitSahaliaJacod:2009}, \cite{Belomestny:2010}, and \cite{FigueroaLopez:2012}). The estimation of semiparametric L\'{e}vy models of infinite jump variation under high-frequency data is not well developed. Jacod and Todorov \cite{JacodTodorov:2014} were the first to introduce an \emph{efficient} estimator of the integrated volatility of an It\^{o} semimartingale model in the presence of a L\'{e}vy jump model of infinite variation with Blumenthal-Getoor index {$\beta\in(1,3/2)$} or when the jump component is symmetric. {\Blue  Their estimator is based on locally estimating the volatility from the empirical characteristic function of the increments of the process over time blocks of decreasing length.} Recently, Mies \cite{Mies:2019} proposed an efficient estimation method for L\'{e}vy models based on a type of approximate semiparametric method of moments with scaling. Specifically, for some suitable moment functions $f_{1},f_{2},\dots,f_{m}$ and a scaling factor $u_{n}\rightarrow\infty$, \cite{Mies:2019} {proposed to look} for the parameters $\wh{\boldsymbol{\theta}}=(\wh{\theta}_{1},\ldots,\wh{\theta}_{m})$ such that
\begin{align}\label{eq:MF00}
\frac{1}{n}\sum_{i=1}^{n}f_{j}\big(u_{n}\Delta_{i}^{n}X\big)-\bE_{\wh{\boldsymbol{\theta}}}\Big(f_{j}\big(u_{n}\Delta^{n}_{i}\wt{Z}\big)\Big)=0,\quad j=1,\ldots,m,
\end{align}
where $\wt{Z}$ is the superposition of a Brownian motion and independent stable L\'{e}vy processes closely approximating $X$ in a certain sense. The distribution measure $\bP_{\boldsymbol{\theta}}$ of $\wt{Z}$ depends on some parameters $\boldsymbol{\theta}$, including the volatility $\sigma$ of $X$, and $\bE_{\boldsymbol{\theta}}(\cdot)$ denotes the expectation with respect to $\bP_{\boldsymbol{\theta}}$. Above, $\Delta_{i}^{n}L:=L_{t_{i}}-L_{t_{i-1}}$ is the $i$-th increment of a generic process $(L_{t})_{t\geq 0}$ given $n$ evenly spaced random samples $L_{t_{0}},\ldots,L_{t_{n}}$ over a fixed time interval $[0,T]$ (i.e., $t_{i}=ih_{n}$ with $h_{n}=T/n$). If $X$ were assumed to {\Red follow} a parametric L\'{e}vy model and we replaced $\bE_{\wh{\boldsymbol{\theta}}}(f_{j}(u_{n}\Delta^{n}_{i}\wt{Z}))$ with $\bE_{\wh{\boldsymbol{\theta}}}(f_{j}(u_{n}\Delta^{n}_{i}X))$ in \eqref{eq:MF00}, we will recover a standard Method of Moment Estimator (MME). However, we are assuming that $X$ is semiparametric and that it can be approximated closely enough by {a} parametric L\'{e}vy model $\wt{Z}$. The scaling $u_{n}$, which is taken to converge to $\infty$ at the order of $1/\sqrt{\ln(n)/n}$, is also a new feature of this method compare to the standard MME.

The moment functions $f_{1},\ldots,f_{m}$ and the scaling factor $u_{n}$ in \eqref{eq:MF00} critically affect the performance of the estimators. To determine an appropriate scaling $u_{n}$, we connect it to the threshold parameter $\varepsilon_{n}$ of a Truncated Realized Quadratic Variation (TRQV),
\begin{align}\label{TRQVb}
\text{TRQV}_{n}(\varepsilon_{n})=\sum_{i=1}^{n}\big(\Delta_{i}^{n}X\big)^{2}{\bf 1}_{\{|\Delta_{i}^{n}X|\leq\varepsilon_{n}\}},
\end{align}
which is known to be a consistent estimator for the integrated volatility of a general semimartingale model. {\Red Again, above $\Delta_i^nX=X_{t_i}-X_{t_{i-1}}$ and we are assuming regular sampling observations $X_{t_1},\dots,X_{t_n}$ with $t_i=ih_n$ and $h_n=T/n$. Next, note that} by taking $f_{1}(x)=x^{2}{\bf 1}_{\{|x|\leq 1\}}$ in \eqref{eq:MF00}, we {\Red recover the TRQV \eqref{TRQVb}, which suggests the relationship} $u_{n}=1/\varepsilon_{n}$. That is, $1/u_{n}$ plays the same role as the threshold in TRQV.

Recently, \cite{FigueroaLopezMancini:2019} studied the problem of optimal thresholding of TRQV {\Red \eqref{TRQVb}} under the mean-square error. Specifically, in the case of a L\'{e}vy process with volatility $\sigma$, it is shown that the threshold $\varepsilon=\varepsilon^{\star}_{n}$ that minimizes the mean-square error, $\bE((\text{TRQV}_{n}(\varepsilon)-\sigma^{2}T)^{2})$, solves the equation:
\begin{align*}
\varepsilon^{2}+2(n-1)\bE\big(b_{1,h_{n}}(\varepsilon)\big)-2T\sigma^{2}=0,
\end{align*}
where {\Red $b_{1,h_{n}}(\varepsilon):=X_{h_n}^{2}{\bf 1}_{\{|X_{h_n}|\leq\varepsilon\}}$}. By analyzing the small-time asymptotic behavior of $\bE(b_{1,h_{n}}(\varepsilon))$ (i.e., when $n\rightarrow\infty$ so that $h_{n}\rightarrow 0$), \cite{FigueroaLopezMancini:2019} {\Red proved} that the optimal threshold $\varepsilon^{\star}_{n}$ for a L\'{e}vy process with a $Y$-stable jump component behaves like
\begin{align}\label{eq:EquepsApprx20a}
\varepsilon^{\star}_{n}\sim\sqrt{(2-Y)\sigma^{2}h_{n}\ln(1/h_{n})},\quad n\rightarrow\infty,
\end{align}
where {\Red $h_{n}=T/n$} is the time span between observations and, as usual, $a_{n}\sim b_{n}$ means $a_{n}/b_{n}\rightarrow 1$ as $n\rightarrow\infty$. The proportionality constant $\sqrt{2-Y}$ roughly tells us that the higher the jump activity is, the lower the optimal threshold has to be if we want to discard the higher noise represented by the small jumps. This fact opens the door to an iterative method to estimate $\sigma^{2}$. We can first estimate $Y$ and $\sigma^{2}$ using, for instance, the method of moments \eqref{eq:MF00}. We can then use the TRQV with the threshold $\wh{\varepsilon}^{\,\star}_{n}=\sqrt{(2-\wh{Y})\wh{\sigma}^{2}h_{n}\ln(1/h_{n})}$.

In this paper, we first extend the result of \cite{FigueroaLopezMancini:2019} to allow for a general tempered stable L\'{e}vy process. Furthermore, we {propose a new approximation for} $\varepsilon^{\star}_{n}$ of the form:
\begin{align}\label{eq:EquepsApprx20}
\wt{\varepsilon}_{n}^{\,\star}:=\sqrt{(2-Y)\sigma^{2}h_{n}\ln\Big(\frac{1}{h_{n}}\Big)+2\sigma^{2}h_{n}\ln\bigg(\frac{(2-\alpha)\sigma}{C}\bigg)},
\end{align}
where $C$ controls the overall intensity of jumps. The approximation \eqref{eq:EquepsApprx20} says that if $C$ is small (relative to $\sigma$) then the threshold can be loosened up (in fact, $\wt{\varepsilon}_{n}^{\,\star}\nearrow\infty$ as $C\searrow 0$ as it should be). In practice $C$ is small compare to $\sigma$ and \eqref{eq:EquepsApprx20} provides a significant correction compare to \eqref{eq:EquepsApprx20a}. We then proceed to {devise} a new method to estimate the volatility, the index of jump activity $Y$, and $C$ by combining a variation of the {approximate semiparametric} method of moments in \cite{Mies:2019}, TRQVs, and the approximate optimal threshold \eqref{eq:EquepsApprx20}. Compared to \cite{Mies:2019} we introduce simpler moment functions $f_{1},\ldots,f_{m}$, and a systematic and objective method to tune the scaling factor $u_{n}$ in \eqref{eq:MF00}. The performance of the proposed procedure is superior to the efficient methods of \cite{JacodTodorov:2014} and \cite{Mies:2019}. Finally, as in \cite{JacodTodorov:2014}, we use a localization technique to estimate the integrated volatility of an It\^{o} semimartingale. Specifically, the idea is to split the time horizon into small blocks where the process is approximately L\'{e}vy and, hence, its volatility level can be estimated using our method. For values of $Y\geq 1.5$, our method outperforms the method proposed by \cite{JacodTodorov:2014}.

The rest of this paper is organized as follows. Section \ref{Sec:Model} provides the framework and assumptions as well as some known preliminary results from the literature. Section \ref{sec:MainThm} obtains the asymptotic behavior of $\bE(b_{1,h_{n}}(\varepsilon))$ and derives \eqref{eq:EquepsApprx20a}. The second-order approximation \eqref{eq:EquepsApprx20} is derived in Section \ref{sec:NumIll} as well as a numerical assessment of the approximations in the case of a CGMY jump component. The new method to estimate the parameters of a tempered stable L\'{e}vy model is presented in Section \ref{Sect:NewEstsigma} together with an analysis of its performance via Monte Carlo simulations. The proofs are deferred to an appendix section.

\section{The Model and Some Preliminary Results}\label{Sec:Model}

Throughout, $\bR_{+}:=[0,\infty)$ and $\bR_{0}:=\bR\backslash\{0\}$, and we let $(\Omega,\sF,\bF,\bP)$ be a complete filtered probability space on which all stochastic processes are defined, where $\bF:=(\sF_{t})_{t\in\bR_{+}}$ satisfies the usual conditions. We consider a L\'{e}vy process $X:=(X_{t})_{t\in\bR_{+}}$ of the form
\begin{align}\label{eq:MnMdlX}
X_{t}=\sigma W_{t}+J_{t},\quad t\in\bR_{+},
\end{align}
where $W:=(W_{t})_{t\in\bR_{+}}$ is a Wiener process and $J:=(J_{t})_{t\in\bR_{+}}$ is an independent pure-jump {tempered stable} L\'{e}vy process with L\'{e}vy triplet $(b,0,\nu)$. The L\'{e}vy measure $\nu$ is assumed to be absolutely continuous with a density $s:\bR_{0}\rightarrow\bR_{+}$ of the form
\begin{align}\label{TmpStbLD}
s(x):=\frac{{\nu(dx)}}{dx}:=\big(C_{+}{\bf 1}_{(0,\infty)}(x)+C_{-}{\bf 1}_{(-\infty,0)}(x)\big)q(x)\,|x|^{-1-Y},\quad x\in\bR_{0}.
\end{align}
Here, $C_{\pm}>0$, $Y\in(1,2)$, and $q:\bR_{0}\rightarrow\bR_{+}$ is a bounded Borel-measurable function. Concretely, we make the following assumptions on $q$.

\begin{assumption}\label{assump:Funtq}
\hfill
\begin{itemize}
\item [(i)] $q(x)\rightarrow 1$, as $x\rightarrow 0$;
\item [(ii)] There exist $\alpha_{\pm}\neq 0$ such that
    \begin{align*}
    \int_{(0,1]}\big|q(x)-1-\alpha_{+}x\big|x^{-Y-1}dx+\int_{[-1,0)}\big|q(x)-1-\alpha_{-}x\big|{\Red |x|^{-Y-1}}dx<\infty;
    \end{align*}
\item [(iii)] $\displaystyle{\limsup_{|x|\rightarrow\infty}\frac{|\ln q(x)|}{|x|}<\infty}$;
\item [(iv)] For any $\varepsilon>0$, $\displaystyle{\inf_{|x|<\varepsilon}q(x)>0}$;
\item [(v)] $\displaystyle{\int_{|x|>1}q(x)^{2}|x|^{-1-Y}dx<\infty}$.
\end{itemize}
\end{assumption}

\begin{remark}
The class of L\'{e}vy processes considered above is sometimes termed \emph{tempered stable} processes (or \emph{tempered-stable-like} processes as in \cite{FigueroaLopezGongHoudre:2016}) and includes a wide range of models appearing in finance. Roughly, the conditions above amount to say that the small jumps of $X$ behave like those of a $Y$-stable L\'{e}vy process. We refer the reader to \cite{FigueroaLopezOlafsson:2019(2)} for further background about this class. The parameter $Y$ is called the \emph{index of jump activity} and coincides with the Blumenthal-Getoor index, which controls the jump activity of $X$ in that $\sum_{s\in(0,t]}|\Delta X_{s}|^{\gamma}<\infty$ for all {$\gamma>Y$} and $t>0$, where $\Delta X_{s}:=X_{s}-X_{s-}$ is the jump of $X$ at time $s$. The range of Y considered here (namely, $Y\in(1,2)$) is the most relevant for financial applications based on several econometric studies of high-frequency financial data (cf. \cite{AitSahaliaJacod:2009} and \cite{FigueroaLopez:2012}) and short-term option pricing data (cf. \cite{FigueroaLopezOlafsson:2019(2)}).
\end{remark}

Using a density transformation technique in \cite[Section 6.33]{Sato:1999}, we can change the probability measure from $\bP$ to another locally absolutely continuous measure $\wt{\bP}$, under which $J$ is a $Y$-stable L\'{e}vy process and $W$ is a standard Brownian motion independent of $J$. Concretely, let
\begin{align*}
\wt{\nu}(dx):=\big(C_{+}{\bf 1}_{(0,\infty)}(x)+C_{-}{\bf 1}_{(-\infty,0)}(x)\big)|x|^{-Y-1}dx,\qquad\wt{b}:=b+\int_{0<|x|\leq 1} x(\tilde{\nu}-\nu)(dx).
\end{align*}
Note that $\wt{\nu}$ is the L\'{e}vy measure of a $Y$-stable L\'{e}vy process and, also,
\begin{align*}
\wt{\nu}(dx)=e^{\varphi(x)}\,\nu(dx),\quad\text{with}\quad\varphi(x):=-\ln q(x).
\end{align*}
Next, define $\wt{\bP}$ such that, for any $t\in\bR_{+}$,
\begin{align}\label{eq:DenTranTildePP}
\ln\left(\frac{d\wt{\bP}\big|_{\sF_{t}}}{d\bP\big|_{\sF_{t}}}\right)=U_{t}:=\lim_{\varepsilon\rightarrow 0}\left(\sum_{s\in(0,t]:|\Delta J_{s}|>\varepsilon}\varphi(\Delta J_{s})+t\int_{|x|>\varepsilon}\big(e^{-\varphi(x)}-1\big)\tilde{\nu}(dx)\right).
\end{align}
By virtue of \cite[Theorem 33.1]{Sato:1999}, a necessary and sufficient condition for the measure transformation from $\bP$ to $\wt{\bP}$ to be well defined is given by
\begin{align*}
\int_{\bR_{0}}\big(e^{\varphi(x)/2}-1\big)^{2}\nu(dx)<\infty,
\end{align*}
which can be shown to follow from Assumption \ref{assump:Funtq}$-$(i) \& (ii) (cf. \cite[Lemma 2.1]{FigueroaLopezOlafsson:2019(2)}). Under $\wt{\bP}$, $J$ is a L\'{e}vy process with L\'{e}vy triplet $(\wt{b},0,\wt{\nu})$, and $W$ is a standard Brownian motion which is independent of $J$. In particular, under $\wt{\bP}$, the centered process $Z:=(Z_{t})_{t\in\bR_{+}}$, given by
\begin{align*}
Z_{t}:=J_{t}-t\wt{\gamma},\quad\wt{\gamma}:=\wt{\bE}(J_{1})=\wt{b}+\int_{|x|>1}x\,\wt{\nu}(dx),
\end{align*}
is a strictly $Y$-stable process with its skewness, scale, and location parameters given by  $(C_{+}-C_{-})/(C_{+}+C_{-})$, $\left\{(C_{+}+C_{-})\Gamma(-Y)|\cos(\pi Y/2)|\right\}^{1/Y}$, and $0$, respectively. Let $p_{Z}$ denote the marginal density of $Z_{1}$ under $\wt{\bP}$. It is well known (cf. \cite[(14.37)]{Sato:1999} and references therein) that
\begin{align*}
p_{Z}(z)\sim C_{\pm}|z|^{-Y-1},\quad\text{as }\,z\rightarrow\pm\infty,\,\,\,\text{respectively},
\end{align*}
so that
\begin{align*}
\wt{\bP}\big(\!\pm\!Z_{1}>z\big)=\frac{{C_{\pm}}}{Y}\,z^{-Y}+O\big(z^{-2Y}\big),\quad z\rightarrow\infty.
\end{align*}

The processes $U:=(U_{t})_{t\in\bR_{+}}$ and $Z$ can be expressed in terms of the jump-measure $N(dt,dx)$ of the process $J$ and its compensator $\wt{N}(dt,dx):=N(dt,dx)-\wt{\nu}(dx)dt$ (under $\wt{\bP}$), as follows:
\begin{align}\label{eq:DecompUpm}
U_{t}&=\wt{U}_{t}+\eta t:=\int_{0}^{t}\int_{\bR_{0}}\varphi(x)\wt{N}(ds,dx)+t\eta,\\
\label{eq:DecompJZpm} J_{t}&=Z_{t}+t\wt{\gamma}:=Z^{+}_{t}+Z^{-}_{t}+t\wt{\gamma},
\end{align}
where
\begin{align*}
Z^{+}_{t}:=\int_{0}^{t}\!\int_{(0,\infty)}x\wt{N}(dt,dx),\quad Z^{-}_{t}:=\int_{0}^{t}\!\int_{(-\infty,0)}x\wt{N}(dt,dx),\quad\eta:=\int_{\bR_{0}}\big(e^{-\varphi(x)}-1+\varphi(x)\big)\wt{\nu}(dx).
\end{align*}
The existence of the integral defining $\eta$ follows from Assumption \ref{assump:Funtq}$-$(i) \& (ii). Clearly, $Z^{+}:=(Z^{+}_{t})_{t\in\bR_{+}}$ and $-Z^{-}:=(-Z^{-}_{t})_{t\in\bR_{+}}$ are independent one-sided $Y$-stable processes with scale, skewness, and location parameters given by $\left(C_{\pm}|\Gamma(-Y)\cos(\pi Y/2)|\right)^{1/Y}$, $1$, and $0$, respectively, so that
\begin{align*}
\wt{\bP}\big(\!\pm\!Z_{1}^{\pm}>z\big)&=\frac{C_{\pm}}{Y}\,z^{-Y}+O\big(z^{-2Y}\big),\quad z\rightarrow\infty,\\
\wt{\bE}\Big(e^{\mp Z_{t}^{\pm}}\Big)&=\exp\bigg(C_{\pm}\Gamma(-Y)\cos\bigg(\frac{\pi Y}{2}\bigg)\text{sgn}(1-Y)t\bigg)<\infty.
\end{align*}
Moreover, it can be shown that (cf. \cite[Lemma 2.1]{FigueroaLopezGongHoudre:2017}) there {\Blue exists a universal constant} $K\in(0,\infty)$, such that for any $z>0$,
\begin{align}\label{eq:1stOrderEstTailUpm}
\wt{\bP}\big(\!\pm\!Z_{1}^{\pm}>z\big)&\leq Kz^{-Y}.
\end{align}
{\Red Combining \eqref{eq:DecompJZpm} and \eqref{eq:1stOrderEstTailUpm}, we deduce that there exists a constant $\wt{K}\in(0,\infty)$ such that, for any $z>0$,
\begin{align}
\label{eq:1stOrderEstTailZ} 
\wt{\bP}\big(\!\pm\!Z_{1}>z\big)&\leq\wt{K}z^{-Y}.
\end{align}
Furthermore, using (14.34) in \cite{Sato:1999} and an argument similar to that in the proof of \cite[Lemma 2.1]{FigueroaLopezGongHoudre:2017}, we can show that:
\begin{align}\label{eq:1stOrderEstDenZ}
p_{Z}(\pm z)&\leq\wt{K}\,z^{-Y-1},\\
\label{eq:2ndOrderEstDenZ} \Big|p_{Z}(\pm z)-C_{\pm} z^{-Y-1}\Big|&\leq\wt{K}\big(z^{-Y-1}\wedge z^{-2Y-1}\big),
\end{align}
where above, without loss of generality, we use the same constant $\widetilde{K}$ as in \eqref{eq:1stOrderEstTailZ}.}

\section{Main Result} \label{sec:MainThm}

The TRQV, defined as
\begin{align}\label{eq:TRV}
\wh{\sigma}^{2}_{n}(\varepsilon)=\frac{1}{T}\sum_{i=1}^{n}\big(\Delta_{i}^{n}X\big)^{2}{\bf 1}_{\{|\Delta_{i}^{n}X|\leq\varepsilon\}},
\end{align}
is one of the most {commonly used estimators} for the integrated volatility of an It\^{o} semimartingale. Above, $\Delta_{i}^{n}X:=X_{t_{i}}-X_{t_{i-1}}$ for $i=1,\ldots,n$, where $X_{t_{0}},X_{t_{1}},\ldots,X_{t_{n}}$ are evenly spaced {observations} of $X$ over a fixed time horizon $[0,T]$, so that $t_{i}=t_{i,n}=ih_{n}$ for $i=0,1,\ldots,n$, with $h_{n}:=T/n$. One of its drawbacks is the necessity of tuning the threshold $\varepsilon$ up, which strongly affects the performance of the estimator. It is shown in \cite{FigueroaLopezMancini:2019} that, for a L\'{e}vy process $X$ with volatility $\sigma>0$, there exists a unique threshold $\varepsilon=\varepsilon^{\star}_{n}$, which minimizes the mean-square error, $\bE((\wh{\sigma}^{2}_{n}(\varepsilon)-\sigma^{2})^{2})$. {Furthermore, the minimizer $\varepsilon^{\star}_{n}$ is such that}
\begin{align}\label{eq:OptTholdgghn}
{\varepsilon^{\star}_{n}\to{}0,}\quad\frac{\varepsilon^{\star}_{n}}{\sqrt{h_{n}}}\rightarrow\infty,\quad\text{as }\,n\rightarrow\infty,
\end{align}
and {solves} the equation
\begin{align}\label{eq:Equeps}
\varepsilon^{2}+2(n-1)\bE\big(b_{1,h_{n}}(\varepsilon)\big)-2T\sigma^{2}=0,
\end{align}
where $b_{1,h_{n}}(\varepsilon):=X_{h_{n}}^{2}{\bf 1}_{\{|X_{h_{n}}|\leq\varepsilon\}}$. Therefore, in order to determine the asymptotic behavior of the optimal threshold $\varepsilon^{\star}_{n}$, we need to study the asymptotic behavior of $\bE(b_{1,h}(\varepsilon))$ as both $h\rightarrow 0+$ and $\varepsilon=\varepsilon(h)\rightarrow 0+$ in such a way that $\varepsilon(h)/\sqrt{h}\to\infty$, as $h\rightarrow 0$. {Our main theoretical result} accomplishes this for the tempered stable L\'{e}vy processes of Section \ref{Sec:Model}, {and its} proof is deferred to Appendix \ref{sec:ProofLemAsymb1eps}.

\begin{theorem}\label{lem:Asymb1eps}
Under Assumption \ref{assump:Funtq}, we have
\begin{align*}
\bE\big(b_{1,h}(\varepsilon)\big)\!=\!\sigma^{2}h\!-\!\frac{\sigma\varepsilon\sqrt{2h}}{\sqrt{\pi}}e^{-\varepsilon^{2}/(2\sigma^{2}h)}\!+\!\frac{C_{+}\!\!+\!C_{-}}{2-Y}h\varepsilon^{2-Y}\!+\!O\Big(he^{-\varepsilon^{2}/(2\sigma^{2}h)}\Big)\!+\!O\big(h\varepsilon^{2-Y/2}\big)\!+\!O\big(h^{2-Y/2}\big),
\end{align*}
as $h\rightarrow 0+$ and $\varepsilon=\varepsilon(h)\rightarrow {0^{+}}$, with $\varepsilon/\sqrt{h}\rightarrow\infty$.
\end{theorem}

The following result gives the asymptotic behavior of the optimal threshold $\varepsilon_{n}^{\star}$. Its proof is similar to that of \cite[Proposition 2]{FigueroaLopezMancini:2019} and is outline below for completeness and also to motivate some approximation methods proposed below.

\begin{corollary}\label{thm:AsymOptThold}
Under Assumption \ref{assump:Funtq}, the optimal threshold $\varepsilon^{\star}_{n}$ is such that
\begin{align}\label{eq:AsymOptThold}
\varepsilon^{\star}_{n}\sim\sqrt{(2-Y)\sigma^{2}h_{n}\ln\frac{1}{h_{n}}},\quad\text{as }\,n\rightarrow\infty.
\end{align}
{Furthermore, setting {$\overline{C}=(C_{+}+C_{-})/2$}, we have, {as $n\rightarrow\infty$,}}
\begin{align}\label{eq:AsymOptTholdbb}
{\varepsilon^{\star}_{n}\!=\!\sqrt{\sigma^{2}h_{n}\!\left[(2\!-\!Y)\ln\!\frac{1}{h_{n}}\!+\!(Y\!-\!1)\ln\ln\!\frac{1}{h_{n}}\!+\!(Y\!-\!1)\ln\!\big((2\!-\!Y)\sigma^{2}\big)\!+\!2\ln\!\bigg(\!\frac{(2\!-\!Y)\sigma}{\overline{C}\sqrt{2\pi}}\!\bigg)\!+\!o(1)\right]}.}\qquad
\end{align}
\end{corollary}

\begin{proof}
For simplicity, we take $T=1$ so that $h_{n}=1/n$. {With $\overline{C}=(C_{+}+C_{-})/2$ and using} the asymptotic behavior of $\bE(b_{1,h_{n}}(\varepsilon_{n}^{\star}))$ described in {Theorem} \ref{lem:Asymb1eps}, we can write \eqref{eq:Equeps} as
\begin{align*}
(\varepsilon_{n}^{\star})^{2}+2(n-1)\bigg(\sigma^{2}h_{n}-\frac{\sqrt{2}\sigma}{\sqrt{\pi}}\varepsilon_{n}^{\star}\sqrt{h_{n}}e^{-(\varepsilon_{n}^{*})^{2}/(2\sigma^{2}h_{n})}+\frac{{2\overline{C}}}{2-Y}h_{n}(\varepsilon_{n}^{\star})^{2-Y}+\text{h.o.t.}\bigg)\!-2nh_{n}\sigma^{2}=0,
\end{align*}
where h.o.t. means ``higher-order terms" {as $n\rightarrow\infty$}. In view of \eqref{eq:OptTholdgghn} and since $Y\in(1,2)$, we have
\begin{align*}
\frac{{2\overline{C}}}{2-Y}(\varepsilon_{n}^{\star})^{2-Y}-\frac{\sqrt{2}\,\sigma}{\sqrt{\pi}}\frac{\varepsilon_{n}^{\star}}{\sqrt{h_{n}}}\,e^{-(\varepsilon_{n}^{*})^{2}/(2\sigma^{2}h_{n})}+o\bigg(\frac{\varepsilon_{n}^{\star}}{\sqrt{h_{n}}}\,e^{-(\varepsilon_{n}^{*})^{2}/(2\sigma^{2}h_{n})}\bigg)+o\big((\varepsilon_{n}^{*})^{2-Y}\big)=0.
\end{align*}
Dividing by $\varepsilon_{n}^{*}$, rearranging the terms, and taking logarithms of both sides, we deduce that
\begin{align*}
(1-Y)\ln\varepsilon_{n}^{*}+o(1)= -\frac{(\varepsilon_{n}^{\star})^{2}}{2\sigma^{2}h_{n}}-\frac{1}{2}\ln h_{n}+\ln\bigg(\frac{\sqrt{2}\,\sigma(2-Y)}{{2\overline{C}}\sqrt{\pi}}\bigg)+o(1),
\end{align*}
{which can be written as
\begin{align}\label{eq:EquAsymOpteps}
\frac{(\varepsilon_{n}^{\star})^{2}}{\sigma^{2}h_{n}}+(1-Y)\ln\bigg(\frac{(\varepsilon_{n}^{\star})^{2}}{\sigma^{2}h_{n}}\bigg)+(1-Y)\ln\big(\sigma^{2}\big)+(2-Y)\ln h_{n}-2\ln\bigg(\frac{\sigma(2-Y)}{\overline{C}\sqrt{2\pi}}\bigg)=o(1).
\end{align}
Dividing by $(\varepsilon_{n}^{*})^{2}/(\sigma^{2}h_{n})$ and using \eqref{eq:OptTholdgghn}, we obtain the first result \eqref{eq:AsymOptThold}. For {the} second asymptotics, note that \eqref{eq:AsymOptThold} implies that
\begin{align*}
\ln\bigg(\frac{(\varepsilon_{n}^{\star})^{2}}{\sigma^{2}h_{n}}\bigg)=\ln\left((2-Y)\ln\frac{1}{h_{n}}\right)+o(1).
\end{align*}	
Finally, plugging the above in \eqref{eq:EquAsymOpteps} and solving for $\varepsilon_{n}^{\star}$ gives the desired asymptotics.}
\end{proof}

The proportionality constant $\sqrt{2-Y}$ of the previous result is intuitive and roughly tells us that the higher the jump activity is, the lower the optimal threshold has to be if we want to discard the higher noise represented by the jumps and to catch information about the Brownian component.

\section{Other Approximations and Illustration for a CGMY Model}\label{sec:NumIll}

In this section, we introduce other approximations to the optimal threshold derived from the formulas in Theorem \ref{lem:Asymb1eps} and the proof of Corollary \ref{thm:AsymOptThold}. We then illustrate their performance in the case of a L\'{e}vy process with a CGMY jump component $J$ (cf. \cite{CarrGemanMadanYor:2002}). The CGMY model is considered a prototypical jump process of infinite activity in finance. In the notation of the L\'{e}vy density \eqref{TmpStbLD}, a CGMY model is given by
\begin{align*}
q(x)=e^{-Mx}{\bf 1}_{(0,\infty)}(x)+e^{Gx}{\bf 1}_{(-\infty,0)}(x)\quad\text{and}\quad C_{+}=C_{-}=C.
\end{align*}
Thus, the conditions of Assumption \ref{assump:Funtq} are satisfied with $\alpha_{+}=-M$ and $\alpha_{-}=G$.
We adopt the parameter setting
\begin{align}\label{USPH0}
C=0.028,\quad G=2.318,\quad M=4.025,\quad Y=1.35.
\end{align}
These values are similar to those used in \cite{FigueroaLopezOlafsson:2019(2)}\footnote{\cite{FigueroaLopezOlafsson:2019(2)} considers the asymmetric case $\nu(dx)=C(x/|x|)\bar{q}(x)|x|^{-1-Y}\,dx$ with $C(1)=0.015$ and $C(-1)=0.041$. Here, we take $C=(C(1)+C(-1))/2$ in order to simplify the simulation of the model. Our values of $G$, $M$, and $Y$ are the same as in \cite{FigueroaLopezOlafsson:2019(2)}.}, who themselves took them from an empirical study in \cite{Kawai:2010}. We take $T=1$ year and $n=252(6.5)(60)$, which corresponds to a frequency of $1$ minute (assuming $252$ trading days and $6.5$ trading hours per day).

To compute $\bE(b_{1,h}(\varepsilon))$, we use Monte Carlo and the change of probability measure \eqref{eq:DenTranTildePP}. Concretely, under $\wt{\bP}$, we have the following representation:
\begin{align*}
\bE\big(b_{1,h}(\varepsilon)\big)&=\wt{\bE}\Big(e^{-U_{h}}\big(\sigma W_{h}+J_{h}\big)^{2}\,{\bf 1}_{\{|\sigma W_{h}+J_{h}|\leq\varepsilon\}}\Big)\\
&=\wt{\bE}\Big(e^{-MZ_{h}^{+}+GZ_{h}^{-}-\eta h}\big(\sigma W_{h}+Z_{h}^{+}+Z_{h}^{-}+\wt{\gamma}h\big)^{2}{\bf 1}_{\{|\sigma W_{h}+Z_{h}^{+}+Z_{h}^{-}+\wt{\gamma}h|\leq\varepsilon\}}\Big),
\end{align*}
where $Z_{h}^{+}$ and $-Z_{h}^{-}$ are independent one-sided $Y$-stable random variables with common scale, skewness, and location parameters given by $C|\Gamma(-Y)\cos(\pi Y/2)|h^{1/Y}$, $1$, and $0$, respectively. Such a distribution can be simulated efficiently\footnote{In our code, we use the R package {\tt stabledist} to generate them.}.

We consider two different approximations of the equation \eqref{eq:Equeps} defining the optimal threshold $\varepsilon^{\star}_{n}$. For the first approximation, we replace {$\bE(b_{1}(\varepsilon,h_{n}))$ (where $h_{n}=1/n$)} in \eqref{eq:Equeps} with its leading order terms as given by Theorem \ref{lem:Asymb1eps}, namely,
\begin{align}\label{eq:EquepsApprx1}
\varepsilon^{2}+2(n-1)\bigg({-\frac{\sqrt{2}\sigma}{\sqrt{\pi}}}\varepsilon\sqrt{h_{n}}e^{-\varepsilon^{2}/(2\sigma^{2}h_{n})}+\frac{2C}{2-Y}h_{n}\varepsilon^{2-Y}\bigg)-{2\sigma^{2}h_{n}}=0.
\end{align}
For the second approximation, we take a simplified version of \eqref{eq:AsymOptTholdbb}, only keeping those terms that are found to be {significant:}
\begin{align}\label{eq:EquepsApprx2}
\wt{\varepsilon}_{n}^{\,\star}:=\sqrt{(2-Y)\sigma^{2}h_{n}\ln\Big(\frac{1}{h_{n}}\Big)+2\sigma^{2}h_{n}\ln\bigg(\frac{(2-Y)\sigma}{C}\bigg)}.
\end{align}
{Interestingly, as $C\rightarrow 0$, we} have $\wt{\varepsilon}_{n}^{\,\star}\rightarrow\infty$, which makes sense. The approximation \eqref{eq:EquepsApprx2} says that if $C$ is small (relative to $\sigma$) then the threshold can be loosened up.

Figure \ref{Fig:ApproxNum1} shows the graphs of the left-hand expressions of \eqref{eq:Equeps} (solid blue) and the approximation \eqref{eq:EquepsApprx1} (dashed red) against $\varepsilon$ for three different values of $\sigma$: $0.1$, $0.2$, and $0.4$. The solid blue vertical line is the ``true" optimum threshold $\varepsilon=\varepsilon^{\star}_{n}$, the dotted brown vertical line shows $\varepsilon=\wt{\varepsilon}_{n}^{\,\star}$ with $\wt{\varepsilon}_{n}^{\,\star}$ given as in approximation (\ref{eq:EquepsApprx2}), and the dotted/dashed vertical green line is the approximation $\varepsilon=\varepsilon_{n}:=\sqrt{(2-Y)\sigma^{2}h_{n}\ln(1/h_{n})}$ derived in {\eqref{eq:AsymOptThold} of} Corollary \ref{thm:AsymOptThold}. We also show the vertical line passing at the root of \eqref{eq:EquepsApprx1} (vertical dashed red). It is evident that for the considered values of $Y$ and $\sigma$, the root of \eqref{eq:EquepsApprx1} and $\wt{\varepsilon}_{n}^{\,\star}$ are reasonably good approximations of $\varepsilon^{\star}_{n}$. However, we cannot say the same about $\varepsilon_{n}=\sqrt{(2-Y)\sigma^{2}h_{n}\ln(1/h_{n})}$, which is a good approximation of $\varepsilon^{\star}_{n}$ only for small values of $\sigma$ and, otherwise, it underestimates $\varepsilon^{\star}_{n}$.

\begin{figure}[h!]
\vspace{-0.5cm}
\begin{center}
\hspace{-0.5cm}
\includegraphics[width=6cm,height=6cm]{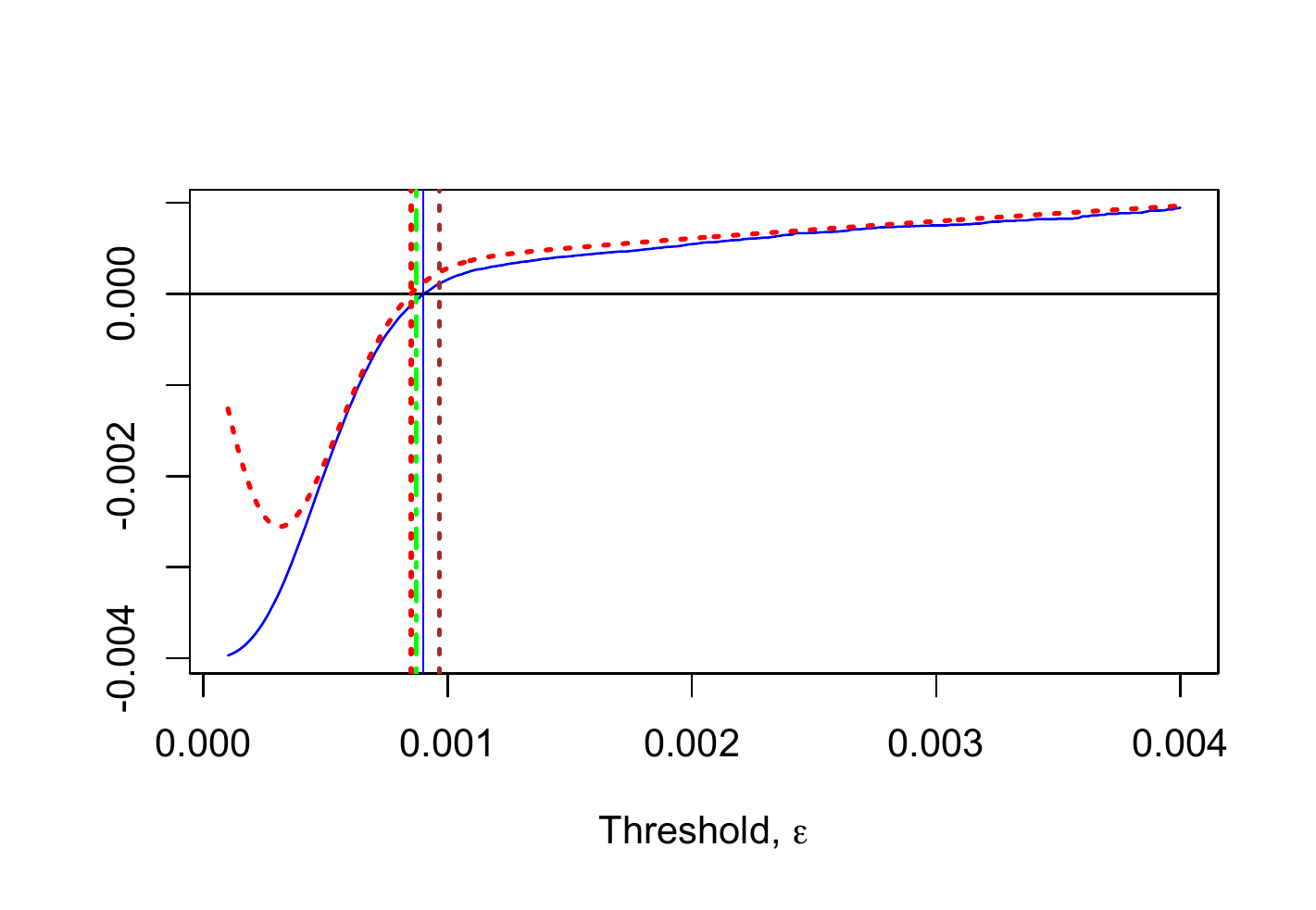}
\hspace{-0.8cm}
\includegraphics[width=6cm,height=6cm]{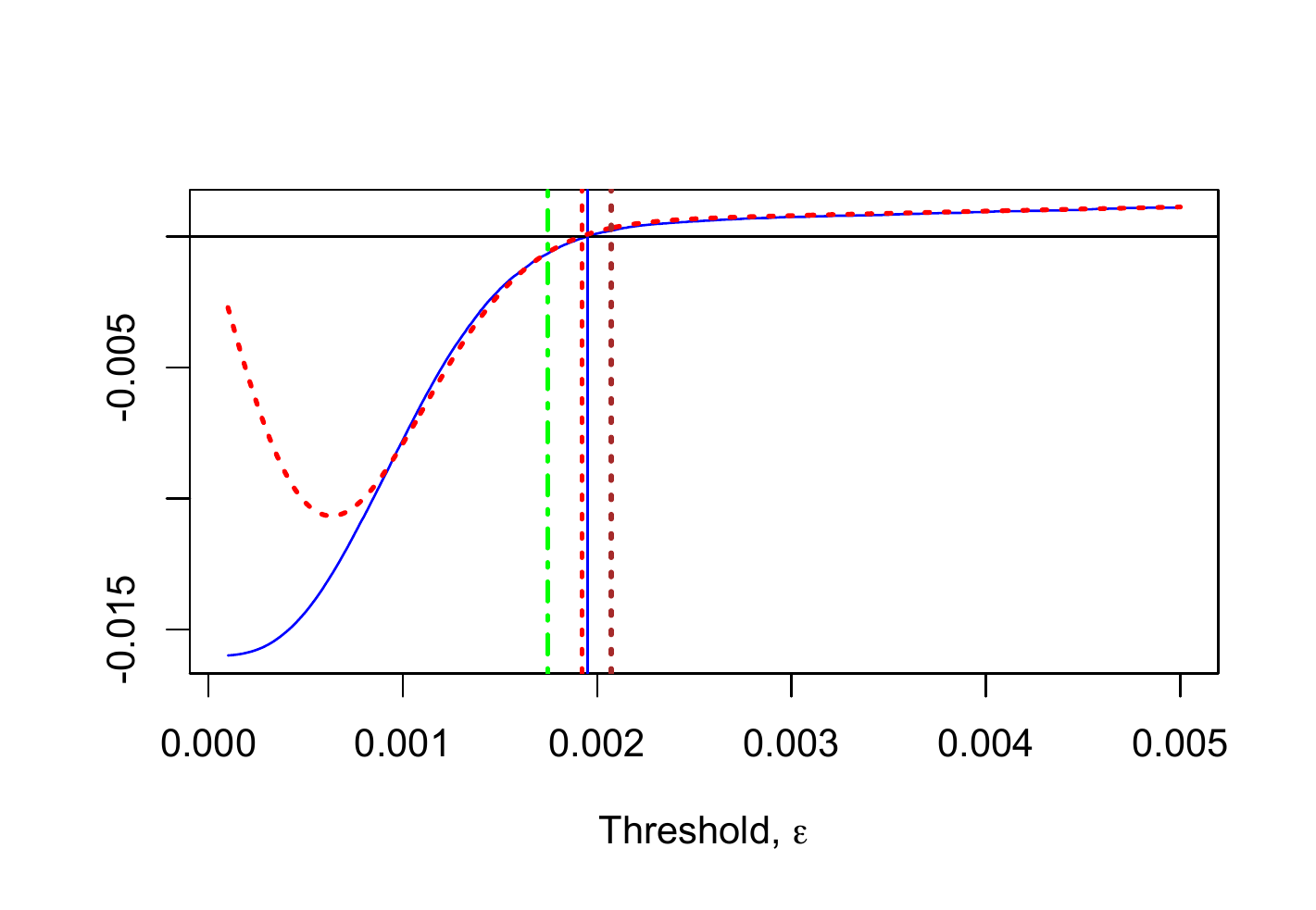}
\hspace{-0.8cm}
\includegraphics[width=6cm,height=6cm]{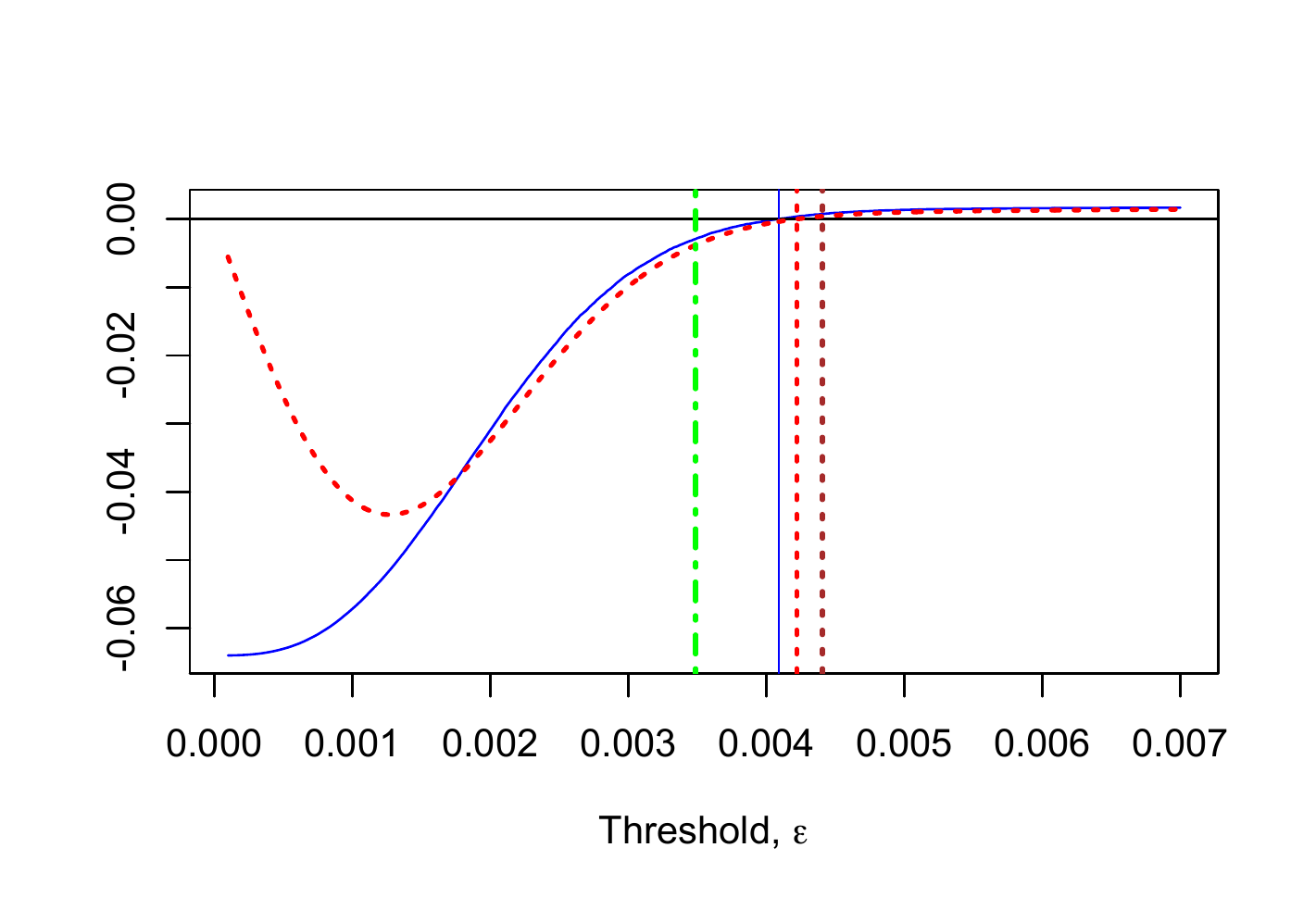}
\end{center}
\vspace{-0.5cm}
\caption{\small Graphs of the respective left-hand expressions of \protect\eqref{eq:Equeps} (solid blue) and \protect\eqref{eq:EquepsApprx1} (dashed red) against $\varepsilon$ for $\sigma=0.1$ (left panel), $\sigma=0.2$ (center panel), and $\sigma=0.4$ (right panel), respectively. We also show the vertical lines $\varepsilon=\varepsilon^{\star}_{n}$ (solid blue), $\varepsilon=\text{the root of \protect\eqref{eq:EquepsApprx1}}$ (dashed red), $\varepsilon=\wt{\varepsilon}^{\,\star}_{n}$ (dotted brown), and $\varepsilon=\sqrt{(2-Y)\sigma^{2}h_{n}\ln(1/h_{n})}$ (dotted/dashed green). The parameters for the CGMY model are set as $C=0.028$, $G=2.318$, $M=4.025$, and $Y=1.35$.}\label{Fig:ApproxNum1}
\end{figure}

Next, we consider the value of $Y=1.5$, while all the other CGMY parameter values remain unchanged. Figure \ref{Fig:ApproxNum2} below shows the graphs of the left-hand expressions of \eqref{eq:Equeps} (solid blue) and  \eqref{eq:EquepsApprx1} (dashed red), against $\varepsilon$ for three different values of $\sigma$: $0.1$, $0.2$, and $0.4$. The Equation \eqref{eq:EquepsApprx1} derived from Theorem \ref{lem:Asymb1eps} is a relatively accurate approximation of \eqref{eq:Equeps}, especially for larger values of $\sigma$. As before, the approximation \eqref{eq:AsymOptThold} established in Corollary \ref{thm:AsymOptThold} is accurate for small and medium values of $\sigma$ but not for larger values. The approximation \eqref{eq:EquepsApprx2} is reasonably accurate for all considered values of $\sigma$.

\begin{figure}[h!]
\vspace{-0.5cm}
\begin{center}
\hspace{-0.5cm}	
\includegraphics[width=6.0cm,height=6cm]{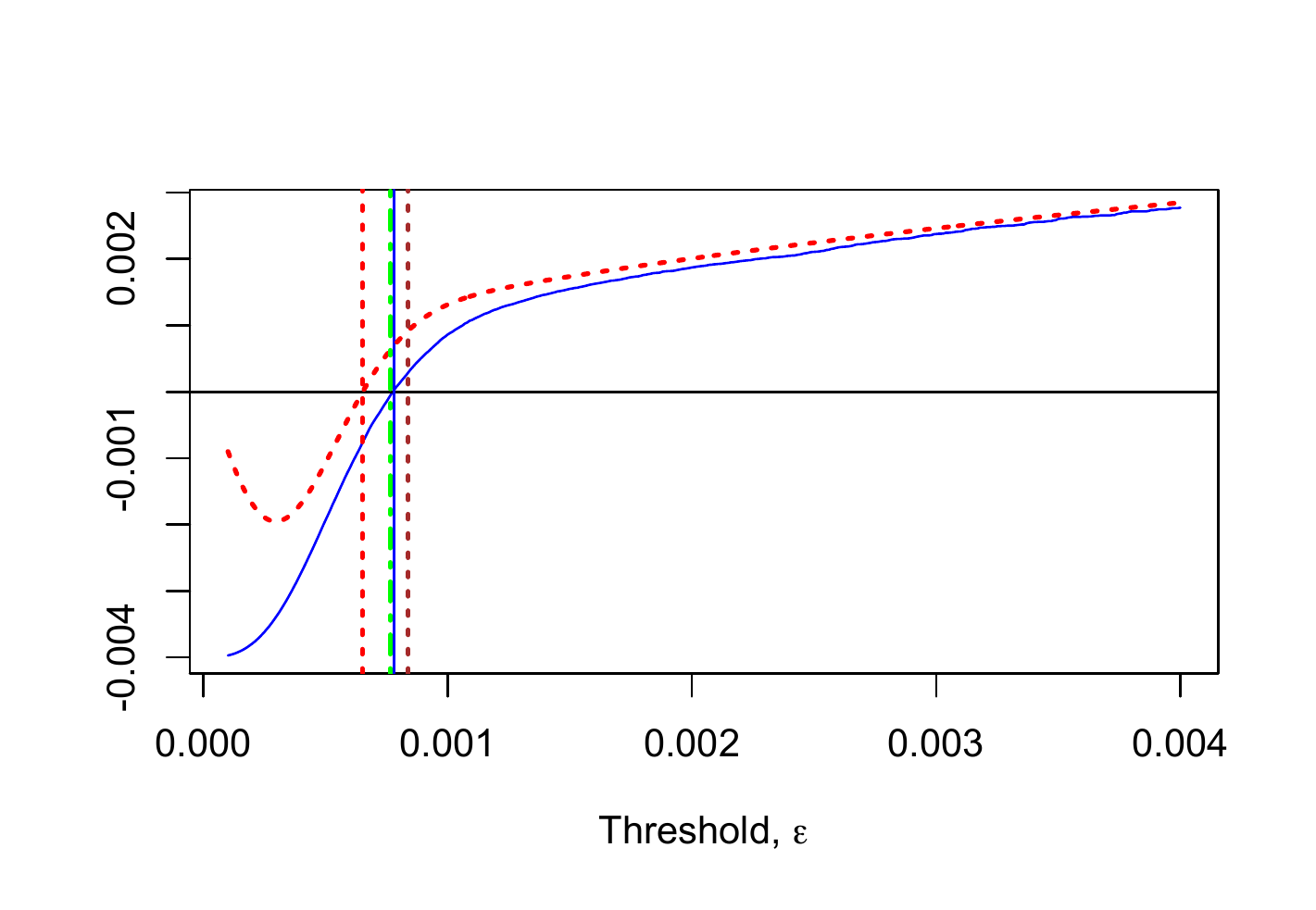}
\hspace{-0.8cm}
\includegraphics[width=6.0cm,height=6cm]{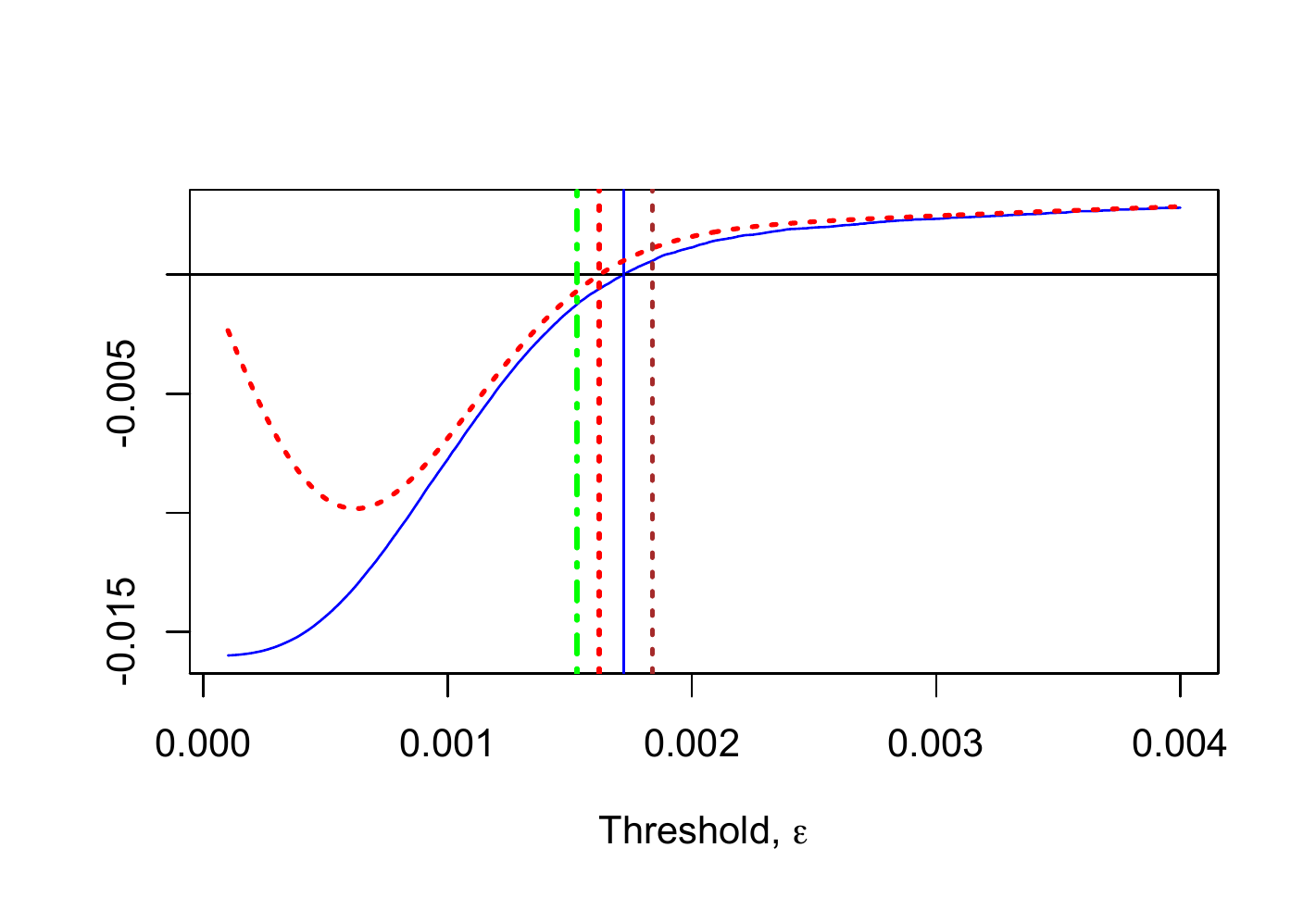}
\hspace{-0.8cm}
\includegraphics[width=6.0cm,height=6cm]{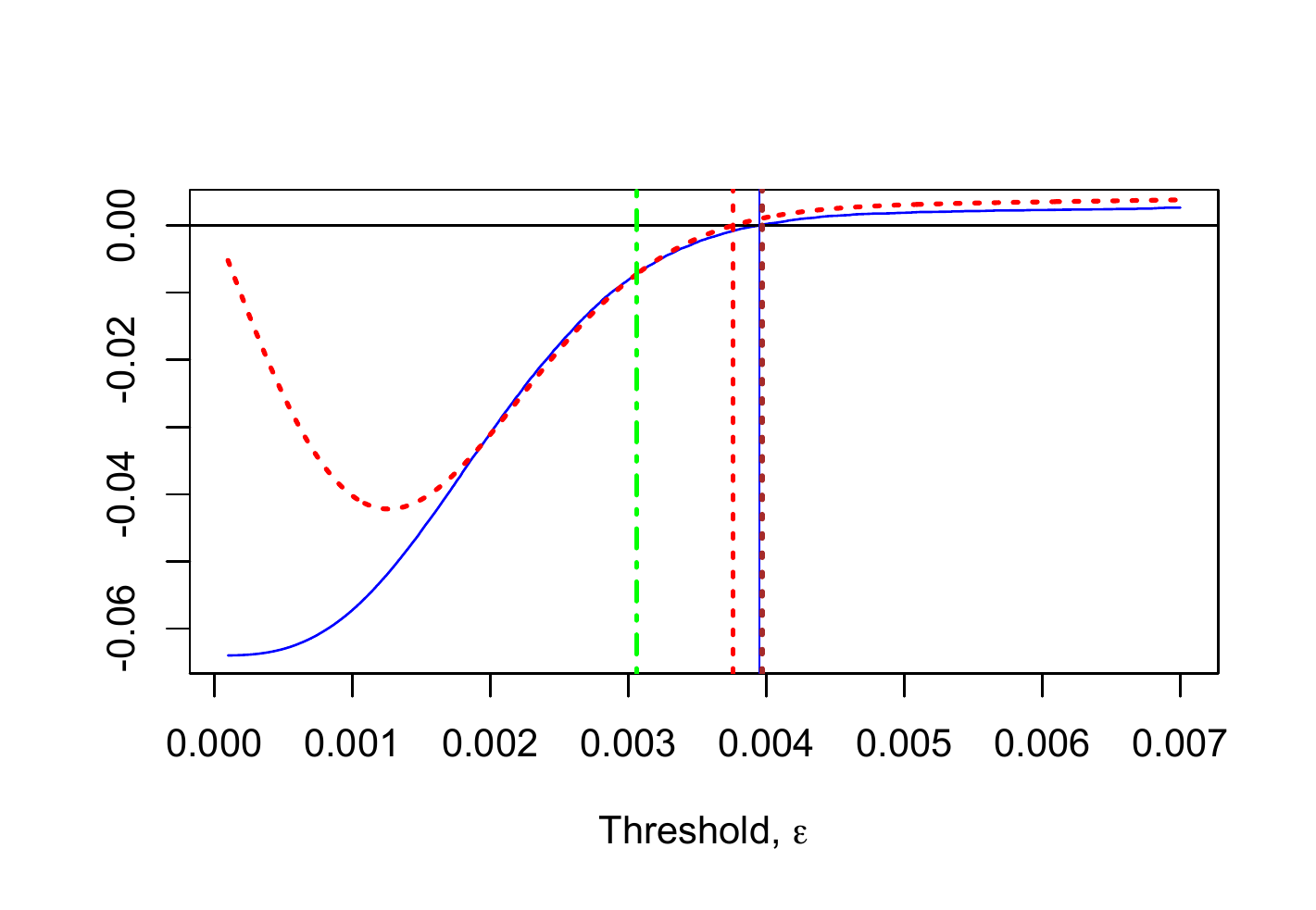}
\end{center}
\vspace{-0.5cm}
\caption{\small Graphs of the respective left-hand expressions of \protect\eqref{eq:Equeps} (solid blue) and \protect\eqref{eq:EquepsApprx1} (dashed red) against $\varepsilon$ for $\sigma=0.1$ (left panel), $\sigma=0.2$ (center panel), and $\sigma=0.4$ (right panel), respectively. We also show the vertical lines $\varepsilon=\varepsilon^{\star}_{n}$ (solid blue), $\varepsilon=\text{the root of \protect\eqref{eq:EquepsApprx1}}$ (dashed red), $\varepsilon=\wt{\varepsilon}^{\,\star}_{n}$ (dotted brown), and $\varepsilon=\sqrt{(2-Y)\sigma^{2}h_{n}\ln(1/h_{n})}$ (dotted/dashed green). The parameters for the CGMY model are set as $C=0.028$, $G=2.318$, $M=4.025$, and $Y=1.5$.}\label{Fig:ApproxNum2}
\end{figure}

Finally, we consider the value of $Y=1.7$. All the other CGMY parameter values remain the same. The approximations are shown in Figure \ref{Fig:ApproxNum3}. We deduce that for such a large value of $Y$, the approximation \eqref{eq:EquepsApprx1} derived from {Theorem} \ref{lem:Asymb1eps} is not accurate anymore, though it improves as $\sigma$ gets larger. On the other hand, {the other suggested approximation \eqref{eq:EquepsApprx2} is still relatively} accurate to approximate the optimal threshold {$\varepsilon_{n}^{\star}$} (the root of \eqref{eq:Equeps}). We again have that for small {and medium} values of $\sigma$, the approximation \eqref{eq:AsymOptThold} is good, which is not the case for large values of $\sigma$.

\begin{figure}
\vspace{-0.5cm}
\begin{center}
\hspace{-0.5cm}	
\includegraphics[width=6.0cm,height=6cm]{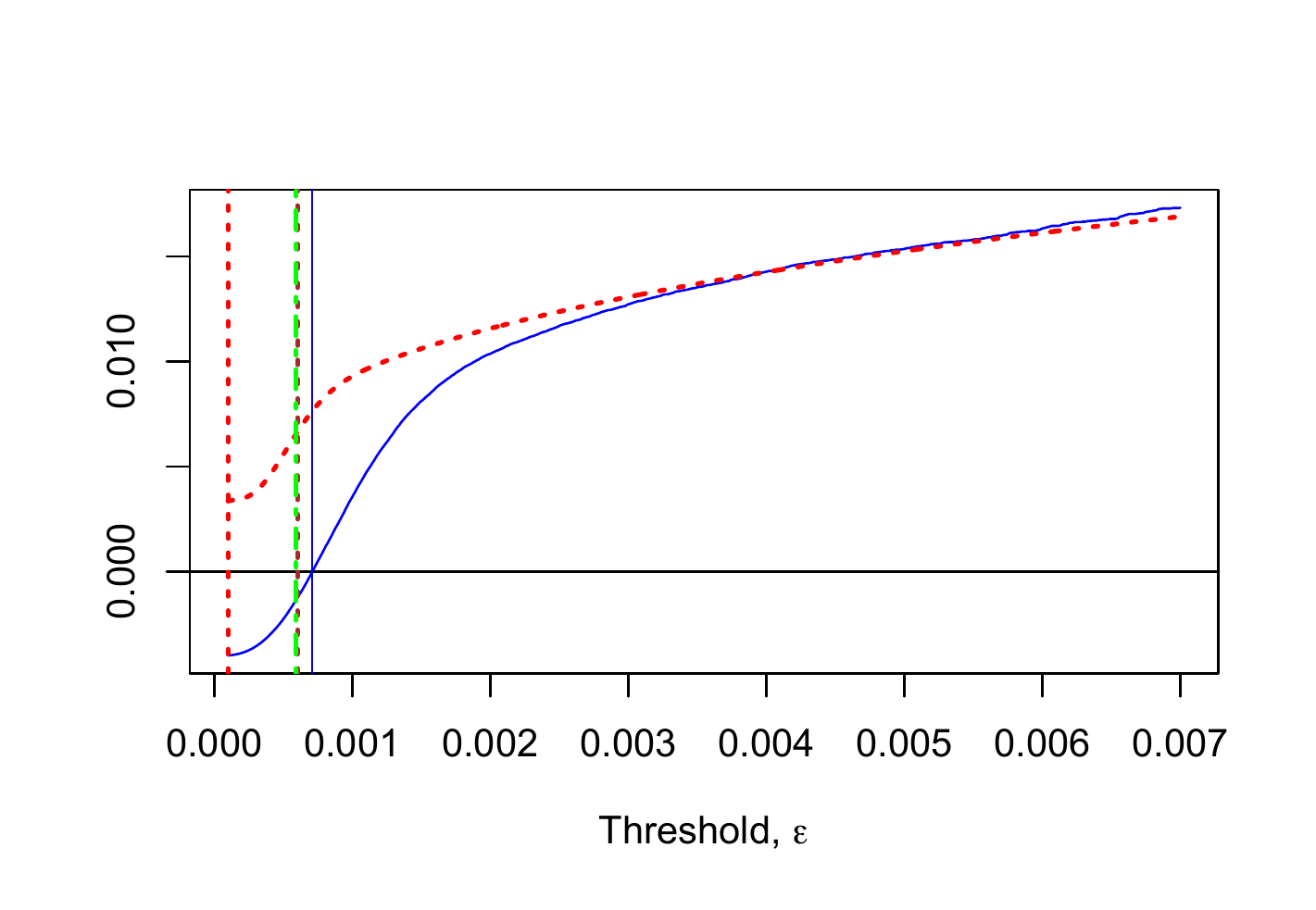}
\hspace{-0.8cm}
\includegraphics[width=6.0cm,height=6cm]{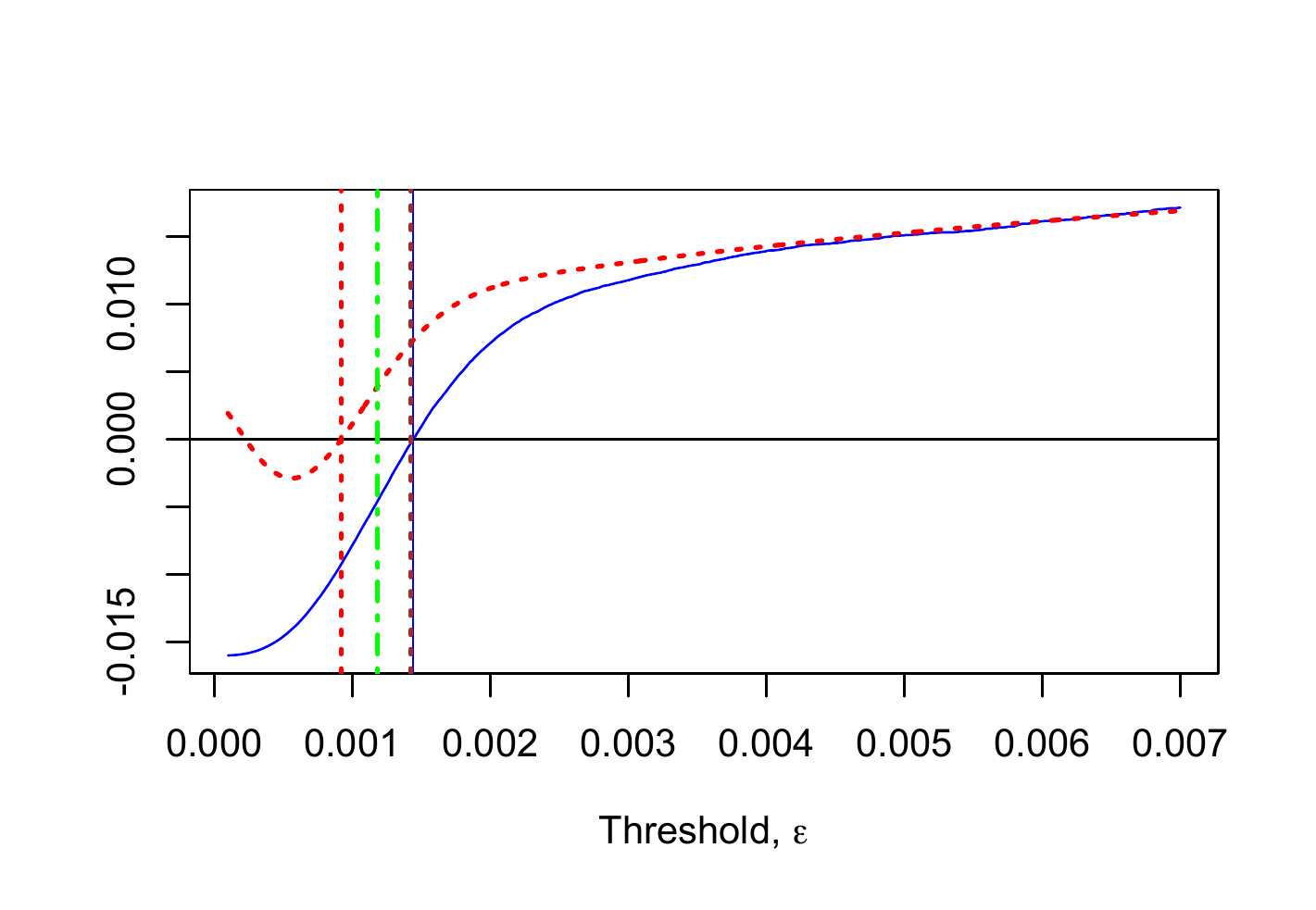}
\hspace{-0.8cm}
\includegraphics[width=6.0cm,height=6cm]{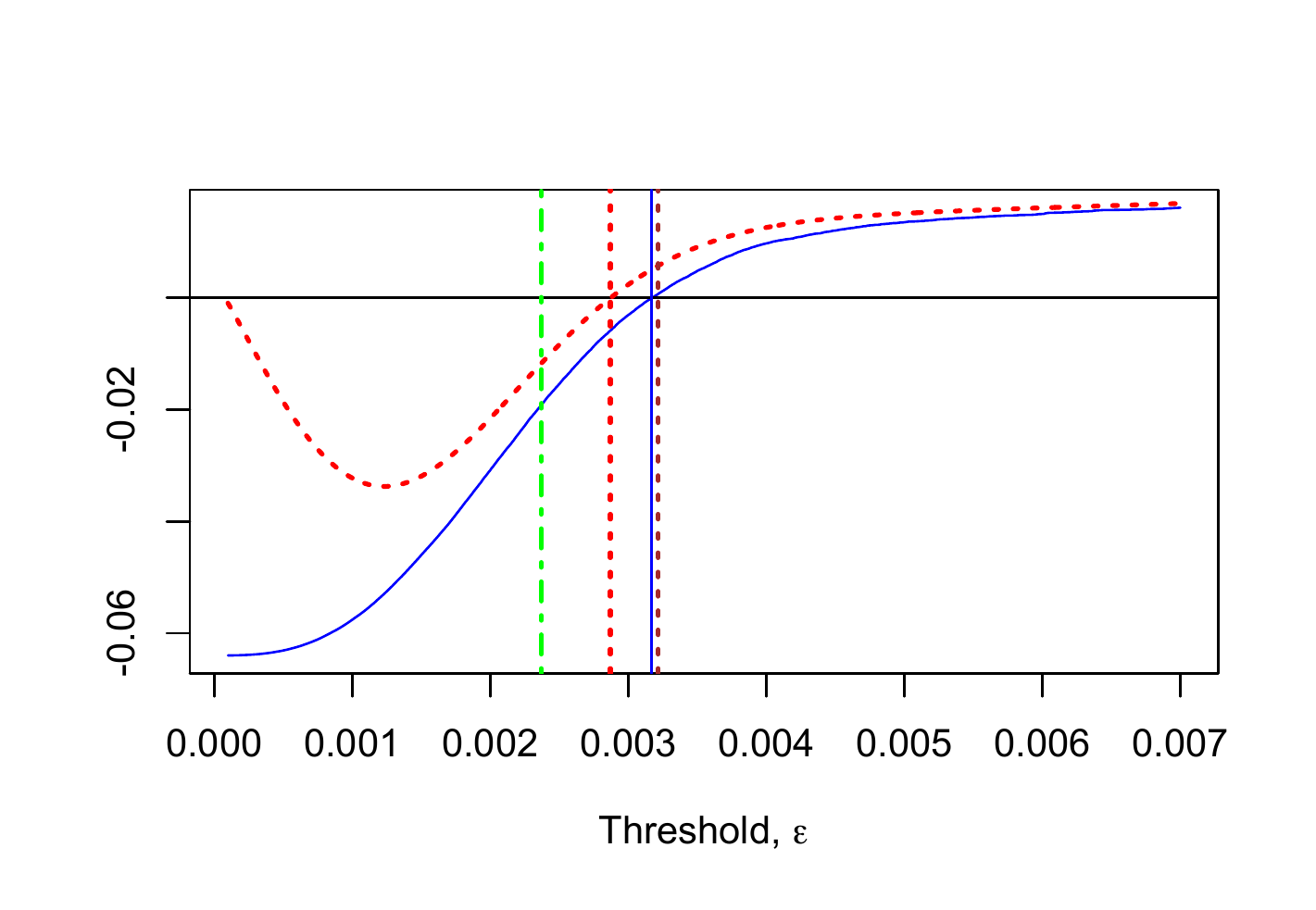}
\end{center}
\vspace{-0.6cm}
\caption{\small Graphs of the respective left-hand expressions of \protect\eqref{eq:Equeps} (solid blue) and \protect\eqref{eq:EquepsApprx1} (dashed red) against $\varepsilon$ for $\sigma=0.1$ (left panel), $\sigma=0.2$ (center panel), and $\sigma=0.4$ (right panel), respectively. We also show the vertical lines $\varepsilon=\varepsilon^{\star}_{n}$ (solid blue), $\varepsilon=\text{the root of \protect\eqref{eq:EquepsApprx1}}$ (dashed red), $\varepsilon=\wt{\varepsilon}^{\,\star}_{n}$ (dotted brown), and $\varepsilon=\sqrt{(2-Y)\sigma^{2}h_{n}\ln(1/h_{n})}$ (dotted/dashed green). The parameters of the CGMY model are set as $C=0.028$, $G=2.318$, $M=4.025$, and $Y=1.7$.}\label{Fig:ApproxNum3}
\end{figure}

To summarize, while for values of $Y\leq 1.5$, the approximation (\ref{eq:EquepsApprx1}) may be the most accurate, this is not the case anymore for larger values of $Y$. On the other hand, the approximation \eqref{eq:EquepsApprx2} is reasonably good for a large range of values of $Y$. Due to this reason, in our simulations of Section \ref{Sect:NewEstsigma}, we use {\eqref{eq:EquepsApprx2} to assess the finite sample performance of the proposed estimation method below.

\section{A New Method To Estimate The Volatility}\label{Sect:NewEstsigma}

In this section, we propose a new method for estimating the volatility $\sigma^{2}$ and other parameters of a tempered-stable L\'{e}vy process using the TRQV \eqref{eq:TRV} and the approximations of the optimal threshold derived in Section \ref{sec:MainThm}. Then, we illustrate the method in the case of a CGMY L\'{e}vy process. Finally, using a localization technique, we adapt our method to estimate the integrated variance under a Heston stochastic volatility model with CGMY jumps and compare it to the method proposed by \cite{JacodTodorov:2014}, which is known to be efficient when $Y\leq 1.5$.

\subsection{Estimation of Stable-Like L\'{e}vy Measures}\label{Sect51}

As shown by \eqref{eq:Equeps} and the asymptotic expansion of Theorem \ref{lem:Asymb1eps}, the optimal threshold $\varepsilon^{\star}_{n}$ depends on the volatility, and vice versa. It is then natural to consider an iterative method to estimate $\varepsilon^{\star}_{n}$. But before this, we need to estimate $C_{\pm}$ and $Y$. Several methods have been proposed in the literature for this purpose (see, e.g., \cite{AitSahaliaJacod:2009}, \cite{Bull:2016}, and \cite{Reiss:2013}). Mies \cite{Mies:2019} recently proposed an efficient method using the method of moments. In this part, we adapt and modify this method and apply it in combination with the approximations of Theorem \ref{lem:Asymb1eps} to estimate the optimal threshold $\varepsilon^{\star}_{n}$ of the TRQV and subsequently the other parameters $\sigma$, $Y$, and $C_{\pm}$.

Consider a L\'{e}vy process $X:=(X_{t})_{t\in\bR_{+}}$ with characteristic triplet $(\mu,\sigma^{2},\nu)$. The approach of \cite{Mies:2019} builds on the assumption that $\nu$ can be well approximated by the superposition of stable L\'{e}vy measures in the sense that
\begin{align}\label{eq:assumpm}
\big|\nu\big([x,\infty)\big)-\wt{\nu}\big([x,\infty)\big)\big|&\leq L|x|^{-\rho},\quad x\in(0,1],\\
\label{eq:AssumpMinusMies} \big|\nu\big((-\infty,x]\big)-\wt{\nu}\big((-\infty, x]\big)\big|&\leq L|x|^{-\rho},\quad x\in[-1,0),
\end{align}
for some $L,\rho\in(0,\infty)$, where $\wt{\nu}$ is given by
\begin{align}\label{eq:approxm}
\wt{\nu}(dz)=\sum_{m=1}^{{N}}\frac{\alpha_{m}}{|z|^{1+\alpha_{m}}}\big(r_{m}^{+}{\bf 1}_{\{x>0\}}+r_{m}^{-}{\bf 1}_{\{x<0\}}\big)\,dz,
\end{align}
for some $N\in\bN$, $\boldsymbol{\alpha}=(\alpha_{1},\ldots,\alpha_{N})\in(0,2)^{N}$, and $\boldsymbol{r}=(r_{1}^{+},r_{1}^{-},\ldots,r_{N}^{+},r_{N}^{-})\in\bR_{+}^{2N}$ such that
\begin{align*}
\alpha_{1}>\alpha_{2}>\cdots>\alpha_{N}>\frac{\alpha_{0}}{2},\quad\alpha_{N}>\rho,\quad r_{m}^{+}+r_{m}^{-}>0,\quad m=1,\ldots,N,
\end{align*}
for some $\alpha_{0}\in(0,\infty)$. We want to estimate $\boldsymbol{\theta}:=(\sigma^{2},\boldsymbol{r},\boldsymbol{\alpha})$ given $n$ observations, $X_{t_{1}},X_{t_{2}},\ldots,X_{t_{n}}$, of the process $X$ at known times $0=t_{0}<t_{1}<\cdots<t_{n}=T$. As before, we assume the sampling times are evenly spaced and we done the time step between observations as $h_{n}:=T/n$. Conditions \eqref{eq:assumpm}, \eqref{eq:AssumpMinusMies}, and \eqref{eq:approxm} essentially say that we can approximate $X$ by a fully specified L\'{e}vy process $\wt{Z}:=(\wt{Z}_{t})_{t\in\bR_{+}}$ with characteristic triplet $(0,\sigma^{2},\wt{\nu})$ and, hence, with the decomposition
\begin{align*}
\wt{Z}_{t}=\sigma W_{t}+\sum_{m=1}^{N}S_{t}^{m},\quad t\in\bR_{+},
\end{align*}
where $W:=(W_{t})_{t\in\bR_{+}}$ is a standard Brownian motion and $S^{m}:=(S_{t}^{m})_{t\in\bR_{+}}$, $m=1,\ldots,N$, are independent $\alpha_{m}$-stable processes, independent of $W$, each with L\'{e}vy density $\alpha_{m}|z|^{-1-\alpha_{m}}(r_{m}^{+}{\bf 1}_{\{x>0\}}+r_{m}^{-}{\bf 1}_{\{x<0\}})$, respectively.

Mies \cite{Mies:2019} proposed to estimate the parameters, $\boldsymbol{\theta}=(\sigma^{2},\boldsymbol{r},\boldsymbol{\alpha})$, of the approximating process $\wt{Z}$ using the method of moments. We now proceed to briefly review her method. The first step is to choose $3N+1$ moment functions $\boldsymbol{f}=(f_{1},\ldots,f_{3N+1})^{\text{T}}$, one for each parameters of $\wt{Z}$, and a suitable scaling factor $u_{n}\propto 1/\sqrt{h_{n}\ln(1/h_{n})}$, where ``$\propto$" hereafter means ``proportional to". Next, define the MME $\wh{\boldsymbol{\theta}}_{n}$ to be a solution of the following equation
\begin{align}\label{eq:EquMME}
\boldsymbol{F}_{n}(\boldsymbol{\theta}):=\frac{1}{n}\sum_{i=1}^{n}\boldsymbol{f}\big(u_{n}\Delta_{i}^{n}X\big)-\bE_{\boldsymbol{\theta}}\Big(\boldsymbol{f}\big(u_{n}\wt{Z}_{h_{n}}\big)\Big)=\boldsymbol{0},
\end{align}
where $\boldsymbol{0}=(0,\ldots,0)^{\text{T}}\in\bR^{3N+1}$ and $\bE_{\boldsymbol{\theta}}(\boldsymbol{f}(u_{n}\wt{Z}_{h_{n}}))$ denotes the expectation such that $\wt{Z}_{h_{n}}$ is determined by the parameter vector $\boldsymbol{\theta}$. Since $\wt{Z}_{h_{n}}$ is fully specified and, thus, its characteristic function is available, $\bE_{\boldsymbol{\theta}}(\boldsymbol{f}(u_{n}\wt{Z}_{h_{n}}))$ can be computed by, e.g., Fourier methods.

Our idea is to combine a version of Mies' method with our results in Section \ref{sec:MainThm} to improve our estimation of $\sigma^{2}$ and $Y$. Concretely, we propose to first find the roots of $\boldsymbol{F}_{n}(\boldsymbol{\theta})$ and plug them into a suitable approximation of Equation \eqref{eq:Equeps} to obtain an estimate of the optimal threshold $\varepsilon^{\star}_{n}$. This can in turn be used to estimate the volatility via thresholding. To solve the $3N+1$ equations \eqref{eq:EquMME}, we propose to solve an optimization problem with objective function of the form
\begin{align}\label{eq:ObjFunt}
V_{n}(\boldsymbol{\theta};\boldsymbol{f}):=\boldsymbol{F}_{n}^{\text{T}}(\boldsymbol{\theta})\mathsf{\Lambda}_{n}^{-1}(\boldsymbol{\theta})\mathsf{\Lambda}_{n}^{-1}(\boldsymbol{\theta})\boldsymbol{F}_{n}(\boldsymbol{\theta}),
\end{align}
where
\begin{align*}
\mathsf{\Lambda}_{n}(\boldsymbol{\theta}):=\text{diag}\big(h_{n}u_{n}^{2},h_{n}u_{n}^{\alpha_{1}},h_{n}u_{n}^{\alpha_{1}},h_{n}u_{n}^{\alpha_{1}},\ldots,h_{n}u_{n}^{\alpha_{N}},h_{n}u_{n}^{\alpha_{N}},h_{n}u_{n}^{\alpha_{N}}\big).
\end{align*}
This particular weights are motivated by the scaling of the Central Limit Theorem for $\boldsymbol{F}_{n}(\boldsymbol{\theta})$ established in \cite[Lemma 5.4]{Mies:2019}.

For simplicity, suppose we only want to estimate $\alpha=\alpha_{1}$, $r_{1}^{\pm}$, and $\sigma^{2}$ (the method can easily be adapted to estimate more parameters of $\wt{\nu}$). We then {\Blue propose} the following {procedure}:

\begin{enumerate}
\item Start with some initial values $\boldsymbol{\theta}_{0}:=(\sigma_{0}^{2},\boldsymbol{r}_{0},\boldsymbol{\alpha}_{0})$ and a suitable scaling factor $u_{n}$ (to be specified later on).
\item Find the roots of $\boldsymbol{F}_{n}(\boldsymbol{\theta})$, which we call $\wh{\boldsymbol{\theta}}_{n,1}:=(\wh{\sigma}_{n,1}^{2},\wh{\boldsymbol{r}}_{n,1},\wh{\boldsymbol{\alpha}}_{n,1})$, by minimizing the objective function $V_{n}(\boldsymbol{\theta};\boldsymbol{f})$ in \eqref{eq:ObjFunt}.
\item \label{item:GenIterStep3}
    Using $\wh{\boldsymbol{\theta}}_{n,1}$, we solve a suitable approximation of \eqref{eq:Equeps} (e.g., \eqref{eq:EquepsApprx1} or \eqref{eq:EquepsApprx2}) to get an estimation of the optimal threshold $\varepsilon^{\star}_{n}$, denoted by $\wh{\varepsilon}_{n,1}$. This estimate is then used to compute an estimate of $\sigma^{2}$ as
    \begin{align*}
    \wh{\sigma}^{2}_{n,2}:=\frac{1}{T}\sum_{i=1}^{n}\big(\Delta_{i}^{n}X\big)^{2}{\bf 1}_{\{|\Delta_{i}^{n}X|\leq\wh{\varepsilon}_{n,1}\}}.
    \end{align*}
\item \label{item:GenIterStep4}
    Fix $\wh{\sigma}^{2}_{n,2}$ and use $3N$ moment functions $\boldsymbol{g}:=(g_{1},\ldots,g_{3N})^{\text{T}}$ to find
    $(\wh{\boldsymbol{r}}_{n,2},\wh{\boldsymbol{\alpha}}_{n,2})$ by solving
    \begin{align}\label{eq:FunctGn}
    \boldsymbol{G}_{n}\big(\boldsymbol{r},\boldsymbol{\alpha};\wh{\sigma}^{2}_{n,2}\big):=\frac{1}{n}\sum_{i=1}^{n}\boldsymbol{g}\big(u_{n}\Delta_{i}^{n}X\big)-\bE_{(\wh{\sigma}^{2}_{n,2},\boldsymbol{r},\boldsymbol{\alpha})}\Big(\boldsymbol{g}\big(u_{n}\wt{Z}_{h_{n}}\big)\Big)=\boldsymbol{0},
    \end{align}
    or minimizing
    \begin{align}\label{eq:ObjFuntGn}
    V_{n}(\boldsymbol{r},\boldsymbol{\alpha};\wh{\sigma}^{2}_{n,2},\boldsymbol{g}):=\boldsymbol{G}_{n}^{\text{T}}\big(\boldsymbol{r},\boldsymbol{\alpha};\wh{\sigma}^{2}_{n,2}\big)\boldsymbol{G}_{n}\big(\boldsymbol{r},\boldsymbol{\alpha};\wh{\sigma}^{2}_{n,2}\big).
    \end{align}
\item \label{item:GenIterStep5}
    Using $(\wh{\sigma}^{2}_{n,2},\wh{\boldsymbol{r}}_{n,2},\wh{\boldsymbol{\alpha}}_{n,2})$ and solving the same approximating equation as in Step \ref{item:GenIterStep3}, we obtain a new estimation of $\varepsilon^{\star}_{n}$, denoted by $\wh{\varepsilon}_{n}^{\,\star}$, and update $\wh{\sigma}^{2}_{n,2}$ by
    \begin{align*}
    \big(\wh{\sigma}^{\star}_{n}\big)^{2}:=\frac{1}{T}\sum_{i=1}^{n}\big(\Delta_{i}^{n}X\big)^{2}{\bf 1}_{\{|\Delta_{i}^{n}X|\leq\wh{\varepsilon}_{n}^{\,\star}\}}.
    \end{align*}
\end{enumerate}

\begin{remark}
We could stop right after Step \ref{item:GenIterStep3} and make $\wh{\sigma}^{2}_{n,2}$ our final estimate of the volatility. However, our simulation results show that Step \ref{item:GenIterStep4} significantly improves our estimates of $(\boldsymbol{r},\boldsymbol{\alpha})$ when $Y\leq 1.5$. We could also follow the Steps 1$-$5 and repeat Steps \ref{item:GenIterStep4} and \ref{item:GenIterStep5} iteratively by letting $\wh{\sigma}^{2}_{n,2}=(\wh{\sigma}^{\star}_{n})^{2}$ until the sequence of estimates $(\wh{\sigma}^{\star}_{n})^{2}$ stabilizes. This approach, however, tends to increase the sample error of the estimators. {\Blue In the steps 2 and 4 above, there is no guarantee that the roots therein exist in a finite sample setting. This is another reason to use the minimum (or local minimum) of $V_{n}(\boldsymbol{\theta};\boldsymbol{f})$ and $ \boldsymbol{G}_{n}\big(\boldsymbol{r},\boldsymbol{\alpha};\wh{\sigma}^{2}_{n,2}\big)$ instead of solving \eqref{eq:EquMME} and \eqref{eq:FunctGn}, respectively}.
\end{remark}

\subsection{Estimation of a L\'{e}vy Process with CGMY Jump Component}\label{constCGMY}

In this subsection, we apply the method introduced in the previous subsection to the case of a CGMY jump component and compare it to the estimators of Mies \cite{Mies:2019} and Jacod and Todorov \cite{JacodTodorov:2014}. Specifically, we work with simulated data from the model \eqref{eq:MnMdlX} where $J$, is a pure-jump CGMY L\'{e}vy process, independent of the Brownian motion $W$, with L\'{e}vy measure
\begin{align*}
\nu_{CGMY}(dx):=\nu_{CGMY}(x)dx:=\frac{C}{|x|^{1+Y}}\Big(e^{-Mx}{\bf 1}_{\{x>0\}}+e^{Gx}{\bf 1}_{\{x<0\}}\Big)dx.
\end{align*}
We use the same values of $C$, $G$, and $M$ as in \eqref{USPH0}, but with different values of $\sigma$ and $Y$. We consider observations of a $5$ minutes frequency over a one-year ($252$ days) time horizon with a trading time of $6.5$ hours per day (so that {\Blue $n=252\times 6.5\times 12=19656$}). It should be clarified that we are indeed in the same setting as that of Subsection \ref{Sect51} since, as $x\rightarrow 0$, $\nu_{CGMY}(x)=C|x|^{-1-Y}+O(|x|^{-Y})$. This suggests us to take $N=1$ in \eqref{eq:approxm} and to use a $Y$-stable process to approximate the CGMY process because only the parameters $\sigma^{2}$, $C$, and $Y$ are of primary interest. Then Assumptions \eqref{eq:assumpm}$-$\eqref{eq:approxm} are satisfied with $\rho=Y$ and $\wt{Z}_{t}=\sigma W_{t}+S_{t}$, $t\in\bR_{+}$, where $(S_{t})_{t\in\bR_{+}}$ is a $Y$-stable process with L\'{e}vy measure $\wt{\nu}(dx):=C|x|^{-1-Y}dx$. The parameters of the approximating model are $\boldsymbol{\theta}=(\sigma^{2},C,Y)$.

Next, we choose the $3$ moment functions $\boldsymbol{f}=(f_{1},f_{2},f_{3})^{\text{T}}$ as
\begin{align}\label{eq:3MomFunctsCGMY}
f_{1}(x):=e^{-|x|},\quad f_{2}(x):=e^{-\sqrt{|x|}},\quad f_{3}(x):=x^{2}{\bf 1}_{\{|x|<1\}},\quad x\in\bR,
\end{align}
and a suitable scaling factor $u_{n}$ to be specified below. These functions are simpler than the ones proposed in \cite[Section 4]{Mies:2019} and were chosen because of their superior performance. Even though the moment functions \eqref{eq:3MomFunctsCGMY} do not meet the strict constraints imposed in \cite{Mies:2019} (see Assumptions (F1)$-$(F2) therein), we believe that most of the assumptions therein are not needed for the validity of the asymptotic theory in \cite{Mies:2019}. This will be investigated in a future work together with an objective and systematic method to calibrate the moment functions.

To determine a suitable scaling factor $u_{n}$, we will connect it to the threshold parameter $\varepsilon$ of the TRQV estimator \eqref{eq:TRV}. The key observation is to analyze the moment equation corresponding to the function $f_{3}$, namely,
\begin{align*}
\frac{1}{n}\sum_{i=1}^{n}f_{3}\big(u_{n}\Delta_{i}^{n}X\big)-\bE_{\boldsymbol{\theta}}\Big(f_{3}\big(u_{n}\wt{Z}_{h_{n}}\big)\Big)=0,
\end{align*}
which, after some trivial simplifications, can be written as
\begin{align*}
\frac{1}{n}\sum_{i=1}^{n}\big(\Delta_{i}^{n}X\big)^{2}{\bf 1}_{\{|\Delta_{i}^{n}X|\leq 1/u_{n}\}}-\bE_{\boldsymbol{\theta}}\Big(\wt{Z}_{h_{n}}^{2}{\bf 1}_{\{|\wt{Z}_{h_{n}}|\leq 1/u_{n}\}}\Big)=0.
\end{align*}
This suggests that $1/u_{n}$ has a similar role to that of the threshold $\varepsilon$ in the TRQV estimator; namely, the choice of $u_{n}$ should ensure that $\sum_{i=1}^{n}f_{3}(u_{n}\Delta_{i}^{n}X)$ is dominated by the Brownian component or, equivalently, to eliminate the increments in which the jump component $J$ of $X$ dominates the Brownian component. Hence, in what follows, we will fix $u_{n}$ as $1/\varepsilon_{n}$, where $\varepsilon_{n}=\sqrt{2{\sigma}_{0}^{2}h_{n}\ln(1/h_{n})}$ and $\sigma_{0}^{2}$ is a suitable initial estimate of $\sigma^{2}$. We consider the following initial values for $\sigma^{2}$:
\begin{align}\label{eq:IniSqusigma}
\wh{\sigma}_{n,01}^{2}\!\!:=\!\frac{1}{T}\!\sum_{i=1}^{n}\!\big(\Delta_{i}^{n}X\big)^{2}{\bf 1}_{\{|\Delta_{i}^{n}X|\leq\sqrt{h_{n}\!\ln(1/h_{n})}\}},\quad\wh{\sigma}_{n,02}^{2}\!:=\!\frac{1}{T}\!\sum_{i=1}^{n}\!\big(\Delta_{i}^{n}X\big)^{2}{\bf 1}_{\{|\Delta_{i}^{n}X|\leq\sqrt{2\wh{\sigma}_{n,00}^{2}h_{n}\!\ln(1/h_{n})}\}},\quad\,\,\,
\end{align}
where
\begin{align*}
\wh{\sigma}_{n,00}^{2}:=\frac{1}{T}\sum_{i=1}^{n}\big(\Delta_{i}^{n}X\big)^{2}{\bf 1}_{\{|\Delta_{i}^{n}X|\leq\sqrt{2\wh{\sigma}_{n,\text{RV}}^{2}h_{n}\ln(1/h_{n})}\}},\quad\wh{\sigma}_{n,\text{RV}}^{2}:=\frac{1}{T}\sum_{i=1}^{n}\big(\Delta_{i}^{n}X\big)^{2}.
\end{align*}
Broadly, we recommend to use the loose estimator $\wh{\sigma}_{n,01}^{2}$ as our initial value $\sigma_{0}^{2}$ if the volatility is ``large" (say, 0.4 or larger), and, otherwise, use the tighter estimator $\wh{\sigma}_{n,02}^{2}$.

For the moment functions $\boldsymbol{g}:=(g_{1},g_{2})^{\text{T}}$ in Step \ref{item:GenIterStep4} of the algorithm in Subsection \ref{Sect51} (the ones used to correct estimates of $\boldsymbol{r}$ and $\boldsymbol{\alpha}$ while fixing that of $\sigma^{2}$), we choose
\begin{align}\label{eq:2MomFunctsCGMY}
g_{1}(x):=\big(1-|x|\big){\bf 1}_{\{|x|<1\}},\quad g_{2}(x):=\big(1-x^{2}\big){\bf 1}_{\{|x|<1\}},\quad x\in\bR.
\end{align}
Finally, we use the approximation \eqref{eq:EquepsApprx2} in Steps \ref{item:GenIterStep3} and \ref{item:GenIterStep5} of the algorithm outlined in Subsection \ref{Sect51}. For clarity and easy reference, we outline below the precise estimation procedure for the case of the CGMY model.
\begin{enumerate}
\item Start with some initial guesses $(\sigma_{0}^{2},C_{0},Y_{0})$. Here we take $C_{0}=0.1$, $Y_{0}=1.3$, and $\sigma_{0}^{2}=\wh\sigma_{n,01}^{2}$ when $\sigma=0.4$ or $\sigma_{0}^{2}=\wh{\sigma}_{n,02}^{2}$ when $\sigma=0.2$, as defined in \eqref{eq:IniSqusigma}. Given $\sigma_{0}^{2}$, we fix the scaling factor $u_{n}=1/\sqrt{2\sigma_{0}^{2}h_{n}\ln(1/h_{n})}$.
\item Find the roots of $\boldsymbol{F}_{n}(\boldsymbol{\theta})$ with the moment functions $\boldsymbol{f}$ in \eqref{eq:3MomFunctsCGMY}, which we call $(\wh{\sigma}_{n,1}^{2},\wh{C}_{n,1},\wh{Y}_{n,1})$, by minimizing the objective function $V_{n}(\boldsymbol{\theta};\boldsymbol{f})$ in \eqref{eq:ObjFunt}.
\item \label{item:IterCGMYStep3}
    Using $(\wh{\sigma}_{n,1}^{2},\wh{C}_{n,1},\wh{Y}_{n,1})$, we apply the second-order approximation \eqref{eq:EquepsApprx2} to get an estimate of $\varepsilon_{n}^{\star}$, denoted by $\wh{\varepsilon}_{n,1}$, and compute its corresponding TRQV estimator $\wh{\sigma}^{2}_{n,2}$.
\item Fix $\wh{\sigma}_{n,2}^{2}$ and then use the moment functions $\boldsymbol{g}$ in \eqref{eq:2MomFunctsCGMY} with $u_{n}=1/\sqrt{2\wh{\sigma}^{2}_{n,2}h_{n}\ln(1/h_{n})}$ to get the estimates $(\wh{C}_{n,2},\wh{Y}_{n,2})$ by solving the roots of $\boldsymbol{G}_{n}(C,Y;\wh{\sigma}_{n,2}^{2})$ in \eqref{eq:FunctGn} or minimizing $V_{n}(C,Y;\wh{\sigma}_{n,2}^{2},\boldsymbol{g})$ in \eqref{eq:ObjFuntGn}\footnote{In Steps 2 and 4, we choose the minimization method. We use the R function {\tt nloptr} from the package {\tt nloptr} with the algorithm NLOPT\!\_\!\_LD\!\_\!\_LBFGS.}
\item \label{item:IterCGMYStep5}
    Using $(\wh{\sigma}_{n,2}^{2},\wh{C}_{n,2},\wh{Y}_{n,2})$, we again apply \eqref{eq:EquepsApprx2} to get a new estimate of $\varepsilon^{\star}_{n}$, denoted by $\wh{\varepsilon}_{n}^{\,\star}$. This threshold is plugged into the TRQV estimator to compute a final estimation of $\sigma^{2}$, denoted by $(\wh{\sigma}^{\star}_{n})^{2}$.
\end{enumerate}

We compare the simulated performance of our estimator $(\wh{\sigma}_{n}^{\star})^{2}$ to the estimator $\wh{\sigma}_{n,1}^{2}$ (which could be considered the plain estimator proposed by \cite{Mies:2019}) and the estimator $\wh{\sigma}_{n,\text{JT}}^{2}$ in \cite{JacodTodorov:2014}. In the latter one, we use the equation (5.3) therein with $\zeta=1.5$ and {\Blue $k_{n}=252\times 6.5\times 12=19656$}, which is reasonable since the volatility is constant and there is no need to localize the estimator (so we only need one block). We take the scaling factor $u_{n}=(\ln(1/h_{n}))^{-1/30}$, as proposed in the simulation portion of \cite{Mies:2019}, and $\bar{u}_{n}=(8/3)u_{n}$ for the term $S_{T}^{n}$ of equation (5.3) in \cite{JacodTodorov:2014}. \cite{JacodTodorov:2014} suggests to use $u_{n}=(\ln(1/h_{n}))^{-1/30}/\sqrt{BV}$ and $\bar{u}_{n}=0.3\,u_{n}$, where $BV$ is the bipower variation, which we also tried in our simulation, but obtained worst results. In fact, we tried different parameters settings for $\zeta$, $k_{n}$, $u_{n}$, and $\bar{u}_{n}$, and select the values with the best performance. In each simulation, we divided the one-year data into $12$ months and compute the estimate of $\sigma^{2}$ for each month, and then take the average of these $12$ monthly estimators as $\wh{\sigma}_{n,\text{JT}}^{2}$.

The results are summarized in Tables \ref{table:C28Y17sigma202}$-$\ref{table:C28Y135sigma401} for different parameter settings. The tables report the sample means, standard deviations (SDs), sample mean and SD of relative errors, and MSEs for different parameter settings based on $2000$ simulations. We also report the {TRQV} estimator, denoted by $(\wt{\sigma}_{n}^{\star})^{2}$, using {the threshold $\wt{\varepsilon}_{n}^{\,\star}$ given in \eqref{eq:EquepsApprx2}} with the true values of $\sigma^{2}$, $C$, and $Y$. Finally, we also {report the TRQV} estimator, denoted by $(\sigma^{\star}_{n})^{2}$, corresponding to the true {optimal threshold} $\varepsilon^{\star}_{n}$ obtained by solving \eqref{eq:Equeps} after finding $\bE(b_{1}(\varepsilon))$ via a large scale Monte Carlo experiment.

Tables \ref{table:C28Y17sigma202}$-$\ref{table:C28Y135sigma202} show that, when $\sigma=0.2$, the MSEs of $\sigma_{0}^{2}=\wh{\sigma}_{n,02}^{2}$, $\wh{\sigma}_{n,1}^{2}$, $\wh{\sigma}_{n,2}^{2}$, and $(\wh{\sigma}_{n}^{\star})^{2}$ are getting smaller in each step. The MSE of $(\wh{\sigma}_{n}^{\star})^{2}$ is about $81.8\%$, $71.8\%$, and $56.8\%$ lower than the MSE of $\wh{\sigma}_{n,1}^{2}$, for $Y=1.7,\,1.5,\,1.35$, respectively, while this is $98.5\%$, $23.4\%$, and $21.7\%$ lower than the MSE of $\wh{\sigma}_{n,\text{JT}}^{2}$, for $Y=1.7,\,1.5,\,1.35$, respectively. Similarly, as shown in Tables \ref{table:C28Y17sigma401}$-$\ref{table:C28Y135sigma401}, when $\sigma=0.4$, the MSEs of $\sigma_{0}^{2}=\wh{\sigma}_{n,02}^{2}$, $\wh{\sigma}_{n,1}^{2}$, $\wh{\sigma}_{n,2}^{2}$, and $(\wh{\sigma}_{n}^{\star})^{2}$ are also getting smaller in each step. The MSE of $(\wh{\sigma}_{n}^{\star})^{2}$ is about $35.3\%$, $47.3\%$, and $56.8\%$ lower than the MSE of $\wh{\sigma}_{n,1}^{2}$, for $Y=1.7,\,1.5,\,1.35$, respectively. The MSE of $(\wh{\sigma}_{n}^{\star})^{2}$ are $96.5\%$, $3.5\%$, and $58.5\%$ lower than the MSE of $\wh{\sigma}_{n,\text{JT}}^{2}$, for $Y=1.7,\,1.5,\,1.35$, respectively. So the iterative method has a good performance and significantly improves the MSEs of the estimators $\wh{\sigma}_{n,1}^{2}$ and $\wh{\sigma}_{n,\text{JT}}^{2}$. Regarding the estimates of $Y$ and $C$, we notice that the second step estimates $\wh{Y}_{n,2}$ and $\wh{C}_{n,2}$ (obtained from fixing $\wh{\sigma}^{2}_{n,2}$ and then applying \eqref{eq:FunctGn}) are significantly better than the first step estimates $\wh{Y}_{n,1}$ and $\wh{C}_{n,1}$ when $Y\leq 1.5$. When $Y=1.7$, there is no significant improvement.

\medskip
\begin{table}[ht!]
\centering
\begin{tabular}{cccccc}
\hline
& \begin{tabular}{c} Sample Mean \end{tabular} & \begin{tabular}{c} Sample SD \end{tabular} & \begin{tabular}{c} Mean of \\ Relative Error \end{tabular} & \begin{tabular}{c} SD of \\ Relative Error \end{tabular} & \begin{tabular}{c} MSE \end{tabular} \\
\hline
$\sigma_{0}^{2}$ & 0.083128 & 0.001101 & 1.078190 & 0.027526 & 1.8612E-03 \\
$\wh{\sigma}_{n,1}^{2}$ & 0.045788 & 0.001450 & 0.144696 & 0.036257 & 3.5602E-05 \\
$\wh{C}_{n,1}$ & 0.038116 & 0.008910 & 0.361281 & 0.318223 & 1.8172E-04 \\
$\wh{Y}_{n,1}$ & 1.649710 & 0.028049 & -0.029582 & 0.016499 & 3.3158E-03 \\
$\wh{\varepsilon}_{n,1}$ & 0.003361 & 0.000105 & 0.053750 & 0.032827 & 4.0365E-08 \\
$\wh{\sigma}_{n,2}^{2}$ & 0.043465 & 0.002063 & 0.086634 & 0.051572 & {1.6264E-05} \\
$\wh{C}_{n,2}$ & 0.039492 & 0.009402 & 0.410411 & 0.335795 & 2.2046E-04 \\
$\wh{Y}_{n,2}$ & 1.649803 & 0.030043 & -0.029528 & 0.017672 & 3.4224E-03 \\
$\wh{\varepsilon}_{n}^{\,\star}$ & 0.003219 & 0.000122 & 0.009143 & 0.038203 & 1.5703E-08 \\
$(\wh{\sigma}_{n}^{\star})^{2}$ & 0.040622 & 0.002467 & 0.015561 & 0.061681 & 6.4747E-06 \\
$\wh{\sigma}_{n,\text{JT}}^{2}$ & 0.060648 & 0.003161 & 0.516191 & 0.079013 & 4.3631E-04 \\
$(\wt{\sigma}_{n}^{\star})^{2}$ & 0.037780 & 0.000326 & -0.055490 & 0.008160 & 5.0331E-06 \\
$(\sigma^{\star}_{n})^{2}$ & 0.040061 & 0.000350 & 0.001528 & 0.008747 & 1.2615E-07 \\
\hline
\end{tabular}
\caption{\small Estimation based on simulated $5$-minute observations of $2000$ paths over a one-year time horizon. The parameters are $C=0.028$, $Y=1.7$, and $\sigma=0.2$. We take $\sigma_{0}^{2}=\wh{\sigma}_{n,02}^{2}$. In this case, we compute $\varepsilon^{\star}_{n}=0.00319$ and $\wt{\varepsilon}_{n}^{\,\star}=0.003080$.} \label{table:C28Y17sigma202}
\end{table}

\begin{table}[ht]
\centering
\begin{tabular}{cccccc}
\hline
& \begin{tabular}{c} Sample Mean\end{tabular} & \begin{tabular}{c} Sample SD \end{tabular} & \begin{tabular}{c} Mean of \\ Relative Error \end{tabular} & \begin{tabular}{c} SD of \\ Relative Error \end{tabular} & \begin{tabular}{c} MSE \end{tabular} \\
\hline
$\sigma_{0}^{2}$ & 0.048487 & 0.000569 & 0.212166 & 0.014213 & 7.2347E-05 \\
$\wh{\sigma}_{n,1}^{2}$ & 0.034272 & 0.001892 & -0.143197 & 0.047310 & 3.6390E-05 \\
$\wh{C}_{n,1}$ & 0.009674 & 0.004040 & -0.654512 & 0.144291 & 3.5218E-04 \\
$\wh{Y}_{n,1}$ & 1.718718 & 0.052867 & 0.145812 & 0.035245 & 5.0632E-02 \\
$\wh{\varepsilon}_{n,1}$ & 0.003300 & 0.000184 & -0.122242 & 0.048974 & 2.4517E-07 \\
$\wh{\sigma}_{n,2}^{2}$ & 0.036122 & 0.002019 & -0.096961 & 0.050470 & 1.9118E-05 \\
$\wh{C}_{n,2}$ & 0.018535 & 0.009119 & -0.338043 & 0.325676 & 1.7275E-04 \\
$\wh{Y}_{n,2}$ & 1.618217 & 0.071105 & 0.078811 & 0.047403 & 1.9031E-02 \\
$\wh{\varepsilon}_{n}^{\,\star}$ & 0.003470 & 0.000206 & -0.077027 & 0.054674 & 1.2614E-07 \\
$(\wh{\sigma}_{n}^{\star})^{2}$ & 0.037833 & 0.001989 & -0.054178 & 0.049722 & 8.6521E-06 \\
$\wh{\sigma}_{n,\text{JT}}^{2}$ & 0.043206 & 0.001006 & 0.080159 & 0.025152 & 1.1293E-05 \\
$(\wt{\sigma}_{n}^{\star})^{2}$ & 0.041920 & 0.000393 & 0.048008 & 0.009837 & 3.8424E-06 \\
$(\sigma^{\star}_{n})^{2}$ & 0.040460 & 0.000371 & 0.011491 & 0.009284 & 3.4918E-07 \\
\hline
\end{tabular}
\caption{\small Estimation based on simulated $5$-minute observations of $2000$ paths over a one-year time horizon. The parameters are $C=0.028$, $Y=1.5$, and $\sigma=0.2$. We take $\sigma_{0}^{2}=\wh{\sigma}_{n,02}^{2}$. In this case, we compute $\varepsilon^{\star}_{n}=0.00376$ and $\wt{\varepsilon}_{n}^{\,\star}=0.003974$.} \label{table:C28Y15sigma202}
\end{table}

\begin{table}[ht]
\centering
\begin{tabular}{cccccc}
\hline
& \begin{tabular}{c} Sample Mean \end{tabular} & \begin{tabular}{c} Sample SD \end{tabular} & \begin{tabular}{c} Mean of \\ Relative Error \end{tabular} & \begin{tabular}{c} SD of \\ Relative Error \end{tabular} & \begin{tabular}{c} MSE \end{tabular} \\
\hline
$\sigma_{0}^{2}$ & 0.042889 & 0.000460 & 0.072221 & 0.011491 & 8.5567E-06 \\
$\wh{\sigma}_{n,1}^{2}$ & 0.039098 & 0.001497 & -0.022549 & 0.037431 & 3.0552E-06 \\
$\wh{C}_{n,1}$ & 0.006717 & 0.004401 & -0.760116 & 0.157188 & 4.7235E-04 \\
$\wh{Y}_{n,1}$ & 1.620586 & 0.059843 & 0.200434 & 0.044328 & 7.6798E-02 \\
$\wh{\varepsilon}_{n,1}$ & 0.004171 & 0.000238 & -0.013883 & 0.056241 & 6.0045E-08 \\
$\wh{\sigma}_{n,2}^{2}$ & 0.039907 & 0.000981 & -0.002313 & 0.024514 & 9.7007E-07 \\
$\wh{C}_{n,2}$ & 0.049001 & 0.031537 & 0.750027 & 1.126331 & 1.4356E-03 \\
$\wh{Y}_{n,2}$ & 1.285035 & 0.105981 & -0.048123 & 0.078504 & 1.5452E-02 \\
$\wh{\varepsilon}_{n}^{\,\star}$ & 0.004374 & 0.000181 & 0.034146 & 0.042847 & 5.3710E-08 \\
$(\wh{\sigma}_{n}^{\star})^{2}$ & 0.040596 & 0.000718 & 0.014911 & 0.017949 & 8.7122E-07 \\
$\wh{\sigma}_{n,\text{JT}}^{2}$ & 0.040842 & 0.000634 & 0.021061 & 0.015858 & 1.1121E-06 \\
$(\wt{\sigma}_{n}^{\star})^{2}$ & 0.040789 & 0.000389 & 0.019736 & 0.009719 & 7.7435E-07 \\
$(\sigma^{\star}_{n})^{2}$ & 0.040235 & 0.000378 & 0.005870 & 0.009440 & 1.9770E-07 \\
\hline
\end{tabular}
\caption{\small Estimation based on simulated $5$-minute observations of $2000$ paths over a one-year time horizon. The parameters are $C=0.028$, $Y=1.35$, and $\sigma=0.2$. We take $\sigma_{0}^{2}=\wh{\sigma}_{n,02}^{2}$. In this case, we compute $\varepsilon^{\star}_{n}=0.00423$ and $\wt{\varepsilon}_{n}^{\,\star}=0.004421$.} \label{table:C28Y135sigma202}
\end{table}

\begin{table}[ht!]
\vspace{0.5cm}
\centering
\begin{tabular}{cccccc}
\hline
& \begin{tabular}{c} Sample Mean \end{tabular} & \begin{tabular}{c} Sample SD \end{tabular} & \begin{tabular}{c} Mean of \\ Relative Error \end{tabular} & \begin{tabular}{c} SD of \\ Relative Error \end{tabular} & \begin{tabular}{c} MSE \end{tabular} \\
\hline
$\sigma_{0}^{2}$ & 0.233102 & 0.003657 & 0.456887 & 0.022855 & 5.3573E-03 \\
$\wh{\sigma}_{n,1}^{2}$ & 0.164513 & 0.003573 & 0.028205 & 0.022332 & 3.3133E-05 \\
$\wh{C}_{n,1}$ & 0.044023 & 0.010945 & 0.572264 & 0.390893 & 3.7654E-04 \\
$\wh{Y}_{n,1}$ & 1.641484 & 0.029400 & -0.034421 & 0.017294 & 4.2884E-03 \\
$\wh{\varepsilon}_{n,1}$ & 0.007081 & 0.000133 & 0.004358 & 0.018910 & 1.8717E-08 \\
$\wh{\sigma}_{n,2}^{2}$ & 0.162042 & 0.003665 & 0.012764 & 0.022907 & 1.7603E-05 \\
$\wh{C}_{n,2}$ & 0.051271 & 0.009833 & 0.831099 & 0.351195 & 6.3823E-04 \\
$\wh{Y}_{n,2}$ & 1.619172 & 0.022776 & -0.047546 & 0.013397 & 7.0519E-03 \\
$\wh{\varepsilon}_{n}^{\,\star}$ & 0.007028 & 0.000167 & -0.003155 & 0.023635 & 2.8259E-08 \\
$(\wh{\sigma}_{n}^{\star})^{2}$ & 0.160897 & 0.004536 & 0.005605 & 0.028350 & 2.1380E-05 \\
$\wh{\sigma}_{n,\text{JT}}^{2}$ & 0.177376 & 0.017487 & 0.108603 & 0.109294 & 6.0774E-04 \\
$(\wt{\sigma}_{n}^{\star})^{2}$ & 0.159377 & 0.001432 & -0.003894 & 0.008947 & 2.4375E-06 \\
$(\sigma^{\star}_{n})^{2}$ & 0.161471 & 0.001461 & 0.009193 & 0.009133 & 4.2992E-06 \\
\hline
\end{tabular}
\caption{\small Estimation based on simulated $5$-minute observations of $2000$ paths over a one-year time horizon. The parameters are $C=0.028$, $Y=1.7$, and $\sigma=0.4$. We take $\sigma_{0}^{2}=\wh{\sigma}_{n,01}^{2}$. In this case, we compute $\varepsilon^{\star}_{n}=0.00705$ and $\wt{\varepsilon}_{n}^{\,\star}=0.006954$.} \label{table:C28Y17sigma401}
\end{table}

\begin{table}[ht!]
\centering
\begin{tabular}{cccccc}
\hline
& \begin{tabular}{c} Sample Mean \end{tabular} & \begin{tabular}{c} Sample SD \end{tabular} & \begin{tabular}{c} Mean of \\ Relative Error \end{tabular} & \begin{tabular}{c} SD of \\ Relative Error \end{tabular} & \begin{tabular}{c} MSE \end{tabular} \\
\hline
$\sigma_{0}^{2}$ & 0.184729 & 0.002771 & 0.154558 & 0.017316 & 6.1921E-04 \\
$\wh{\sigma}_{n,1}^{2}$ & 0.150140 & 0.005615 & -0.061628 & 0.035094 & 1.2876E-04 \\
$\wh{C}_{n,1}$ & 0.010208 & 0.004122 & -0.635430 & 0.147228 & 3.3355E-04 \\
$\wh{Y}_{n,1}$ & 1.750480 & 0.033385 & 0.166986 & 0.022257 & 6.3855E-02 \\
$\wh{\varepsilon}_{n,1}$ & 0.007357 & 0.000442 & -0.124190 & 0.052650 & 1.2838E-06 \\
$\wh{\sigma}_{n,2}^{2}$ & 0.150856 & 0.004893 & -0.057151 & 0.030584 & 1.0756E-04 \\
$\wh{C}_{n,2}$ & 0.021081 & 0.004636 & -0.247106 & 0.165562 & 6.9362E-05 \\
$\wh{Y}_{n,2}$ & 1.635476 & 0.038174 & 0.090317 & 0.025449 & 1.9811E-02 \\
$\wh{\varepsilon}_{n}^{\,\star}$ & 0.007556 & 0.000369 & -0.100445 & 0.043904 & 8.4791E-07 \\
$(\wh{\sigma}_{n}^{\star})^{2}$ & 0.153194 & 0.004645 & -0.042537 & 0.029031 & 6.7897E-05 \\
$\wh{\sigma}_{n,\text{JT}}^{2}$ & 0.161593 & 0.008238 & 0.009958 & 0.051485 & 7.0396E-05 \\
$(\wt{\sigma}_{n}^{\star})^{2}$ & 0.161458 & 0.001550 & 0.009112 & 0.009684 & 4.5267E-06 \\
$(\sigma^{\star}_{n})^{2}$ & 0.161165 & 0.001548 & 0.007280 & 0.009678 & 3.7547E-06 \\
\hline
\end{tabular}
\caption{\small Estimation based on simulated $5$-minute observations of $2000$ paths over a one-year time horizon. The parameters are $C=0.028$, $Y=1.5$, and $\sigma=0.4$. We take $\sigma_{0}^{2}=\wh{\sigma}_{n,01}^{2}$. In this case, we compute $\varepsilon^{\star}_{n}=0.0084$ and $\wt{\varepsilon}_{n}^{\,\star}=0.008444$.} \label{table:C28Y15sigma401}
\end{table}

\begin{table}[ht!]
\vspace{0.5cm}
\centering
\begin{tabular}{cccccc}
\hline
& \begin{tabular}{c} Sample Mean \end{tabular} & \begin{tabular}{c} Sample SD \end{tabular} & \begin{tabular}{c} Mean of \\ Relative Error \end{tabular} & \begin{tabular}{c} SD of \\ Relative Error \end{tabular} & \begin{tabular}{c} MSE \end{tabular} \\
\hline
$\sigma_{0}^{2}$ & 0.172473 & 0.002347 & 0.077956 & 0.014667 & 1.6108E-04 \\
$\wh{\sigma}_{n,1}^{2}$ & 0.154598 & 0.007496 & -0.033763 & 0.046851 & 8.5376E-05 \\
$\wh{C}_{n,1}$ & 0.003444 & 0.002236 & -0.877006 & 0.079863 & 6.0801E-04 \\
$\wh{Y}_{n,1}$ & 1.762248 & 0.118304 & 0.305369 & 0.087633 & 1.8394E-01 \\
$\wh{\varepsilon}_{n,1}$ & 0.008141 & 0.000684 & -0.120849 & 0.073887 & 1.7204E-06 \\
$\wh{\sigma}_{n,2}^{2}$ & 0.154037 & 0.004898 & -0.037268 & 0.030615 & 5.9550E-05 \\
$\wh{C}_{n,2}$ & 0.016458 & 0.010854 & -0.412220 & 0.387637 & 2.5103E-04 \\
$\wh{Y}_{n,2}$ & 1.523080 & 0.171764 & 0.128208 & 0.127233 & 5.9460E-02 \\
$\wh{\varepsilon}_{n}^{\,\star}$ & 0.007751 & 0.002270 & -0.162936 & 0.245181 & 7.4310E-06 \\
$(\wh{\sigma}_{n}^{\star})^{2}$ & 0.156800 & 0.005165 & -0.020002 & 0.032280 & 3.6916E-05 \\
$\wh{\sigma}_{n,\text{JT}}^{2}$ & 0.158787 & 0.009349 & -0.007584 & 0.058430 & 8.8874E-05 \\
$(\wt{\sigma}_{n}^{\star})^{2}$ & 0.160665 & 0.001567 & 0.004155 & 0.009797 & 2.8988E-06 \\
$(\sigma^{\star}_{n})^{2}$ & 0.160439 & 0.001562 & 0.002746 & 0.009763 & 2.6332E-06 \\
\hline
\end{tabular}
\caption{\small Estimation performance based on simulated $5$-minute observations of $2000$ paths over a one-year time horizon. The parameters are $C=0.028$, $Y=1.35$, and $\sigma=0.4$. We take $\sigma_{0}^{2}=\wh{\sigma}_{n,01}^{2}$. In this case, we compute $\varepsilon^{\star}_{n}=0.00926$ and $\wt{\varepsilon}_{n}^{\,\star}=0.009338$.} \label{table:C28Y135sigma401}
\end{table}

\subsection{Integrated Variance Estimation for a Stochastic Volatility Model with Jumps}

In this subsection, we apply the method in the previous subsection to estimate the integrated variance under a stochastic volatility model with a CGMY jump component. We also examine the finite sample performance of the resulting estimator and compare it with the estimator of Jacod and Todorov \cite{JacodTodorov:2014}. The basic idea is to split the time-period $[0,T]$ into smaller subintervals so that $\sigma$ would be approximately constant in each subinterval and, hence, $X$ is approximately L\'{e}vy within that interval. We then apply the method developed in Subsection \ref{Sect51} to each subinterval to estimate the volatility level in each subinterval and finally aggregate the resulting estimates to estimate the integrated volatility.

Specifically, we consider the following Heston model:
\begin{align*}
X_{t}=1+\int_{0}^{t}\sqrt{V_{s}}\,dW_{s}+J_{t},\quad V_{t}=\theta+\int_{0}^{t}\kappa\big(\theta-V_{s}\big)\,ds+\xi\int_{0}^{t}\sqrt{V_{s}}\,dB_{s},\quad t\in\bR_{+},
\end{align*}
where $(W_{t})_{t\in\bR_{+}}$ and $(B_{t})_{t\in\bR_{+}}$ are two independent standard Brownian motions and $(J_{t})_{t\in\bR_{+}}$ is a CGMY L\'{e}vy process independent of $(W_{t})_{t\in\bR_{+}}$ and $(B_{t})_{t\in\bR_{+}}$. The parameters of the volatility specification are set as
\begin{align*}
\kappa=5,\quad\xi=0.5,\quad\theta=0.16.
\end{align*}
The values of $\kappa$ and $\xi$ above are borrowed from \cite{ZhangMyklandAitSahalia:2005}. In the simulation, we experiment with values of $Y=1.5$ and $Y=1.7$, and compute the estimated integrated volatility for one day under two different estimators.

We consider $5$-second observations over a one-year ($252$ days) time horizon with $6.5$ trading hours per day. We set $k_{n}=160$, which corresponds to $30$ blocks per day. As mentioned above, we treat the stochastic volatility as a constant volatility in each block, so that we can estimate the integrated volatility for each block by computing our estimator $(\wh{\sigma}_{n}^{\star})^{2}$ with all the estimation parameters specified as in Subsection \ref{constCGMY}. We then add the integrated volatilities for the $30$ blocks to compute our daily estimator of the integrated volatility $\int_{t}^{t+1/252}V_{s}ds$ for that day. For the estimator of {\cite{JacodTodorov:2014}}, we use both equations (4.2) and (5.3) therein with $k_{n}=160$ (number of observation in each block), $\xi=1.5$, and $u_{n}=0.05(-\ln h_{n})^{-1/30}/\sqrt{BV}$, where $BV$ is the bipower variation of the previous day. To assess the accuracy of the different methods, we compute the Median Absolute Deviation (MAD) around the true value, $\int_{t}^{t+1/252}V_{s}ds$, over $200$ simulation paths.

Figure \ref{Fig:path53} shows the estimated integrated volatility for each day computed by our MME (solid black line) and the JT estimator in \cite{JacodTodorov:2014} (dotted blue line) (see Eq.~(5.3) therein), and the true daily integrated volatility (dashed red) for $1$ simulation path. When $Y=1.7$, the JT estimator tends to jitter around the true value, while our MME exhibits better performance. This behavior is further corroborated by Table \ref{tab7}, which shows the MADs of both our MME and JT estimator for $6$ different days based on $200$ simulated paths. However, when $Y=1.5$, the JT estimator outperforms our MME, as shown in the right panel of Figure \ref{Fig:path53} and Table \ref{tab7}. Now, it is important to point out that the daily estimate (5.3) of \cite{JacodTodorov:2014} is based on more data than that used in our estimates. Indeed, the estimator (5.3) in \cite{JacodTodorov:2014} employs a debiasing procedure, whose debiasing term consists of two components: one that can be computed using the data in each day and another term, depending only on $Y$, that is computed using the data during the whole time horizon (in this case, one year worth of data). To explore the performance of the estimator using only contemporary data, we also analyze the performance of the estimator (4.2) in \cite{JacodTodorov:2014}, which can be computed using only the data collected in each day. When $Y=1.7$ (left panel of Figure \ref{Fig:path42}), the daily estimated integrated volatility (4.2) in \cite{JacodTodorov:2014} overestimates the true integrated volatility, and produces extremely large or small estimates at some points. We also observe this behavior when $Y=1.5$ (right panel of Figure \ref{Fig:path42}), but the estimate is much more stable and outperforms our MME most of the time, as shown in Table \ref{tab7}. To summarize, when $Y>1.5$ is large, our approach performs fairly well for the stochastic volatility model.

\begin{figure}[ht!]
\begin{center}
\hspace{-0.3cm}
\includegraphics[width=8cm,height=6cm]{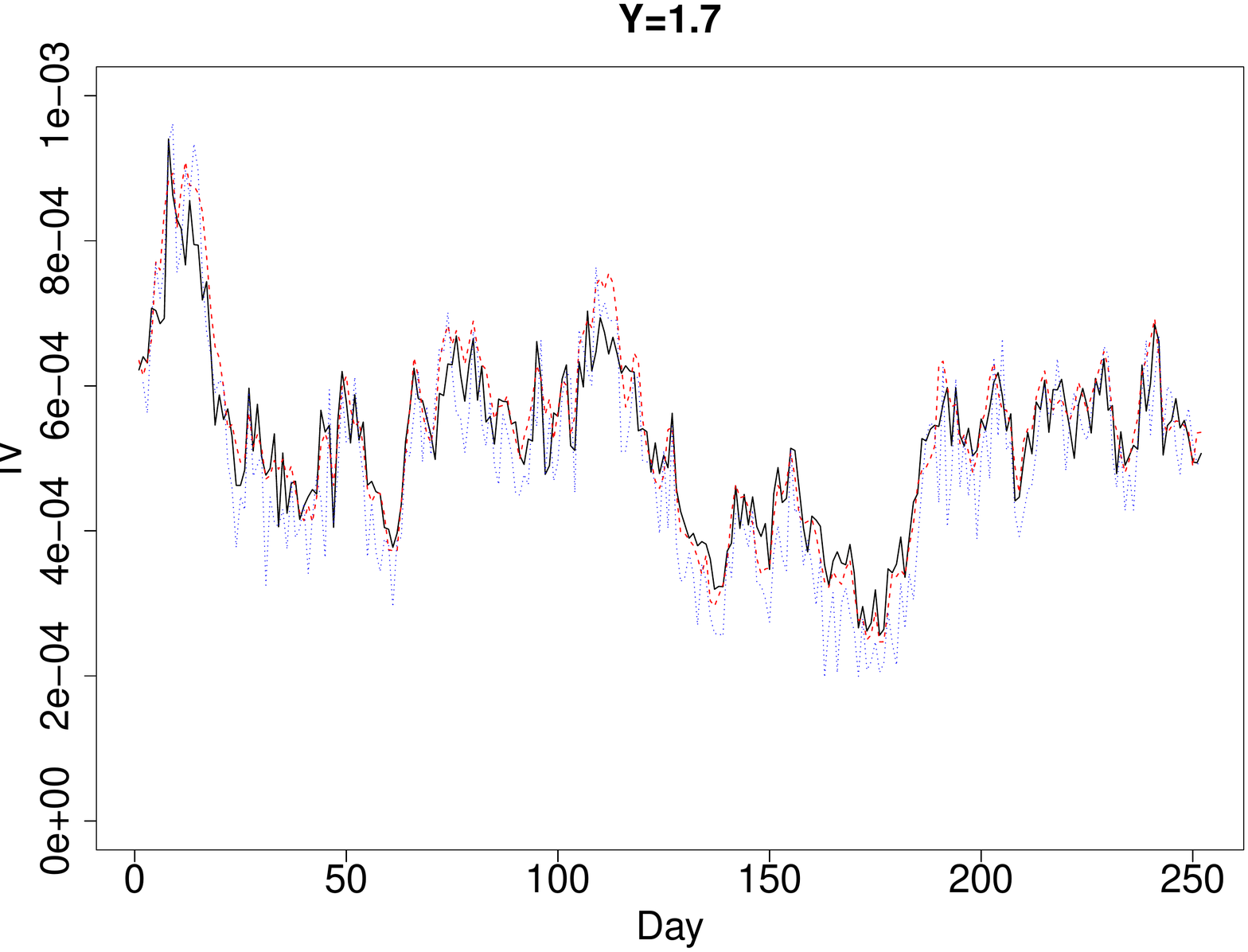}
\hspace{-0.4cm}
\includegraphics[width=8cm,height=6cm]{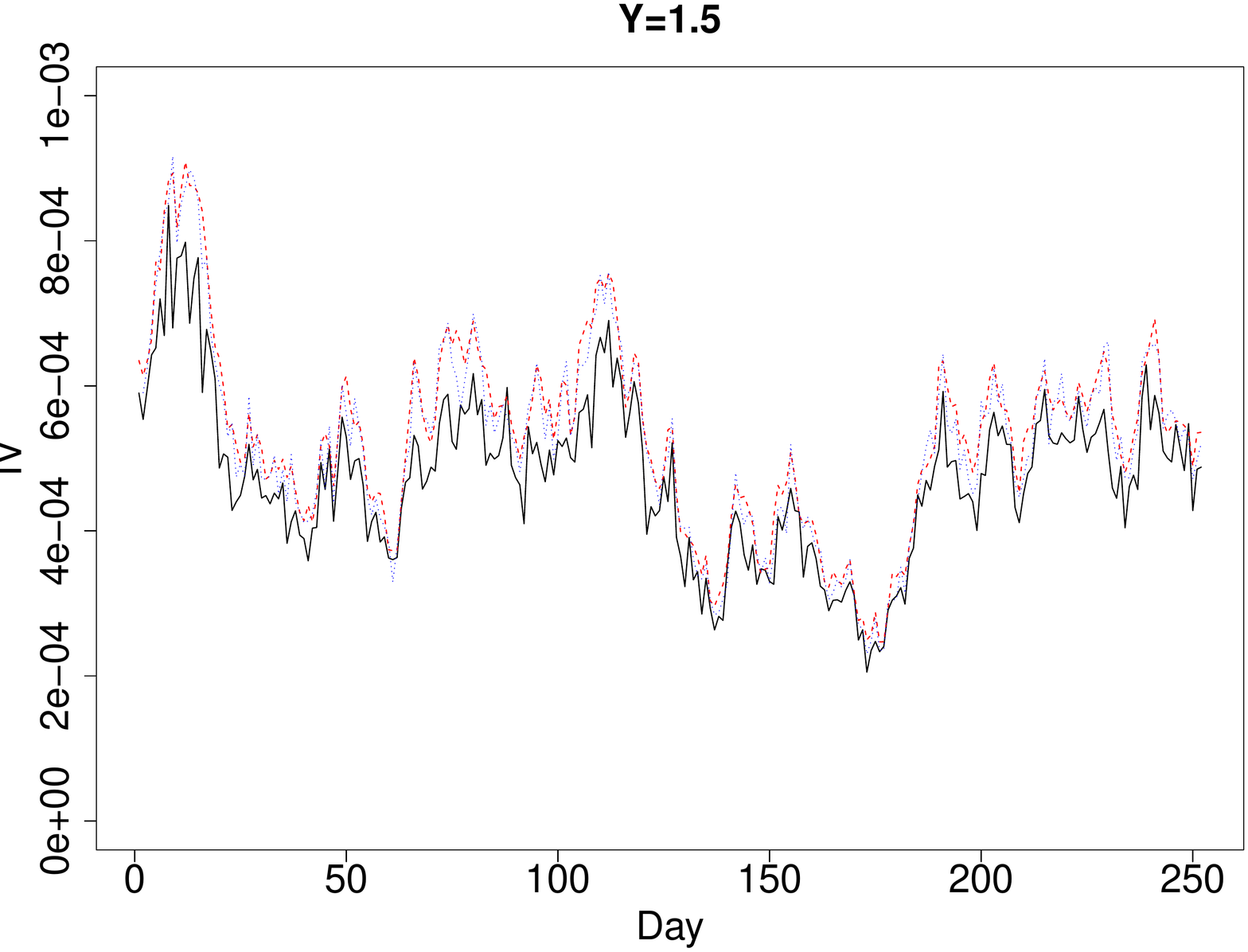}
\end{center}
\vspace{-0.3cm}
\caption{\small Graphs of the daily integrated volatility estimates for $Y=1.7$ (left panel) and $Y=1.5$ (right panel), respectively. The dashed red line is the true daily integrated volatility, while the solid black (respectively, dotted blue) line shows the daily estimates using our MME (respectively, (5.3) in \cite{JacodTodorov:2014}).} \label{Fig:path53}
\end{figure}

\begin{figure}[ht!]
\begin{center}
\hspace{-0.3cm}
\includegraphics[width=8cm,height=6cm]{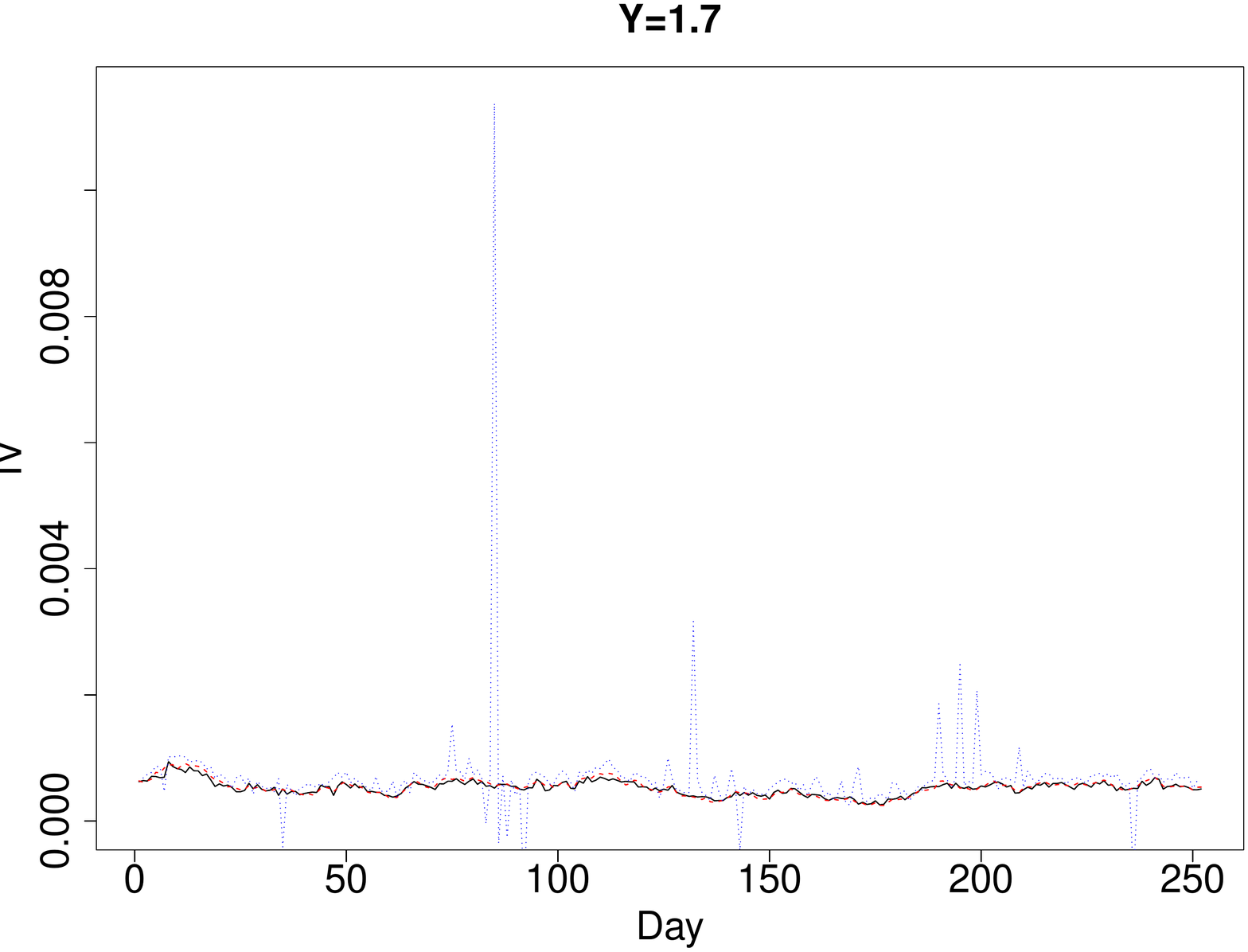}
\hspace{-0.4cm}
\includegraphics[width=8cm,height=6cm]{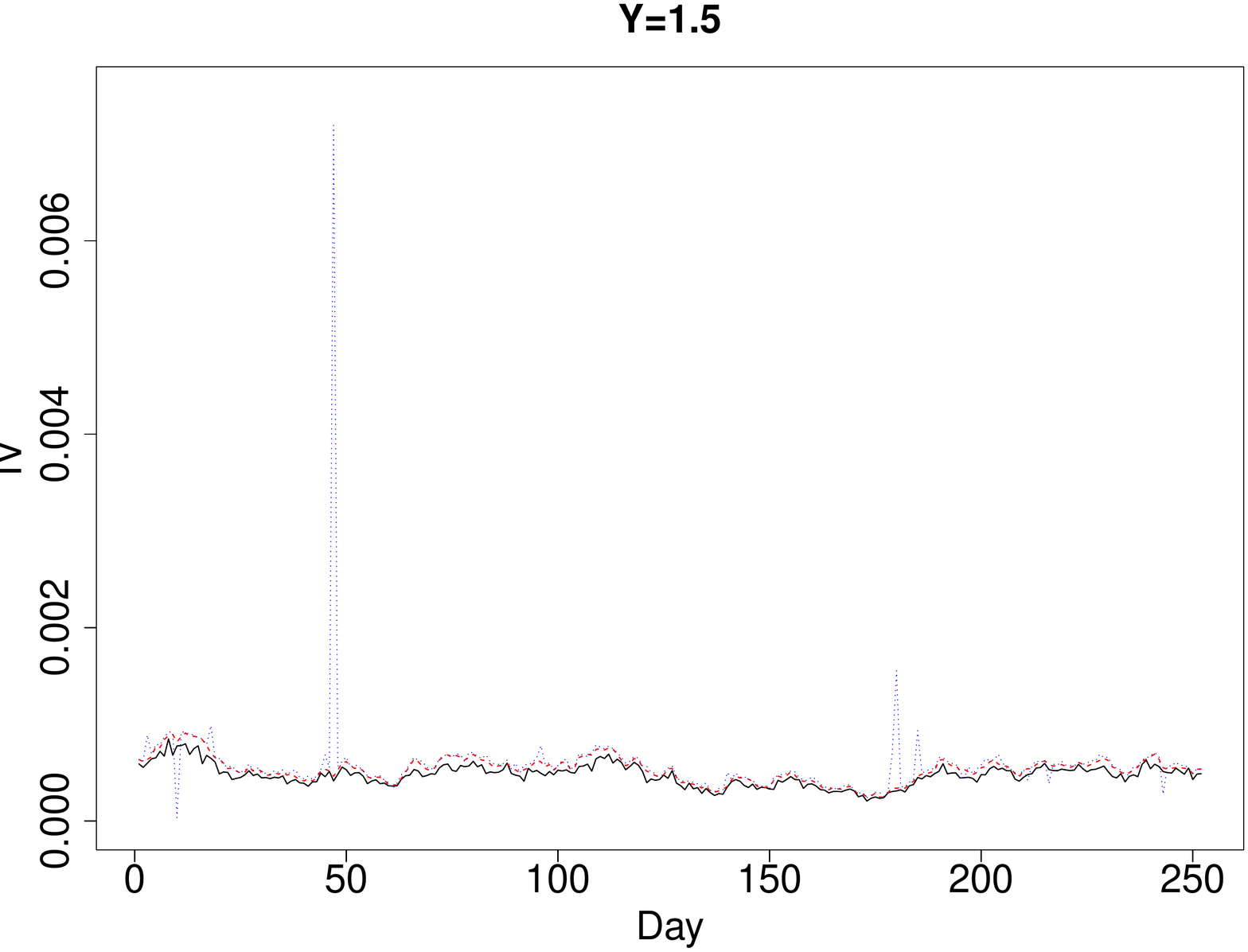}
\end{center}
\vspace{-0.3cm}
\caption{\small Graphs of the daily integrated volatility estimates for $Y=1.7$ (left panel) and $Y=1.5$ (right panel), respectively. The dashed red line is the true daily integrated volatility, while the solid black (respectively, dotted blue) line shows the daily estimates using our MME (respectively, (4.2) in \cite{JacodTodorov:2014}).} \label{Fig:path42}
\end{figure}

\begin{table}[ht!]
\vspace{0.5cm}
\centering
\begin{tabular}{c|ccccccc}
\hline
$Y$ & Method & Day 2 & Day 52 & Day 102 & Day 152 & Day 202 & Day 252 \\
\hline
\multirow{3}{*}{$1.7$} & MAD.MME & 2.4498E-05 & 3.7819E-05 & 3.4568E-05 & 3.7059E-05 & 3.7501E-05 & 2.7117E-05 \\
& MAD.JT (4.2) & 1.0087E-04 & 9.6535E-05 & 1.1127E-04 & 1.0820E-04 & 1.0253E-04 & 1.0623E-04 \\
& MAD.JT (5.3) & 5.0516E-05 & 4.6288E-05 & 4.7421E-05 & 4.9100E-05 & 5.2230E-05 & 4.3879E-05 \\
\hline
\multirow{3}{*}{$1.5$} & MAD.MME & 6.8551E-05 & 6.6529E-05 & 5.6600E-05 & 6.4519E-05 & 6.3395E-05 & 6.2592E-05 \\
& MAD.JT (4.2) & 1.9978E-05 & 1.9429E-05 & 1.8714E-05 & 1.7270E-05 & 2.0778E-05 & 2.2255E-05 \\
& MAD.JT (5.3) & 1.3410E-05 & 1.1542E-05 & 1.0880E-05 & 9.5053E-06 & 1.6174E-05 & 1.4615E-05 \\
\hline
\end{tabular}
\caption{\small The MADs for our MME and the estimators in \cite{JacodTodorov:2014} for $6$ arbitrary sampled days. The results are based on simulated $5$-second observations of $200$ paths over a one-year time horizon with $Y=1.7$ and $Y=1.5$, respectively.}\label{tab7}
\end{table}

\section*{Acknowledgements}
{\Red The authors are grateful to the Editor and two anonymous referees for their multiple suggestions that help to improve the original manuscript.}

\appendix

\section{Proof of Theorem \ref{lem:Asymb1eps}}\label{sec:ProofLemAsymb1eps}

We start with the decomposition:
\begin{align}
\bE\big(b_{1}(\varepsilon)\big)&=\bE\Big(\big(\sigma W_{h}+J_{h}\big)^{2}\,{\bf 1}_{\{|\sigma W_{h}+J_{h}|\leq\varepsilon\}}\Big)\\
\label{eq:Decompb1eps} &=\sigma^{2}\,\bE\Big(W_{h}^{2}\,{\bf 1}_{\{|\sigma W_{h}+J_{h}|\leq\varepsilon\}}\Big)\!+\!\bE\Big(J_{h}^{2}\,{\bf 1}_{\{|\sigma W_{h}+J_{h}|\leq\varepsilon\}}\Big)\!+\!2\sigma\,\bE\Big(W_{h}J_{h}{\bf 1}_{\{|\sigma W_{h}+J_{h}|\leq\varepsilon\}}\Big).\quad
\end{align}
The asymptotic behavior of each of the terms above is obtained in three steps.

\bigskip
\noindent
\textbf{Step 1.} We first analyze the behavior of the first term in \eqref{eq:Decompb1eps} as $h\rightarrow 0$. By \eqref{eq:DenTranTildePP} and \eqref{eq:DecompUpm}, we first decompose it as
\begin{align}
\bE\Big(W_{h}^{2}\,{\bf 1}_{\{|\sigma W_{h}+J_{h}|\leq\varepsilon\}}\Big)&=h-\bE\Big(W_{h}^{2}\,{\bf 1}_{\{|\sigma W_{h}+J_{h}|>\varepsilon\}}\Big)=h-\wt{\bE}\Big(e^{-\wt{U}_{h}-\eta h}\,W_{h}^{2}\,{\bf 1}_{\{|\sigma W_{h}+J_{h}|>\varepsilon\}}\Big)\\
&=h-he^{-\eta h}\,\wt{\bE}\Big(W_{1}^{2}\,{\bf 1}_{\{|\sigma\sqrt{h}W_{1}+J_{h}|>\varepsilon\}}\Big)\\
&\quad -he^{-\eta h}\,\wt{\bE}\Big(\Big(e^{-\wt{U}_{h}}-1\Big)W_{1}^{2}\,{\bf 1}_{\{|\sigma\sqrt{h}W_{1}+J_{h}|>\varepsilon\}}\Big)\\
\label{eq:Decompb1eps1} &=:h-he^{-\eta h}I_{1}(h)-he^{-\eta h}I_{2}(h).
\end{align}
In what follows, we will analyze the asymptotic behavior of $I_{1}(h)$ and $I_{2}(h)$, as $h\rightarrow 0$, respectively, in two sub-steps.

\medskip
\noindent
\textbf{Step 1.1.} Clearly, by \eqref{eq:DecompJZpm} and the symmetry of $W_{1}$,
\begin{align}\label{eq:Decompb1eps12}
I_{1}(h)=\wt{\bE}\Big(W_{1}^{2}\,{\bf 1}_{\{\sigma\sqrt{h}W_{1}+J_{h}>\varepsilon\}}\Big)+\wt{\bE}\Big(W_{1}^{2}\,{\bf 1}_{\{\sigma\sqrt{h}W_{1}-J_{h}>\varepsilon\}}\Big)=:I_{1}^{+}(h)+I_{1}^{-}(h).
\end{align}
Denote by
\begin{align*}
\phi(x):=\frac{1}{\sqrt{2\pi}}e^{-x^{2}/2},\quad\overline{\Phi}(x):=\int_{x}^{\infty}\phi(x)\,dx,\quad x\in\bR.
\end{align*}
By conditioning on $J_{h}$ and using the fact that $\wt{\bE}(W_{1}^{2}{\bf 1}_{\{W_{1}>x\}})=x\phi(x)+\overline{\Phi}(x)$, for all $x\in\bR$, we have that
\begin{align}\label{eq:Decompb1eps12pm}
I_{1}^{\pm}(h)=\wt{\bE}\left(\bigg(\frac{\varepsilon}{\sigma\sqrt{h}}\mp\frac{J_{h}}{\sigma\sqrt{h}}\bigg)\phi\bigg(\frac{\varepsilon}{\sigma\sqrt{h}}\mp\frac{J_{h}}{\sigma\sqrt{h}}\bigg)+\overline{\Phi}\bigg(\frac{\varepsilon}{\sigma\sqrt{h}}\mp\frac{J_{h}}{\sigma\sqrt{h}}\bigg)\right).
\end{align}

Let $p_{J,h}^{\pm}$ be the density of $\pm J_{h}$ under $\wt{\bP}$, and recall the Fourier transform and its inverse transform defined by
\begin{align*}
(\cF g)(z):=\frac{1}{\sqrt{2\pi}}\int_{\bR}g(x)e^{-izx}\,dx,\quad\big(\cF^{-1}g\big)(x):=\frac{1}{\sqrt{2\pi}}\int_{\bR}g(z)e^{izx}\,dz.
\end{align*}
In what follows, we set
\begin{align}
\psi(x)&:=\bigg(\cF^{-1}\phi\bigg(\frac{\cdot}{\sigma\sqrt{h}}-\frac{\varepsilon}{\sigma\sqrt{h}}\bigg)\bigg)(x)=\frac{1}{\sqrt{2\pi}}\int_{\bR}\phi\bigg(\frac{z}{\sigma\sqrt{h}}-\frac{\varepsilon}{\sigma\sqrt{h}}\bigg)e^{izx}\,dz\\
\label{eq:Defpsi} &\,=\frac{\sigma\sqrt{h}\,e^{i\varepsilon x}}{\sqrt{2\pi}}\int_{\bR}\phi(\omega)e^{i\sigma\sqrt{h}\omega x}\,d\omega=\frac{\sigma\sqrt{h}}{\sqrt{2\pi}}\exp\bigg(i\varepsilon x-\frac{1}{2}\sigma^{2}x^{2}h\bigg).
\end{align}
Then, we deduce that
\begin{align*}
\wt{\bE}\left(\phi\bigg(\frac{\varepsilon}{\sigma\sqrt{h}}\pm\frac{J_{h}}{\sigma\sqrt{h}}\bigg)\right)=\int_{\bR}\big(\cF\psi\big)(z)p_{J,h}^{\mp}(z)\,dz=\int_{\bR}\psi(u)\big(\cF p_{J,h}^{\mp}\big)(u)\,du,
\end{align*}
where
\begin{align}
\big(\cF p_{J,h}^{\pm}\big)(u)&=\frac{e^{\mp iu\wt{\gamma}h}}{\sqrt{2\pi}}\exp\!\left(\!-(C_{+}\!+C_{-})\,\bigg|\!\cos\bigg(\!\frac{\pi Y}{2}\!\bigg)\!\bigg|\Gamma(-Y)|u|^{Y}h\bigg(\!1-i\frac{C_{+}\!-C_{-}}{C_{+}\!+C_{-}}\tan\!\bigg(\!\frac{\pi Y}{2}\!\bigg)\text{sgn}(u)\!\bigg)\!\right)\\
\label{eq:FourierDenJhpm} &=:\frac{1}{\sqrt{2\pi}}\exp\Big(c_{1}|u|^{Y}h+ic_{2}|u|^{Y}h\,\text{sgn}(u)\mp iu\wt{\gamma}h\Big),
\end{align}
with {$c_{1}:=(C_{+}\!+C_{-})\,\cos(\pi Y/2)\Gamma(-Y)$ and $c_{2}:=(C_{-}-C_{+})\sin(\pi Y/2)\Gamma(-Y)$. Hence, we have}
\begin{align}\label{eq:Expphi}
\wt{\bE}\left(\phi\bigg(\frac{\varepsilon}{\sigma\sqrt{h}}\pm\frac{J_{h}}{\sigma\sqrt{h}}\bigg)\right)&=\frac{\sigma\sqrt{h}}{2\pi}\!\!\int_{\bR}\exp\!\bigg(\!c_{1}|u|^{Y}h+ic_{2}|u|^{Y}h\,\text{sgn}(u)-\frac{1}{2}\sigma^{2}u^{2}h+iu\big(\varepsilon\!\pm\!\wt{\gamma}h\big)\!\bigg)du\\
&=\frac{\sigma\sqrt{h}}{\pi}\int_{0}^{\infty}\exp\bigg(c_{1}u^{Y}h-\frac{1}{2}\sigma^{2}u^{2}h\bigg)\cos\Big(c_{2}u^{Y}h+u\big(\varepsilon\pm\wt{\gamma}h\big)\Big)du\\
&=\frac{1}{\pi}\int_{0}^{\infty}\exp\!\bigg(c_{1}\cdot\frac{\omega^{Y}h^{1-Y/2}}{\sigma^{Y}}-\frac{\omega^{2}}{2}\bigg)\cos\!\bigg(c_{2}\cdot\frac{\omega^{Y}h^{1-Y/2}}{\sigma^{Y}}+\omega\cdot\frac{\varepsilon\!\pm\!\wt{\gamma}h}{\sigma\sqrt{h}}\bigg)d\omega\\
&=\frac{1}{\pi}\int_{0}^{\infty}\!\exp\!\bigg(c_{1}\!\cdot\!\frac{\omega^{Y}h^{1-Y/2}}{\sigma^{Y}}\!-\!\frac{\omega^{2}}{2}\bigg)\cos\!\bigg(c_{2}\!\cdot\!\frac{\omega^{Y}h^{1-Y/2}}{\sigma^{Y}}\bigg)\cos\!\bigg(\omega\!\cdot\!\frac{\varepsilon\!\pm\!\wt{\gamma}h}{\sigma\sqrt{h}}\bigg)d\omega\\
&\quad -\!\frac{1}{\pi}\!\int_{0}^{\infty}\!\exp\!\bigg(\!c_{1}\!\cdot\!\frac{\omega^{Y}h^{1-Y/2}}{\sigma^{Y}}\!-\!\frac{\omega^{2}}{2}\bigg)\sin\!\bigg(\!c_{2}\!\cdot\!\frac{\omega^{Y}h^{1-Y/2}}{\sigma^{Y}}\bigg)\sin\!\bigg(\!\omega\!\cdot\!\frac{\varepsilon\!\pm\!\wt{\gamma}h}{\sigma\sqrt{h}}\bigg)d\omega.
\end{align}
By expanding the Taylor series for $\exp(c_{1}\sigma^{-Y}|\omega|^{Y}h^{1-Y/2})$, as well as for $\cos(c_{2}\sigma^{-Y}\omega^{Y}h^{1-Y/2})$ and $\sin(c_{2}\sigma^{-Y}\omega^{Y}h^{1-Y/2})$, we deduce that
\begin{align}
&\wt{\bE}\left(\phi\bigg(\frac{\varepsilon}{\sigma\sqrt{h}}\pm\frac{J_{h}}{\sigma\sqrt{h}}\bigg)\right)\\
&\,\,\,\,=\frac{1}{\pi}\!\int_{0}^{\infty}\!\cos\!\bigg(\omega\!\cdot\!\frac{\varepsilon\pm\wt{\gamma}h}{\sigma\sqrt{h}}\bigg)e^{-\omega^{2}/2}d\omega\!+\!\frac{1}{\pi}\!\!\sum_{(k,j)\in\bZ^{2}_{+}:\,(k,j)\neq (0,0)}\!\!(-1)^{j}a_{k,2j,0}^{\pm}(h)\!-\!\frac{1}{\pi}\sum_{k,j=0}^{\infty}(-1)^{j}d_{k,2j+1,0}^{\pm}(h)\\
\label{eq:ExpandExpphi} &\,\,\,\,=\frac{1}{\sqrt{2\pi}}\exp\bigg(\!\!-\!\frac{(\varepsilon\pm\wt{\gamma}h)^{2}}{2\sigma^{2}h}\bigg)+\frac{1}{\pi}\!\sum_{(k,j)\in\bZ^{2}_{+}:\,(k,j)\neq (0,0)}\!(-1)^{j}a_{k,2j,0}^{\pm}(h)-\frac{1}{\pi}\sum_{k,j=0}^{\infty}(-1)^{j}d_{k,2j+1,0}^{\pm}(h),
\end{align}
where, for $m,n\in\bZ_{+}$ and $r\in\bR_{+}$,
\begin{align}\label{eq:KummerPowerCos}
a_{m,n,r}^{\pm}(h)&:=\frac{c_{1}^{m}\,c_{2}^{n}\,h^{(m+n)(1-Y/2)}}{m!\,n!\,\sigma^{(m+n)Y}}\int_{0}^{\infty}\omega^{(m+n)Y+r}\cos\bigg(\omega\cdot\frac{\varepsilon\pm\wt{\gamma}h}{\sigma\sqrt{h}}\bigg)e^{-\omega^{2}/2}\,d\omega,\\
\label{eq:KummerPowerSin} d_{m,n,r}^{\pm}(h)&:=\frac{c_{1}^{m}\,c_{2}^{n}\,h^{(m+n)(1-Y/2)}}{m!\,n!\,\sigma^{(m+n)Y}}\int_{0}^{\infty}\omega^{(m+n)Y+r}\sin\bigg(\omega\cdot\frac{\varepsilon\pm\wt{\gamma}h}{\sigma\sqrt{h}}\bigg)e^{-\omega^{2}/2}\,d\omega.
\end{align}
By applying the formula for the integrals of $\omega^{kY}\cos(\beta\omega)$ and $\omega^{kY}\sin(\beta\omega)$ with respect to $e^{-\omega^{2}/2}$ on $\bR_{+}$, as well as the asymptotics for the Kummer's function $M(a,b,z)$, as $h\rightarrow 0$, we deduce that
\begin{align*}
a_{m,n,r}^{\pm}(h)&=\frac{c_{1}^{m}\,c_{2}^{n}\,h^{(m+n)(1-Y/2)}}{m!\,n!\,\sigma^{(m+n)Y}}\cdot 2^{((m+n)Y+r-1)/2}\,\Gamma\bigg(\frac{(m+n)Y+r+1}{2}\bigg)\\
&\quad\,\cdot M\bigg(\frac{(m+n)Y+r+1}{2},\frac{1}{2},-\frac{\big(\varepsilon\pm\wt{\gamma}h\big)^{2}}{2\sigma^{2}h}\bigg)\\
&\sim\frac{c_{1}^{m}\,c_{2}^{n}\,h^{(m+n)(1-Y/2)}}{m!\,n!\,\sigma^{(m+n)Y}}\cdot 2^{((m+n)Y+r-1)/2}\,\Gamma\bigg(\frac{(m+n)Y+r+1}{2}\bigg)\\
&\quad\,\cdot\left(\frac{\Gamma(1/2)}{\Gamma\big(\!-((m+n)Y+r)/2\big)}\bigg(\frac{\varepsilon^{2}}{2\sigma^{2}h}\bigg)^{-((m+n)Y+r+1)/2}\right.\\
&\qquad\,\,\left.+\frac{\Gamma(1/2)\,e^{-\varepsilon^{2}/(2\sigma^{2}h)}}{\Gamma\big(((m+n)Y+r+1)/2\big)}\bigg(\frac{\varepsilon^{2}}{2\sigma^{2}h}\bigg)^{((m+n)Y+r)/2}\right),
\end{align*}
and that
\begin{align*}
d_{m,n,r}^{\pm}(h)&=\frac{c_{1}^{m}\,c_{2}^{n}\,h^{(m+n)(1-Y/2)}}{m!\,n!\,\sigma^{(m+n)Y}}\cdot\frac{\varepsilon\pm\wt{\gamma}h}{\sigma\sqrt{h}}\cdot 2^{((m+n)Y+r)/2}\,\Gamma\bigg(\frac{(m+n)Y+r}{2}+1\bigg)\\
&\quad\,\cdot M\bigg(\frac{(m+n)Y+r}{2}+1,\frac{3}{2},-\frac{\big(\varepsilon\pm\wt{\gamma}h\big)^{2}}{2\sigma^{2}h}\bigg)\\
&\sim\frac{c_{1}^{m}\,c_{2}^{n}\,h^{(m+n)(1-Y/2)}}{m!\,n!\,\sigma^{(m+n)Y}}\cdot\frac{\varepsilon}{\sigma\sqrt{h}}\cdot 2^{((m+n)Y+r)/2}\,\Gamma\bigg(\frac{(m+n)Y+r}{2}+1\bigg)\\
&\quad\,\cdot\left(\frac{\Gamma(3/2)}{\Gamma\big((1-(m+n)Y-r)/2\big)}\cdot\bigg(\frac{\varepsilon^{2}}{2\sigma^{2}h}\bigg)^{-((m+n)Y+r)/2-1}\right.\\
&\qquad\quad\left.+\frac{\Gamma(3/2)\,e^{-\varepsilon^{2}/(2\sigma^{2}h)}}{\Gamma\big((m+n)Y+r)/2+1\big)}\cdot\bigg(\frac{\varepsilon^{2}}{2\sigma^{2}h}\bigg)^{((m+n)Y+r-1)/2}\right).
\end{align*}
In the asymptotic formulas for the Kummers function in the expression of $a_{n,m,r}^{\pm}(h)$ above, the first term vanishes if $\Gamma(-((m+n)Y+r)/2)$ are infinity. This happens when $-((m+n)Y+r)/2$ is a nonpositive integer. Similarly, in the asymptotic formulas for the Kummers function in $d_{m,n,r}^{\pm}(h)$, the first term vanishes if $(1-(m+n)Y-r)/2$ is a nonpositive integer. Hence, for $m,n\in\bZ_{+}$ and $r\in\bR_{+}$, as $h\rightarrow 0$,
\begin{align}\label{eq:LimitKummerPowerCos}
a_{m,n,r}^{\pm}(h)&=O\bigg(\frac{h^{m+n+(r+1)/2}}{\varepsilon^{(m+n)Y+r+1}}\bigg)\!+O\bigg(\frac{\varepsilon^{(m+n)Y+r}e^{-\varepsilon^{2}/(2\sigma^{2}h)}}{h^{(m+n)(Y-1)+r/2}}\bigg)\!=O\bigg(\frac{h^{m+n+(r+1)/2}}{\varepsilon^{(m+n)Y+r+1}}\bigg),\\
\label{eq:LimitKummerPowerSin} d_{m,n,r}^{\pm}(h)&=O\bigg(\frac{h^{m+n+(r+1)/2}}{\varepsilon^{(m+n)Y+r+1}}\bigg)\!+O\bigg(\frac{\varepsilon^{(m+n)Y+r}e^{-\varepsilon^{2}/(2\sigma^{2}h)}}{h^{(m+n)(Y-1)+r/2}}\bigg)\!=O\bigg(\frac{h^{m+n+(r+1)/2}}{\varepsilon^{(m+n)Y+r+1}}\bigg).
\end{align}
Therefore, by combining \eqref{eq:ExpandExpphi}, \eqref{eq:LimitKummerPowerCos}, and \eqref{eq:LimitKummerPowerSin}, we obtain that
\begin{align}\label{eq:LimitExpphi}
\wt{\bE}\left(\phi\bigg(\frac{\varepsilon}{\sigma\sqrt{h}}\pm\frac{J_{h}}{\sigma\sqrt{h}}\bigg)\right)=\frac{e^{-\varepsilon^{2}/(2\sigma^{2}h)}}{\sqrt{2\pi}}+O\big(h^{3/2}\varepsilon^{-1-Y}\big),\quad\text{as }\,h\rightarrow 0.
\end{align}

Next, we note that
\begin{align*}
\wt{\bE}\left(\mp J_{h}\,\phi\bigg(\frac{\varepsilon}{\sigma\sqrt{h}}\pm\frac{J_{h}}{\sigma\sqrt{h}}\bigg)\right)=\int_{\bR}(\cF\psi)(z)zp_{J,h}^{\mp}(z)\,dz=\int_{\bR}\psi(u)\cF\big(zp_{J,h}^{\mp}(z)\big)(u)\,du,
\end{align*}
where by \eqref{eq:FourierDenJhpm},
\begin{align*}
\cF\big(zp_{J,h}^{\pm}(z)\big)(u)=i\frac{d}{du}\big(\cF p_{J,h}^{\pm}\big)(u)&=\frac{i}{\sqrt{2\pi}}\exp\Big(c_{1}|u|^{Y}h+ic_{2}|u|^{Y}h\,\text{sgn}(u)\mp iu\wt{\gamma}h\Big)\\
&\quad\,\cdot\Big(c_{1}Y|u|^{Y-1}\,\text{sgn}(u)h+ic_{2}Y|u|^{Y-1}h\mp i\wt{\gamma}h\Big).
\end{align*}
Together with \eqref{eq:Defpsi}, we obtain that
\begin{align}
&\wt{\bE}\left(\mp J_{h}\,\phi\bigg(\frac{\varepsilon}{\sigma\sqrt{h}}\pm\frac{J_{h}}{\sigma\sqrt{h}}\bigg)\right)\\
&\quad =\frac{ic_{1}\sigma Yh^{3/2}}{2\pi}\int_{\bR}\,\text{\sgn}(u)|u|^{Y-1}\exp\bigg(c_{1}|u|^{Y}h+ic_{2}|u|^{Y}h\,\text{sgn}(u)-\frac{\sigma^{2}u^{2}h}{2}+iu\big(\varepsilon\pm\wt{\gamma}h\big)\bigg)du\\
&\qquad -\frac{c_{2}\sigma Yh^{3/2}}{2\pi}\int_{\bR}\,|u|^{Y-1}\exp\bigg(c_{1}|u|^{Y}h+ic_{2}|u|^{Y}h\,\text{sgn}(u)-\frac{\sigma^{2}u^{2}h}{2}+iu\big(\varepsilon\pm\wt{\gamma}h\big)\bigg)du\\
\label{eq:DecompExpJhphi} &\qquad\mp\frac{\wt{\gamma}\sigma h^{3/2}}{2\pi}\int_{\bR}\,\exp\bigg(c_{1}|u|^{Y}h+ic_{2}|u|^{Y}h\,\text{sgn}(u)-\frac{\sigma^{2}u^{2}h}{2}+iu\big(\varepsilon\pm\wt{\gamma}h\big)\bigg)du
\end{align}
The asymptotics of the last term above follows from \eqref{eq:Expphi} and \eqref{eq:LimitExpphi}, namely
\begin{align}
&\frac{\wt{\gamma}\sigma h^{3/2}}{2\pi}\int_{\bR}\,\exp\bigg(c_{1}|u|^{Y}h+ic_{2}|u|^{Y}h\,\text{sgn}(u)-\frac{\sigma^{2}u^{2}h}{2}+iu\big(\varepsilon\pm\wt{\gamma}h\big)\bigg)du\\
\label{eq:LimitExpJhphi3} &\quad =O\Big(h\,e^{-\varepsilon^{2}/(2\sigma^{2}h)}\Big)\!+\!O\big(h^{5/2}\varepsilon^{-1-Y}\big),\quad\text{as }\,h\rightarrow 0.
\end{align}
For the first term in \eqref{eq:DecompExpJhphi}, by expanding the Taylor series for $\exp(c_{1}\sigma^{-Y}|\omega|^{Y}h^{1-Y/2})$, as well as for $\cos(c_{2}\sigma^{-Y}\omega^{Y}h^{1-Y/2})$ and $\sin(c_{2}\sigma^{-Y}\omega^{Y}h^{1-Y/2})$ below, and using \eqref{eq:KummerPowerCos} and \eqref{eq:KummerPowerSin}, we have
\begin{align}
&\frac{ic_{1}\sigma Yh^{3/2}}{2\pi}\int_{\bR}\,\text{\sgn}(u)|u|^{Y-1}\exp\bigg(c_{1}|u|^{Y}h+ic_{2}|u|^{Y}h\,\text{sgn}(u)-\frac{\sigma^{2}u^{2}h}{2}+iu\big(\varepsilon\pm\wt{\gamma}h\big)\bigg)du\\
&\quad = -\frac{c_{1}\sigma Yh^{3/2}}{\pi}\int_{0}^{\infty}u^{Y-1}\exp\bigg(c_{1}u^{Y}h-\frac{\sigma^{2}u^{2}h}{2}\bigg)\sin\Big(c_{2}u^{Y}h+u\big(\varepsilon\pm\wt{\gamma}h\big)\Big)du\\
&\quad = -\frac{c_{1}Yh^{(3-Y)/2}}{\sigma^{Y-1}\pi}\int_{0}^{\infty}\omega^{Y-1}\exp\bigg(\frac{c_{1}\,\omega^{Y}h^{1-Y/2}}{\sigma^{Y}}-\frac{\omega^{2}}{2}\bigg)\sin\bigg(\frac{c_{2}\,\omega^{Y}h^{1-Y/2}}{\sigma^{Y}}+\omega\cdot\frac{\varepsilon\pm\wt{\gamma}h}{\sigma\sqrt{h}}\bigg)d\omega\\
&\quad = -\frac{c_{1}Yh^{(3-Y)/2}}{\sigma^{Y-1}\pi}\int_{0}^{\infty}\omega^{Y-1}\exp\bigg(\frac{c_{1}\,\omega^{Y}h^{1-Y/2}}{\sigma^{Y}}-\frac{\omega^{2}}{2}\bigg)\sin\bigg(\frac{c_{2}\,\omega^{Y}h^{1-Y/2}}{\sigma^{Y}}\bigg)\cos\bigg(\omega\cdot\frac{\varepsilon\pm\wt{\gamma}h}{\sigma\sqrt{h}}\bigg)d\omega\\
&\qquad -\frac{c_{1}Yh^{(3-Y)/2}}{\sigma^{Y-1}\pi}\int_{0}^{\infty}\omega^{Y-1}\exp\bigg(\frac{c_{1}\,\omega^{Y}h^{1-Y/2}}{\sigma^{Y}}-\frac{\omega^{2}}{2}\bigg)\cos\bigg(\frac{c_{2}\,\omega^{Y}h^{1-Y/2}}{\sigma^{Y}}\bigg)\sin\bigg(\omega\cdot\frac{\varepsilon\pm\wt{\gamma}h}{\sigma\sqrt{h}}\bigg)d\omega\\
\label{eq:LimitExpJhphi1} &\quad = -\frac{c_{1}Yh^{(3-Y)/2}}{\sigma^{Y-1}\pi}\left(\sum_{k,j=0}^{\infty}(-1)^{j}a_{k,2j+1,Y-1}^{\pm}(h)+\sum_{k,j=0}^{\infty}(-1)^{j}d_{k,2j,Y-1}^{\pm}(h)\right)=O\big(h^{3/2}\varepsilon^{-Y}\big),
\end{align}
as $h\rightarrow 0$, where we have used the asymptotic formulas \eqref{eq:LimitKummerPowerCos} and \eqref{eq:LimitKummerPowerSin} in the last equality. Finally, for the second term in \eqref{eq:DecompExpJhphi}, again by expanding the Taylor series for $\exp(c_{1}\sigma^{-Y}|\omega|^{Y}h^{1-Y/2})$, $\cos(c_{2}\sigma^{-Y}\omega^{Y}h^{1-Y/2})$, and $\sin(c_{2}\sigma^{-Y}\omega^{Y}h^{1-Y/2})$ below, and using \eqref{eq:KummerPowerCos}, \eqref{eq:KummerPowerSin}, \eqref{eq:LimitKummerPowerCos}, and \eqref{eq:LimitKummerPowerSin}, we deduce that
\begin{align}
&\frac{c_{2}\sigma Yh^{3/2}}{2\pi}\int_{\bR}\,|u|^{Y-1}\exp\bigg(c_{1}|u|^{Y}h+ic_{2}|u|^{Y}h\,\text{sgn}(u)-\frac{\sigma^{2}u^{2}h}{2}+iu\big(\varepsilon\pm\wt{\gamma}h\big)\bigg)du\\
&\quad =\frac{c_{2}\sigma Yh^{3/2}}{\pi}\int_{0}^{\infty}u^{Y-1}\exp\bigg(c_{1}u^{Y}h-\frac{\sigma^{2}u^{2}h}{2}\bigg)\cos\Big(c_{2}u^{Y}h+u\big(\varepsilon\pm\wt{\gamma}h\big)\Big)du\\
&\quad =\frac{c_{2}Yh^{(3-Y)/2}}{\pi\sigma^{Y-1}}\int_{0}^{\infty}\omega^{Y-1}\exp\bigg(\frac{c_{1}\,\omega^{Y}h^{1-Y/2}}{\sigma^{Y}}-\frac{\omega^{2}}{2}\bigg)\cos\bigg(\frac{c_{2}\,\omega^{Y}h^{1-Y/2}}{\sigma^{Y}}+\omega\cdot\frac{\varepsilon\pm\wt{\gamma}h}{\sigma\sqrt{h}}\bigg)d\omega\\
&\quad =\frac{c_{2}Yh^{(3-Y)/2}}{\pi\sigma^{Y-1}}\int_{0}^{\infty}\omega^{Y-1}\exp\bigg(\frac{c_{1}\,\omega^{Y}h^{1-Y/2}}{\sigma^{Y}}-\frac{\omega^{2}}{2}\bigg)\cos\bigg(\frac{c_{2}\,\omega^{Y}h^{1-Y/2}}{\sigma^{Y}}\bigg)\cos\bigg(\omega\cdot\frac{\varepsilon\pm\wt{\gamma}h}{\sigma\sqrt{h}}\bigg)d\omega\\
&\qquad -\frac{c_{2}Yh^{(3-Y)/2}}{\pi\sigma^{Y-1}}\int_{0}^{\infty}\omega^{Y-1}\exp\bigg(\frac{c_{1}\,\omega^{Y}h^{1-Y/2}}{\sigma^{Y}}-\frac{\omega^{2}}{2}\bigg)\sin\bigg(\frac{c_{2}\,\omega^{Y}h^{1-Y/2}}{\sigma^{Y}}\bigg)\sin\bigg(\omega\cdot\frac{\varepsilon\pm\wt{\gamma}h}{\sigma\sqrt{h}}\bigg)d\omega\\
\label{eq:LimitExpJhphi2} &\quad =\frac{c_{2}Yh^{(3-Y)/2}}{\pi\sigma^{Y-1}}\left(\sum_{k,j=0}^{\infty}(-1)^{j}a_{k,2j,Y-1}^{\pm}(h)-\sum_{k,j=0}^{\infty}(-1)^{j}d_{k,2j+1,Y-1}^{\pm}(h)\right)=O\big(h^{3/2}\varepsilon^{-Y}\big),
\end{align}
as $h\rightarrow 0$. Therefore, by combining \eqref{eq:DecompExpJhphi}, \eqref{eq:LimitExpJhphi3}, \eqref{eq:LimitExpJhphi1}, and \eqref{eq:LimitExpJhphi2}, we obtain that
\begin{align}\label{eq:LimitExpJhphi}
\wt{\bE}\left(\mp J_{h}\phi\bigg(\frac{\varepsilon}{\sigma\sqrt{h}}\pm\frac{J_{h}}{\sigma\sqrt{h}}\bigg)\right)=O\Big(he^{-\varepsilon^{2}/(2\sigma^{2}h)}\Big)+O\big(h^{3/2}\varepsilon^{-Y}\big),\quad\text{as }\,h\rightarrow 0.
\end{align}

It remains to analyze the asymptotic behavior of $\bE(\overline{\Phi}((\varepsilon\pm J_{h})/(\sigma\sqrt{h})))$. We first note that there exists a universal constant $K>0$, such that
\begin{align}
&\wt{\bE}\left(\overline{\Phi}\bigg(\frac{\varepsilon}{\sigma\sqrt{h}}\pm\frac{J_{h}}{\sigma\sqrt{h}}\bigg){\bf 1}_{\{\varepsilon\pm J_{h}\geq 0\}}\right)\leq K\,\wt{\bE}\left(\phi\bigg(\frac{\varepsilon}{\sigma\sqrt{h}}\pm\frac{J_{h}}{\sigma\sqrt{h}}\bigg){\bf 1}_{\{\varepsilon\pm J_{h}\geq 0\}}\right)\\
\label{eq:LimitExpPhiPlus} &\quad =O\left(\wt{\bE}\left(\phi\bigg(\frac{\varepsilon}{\sigma\sqrt{h}}\pm\frac{J_{h}}{\sigma\sqrt{h}}\bigg)\right)\right)=O\Big(e^{-\varepsilon^{2}/(2\sigma^{2}h)}\Big)+O\big(h^{3/2}\varepsilon^{-1-Y}\big),\quad\text{as }\,h\rightarrow 0,
\end{align}
where the last inequality above follows from \eqref{eq:LimitExpphi}. Moreover, by \eqref{eq:1stOrderEstTailZ}, as $h\rightarrow 0$,
\begin{align}
&\wt{\bE}\!\left(\overline{\Phi}\bigg(\frac{\varepsilon}{\sigma\sqrt{h}}\pm\frac{J_{h}}{\sigma\sqrt{h}}\bigg){\bf 1}_{\{\varepsilon\pm J_{h}\leq 0\}}\!\right)=\int_{\bR}\phi(u)\,\wt{\bP}\Big(\varepsilon\pm J_{h}\leq\sigma\sqrt{h}u,\,\varepsilon\pm J_{h}\leq 0\Big)du\\
&\quad =\int_{0}^{\infty}\phi(u)\,\wt{\bP}\Big(\pm Z_{h}\leq -\varepsilon\mp\wt{\gamma}h\Big)du+\int_{-\infty}^{0}\phi(u)\,\wt{\bP}\Big(\pm Z_{h}\leq\sigma\sqrt{h}u-\varepsilon\mp\wt{\gamma}h\Big)du\\
&\quad =\frac{1}{2}\,\wt{\bP}\bigg(\!\pm Z_{1}\leq\frac{-\varepsilon\mp\wt{\gamma}h}{h^{1/Y}}\bigg)+\int_{-\infty}^{0}\phi(u)\,\wt{\bP}\bigg(\pm Z_{1}\leq\frac{\sigma\sqrt{h}u-\varepsilon\mp\wt{\gamma}h}{h^{1/Y}}\bigg)du\\
\label{eq:LimitExpPhiMinus} &\quad\leq\frac{1}{2}\,\wt{\bP}\bigg(\!\pm Z_{1}\leq\frac{-\varepsilon\mp\wt{\gamma}h}{h^{1/Y}}\bigg)+\frac{\wt{K}h}{\varepsilon^{Y}}\int_{-\infty}^{0}\phi(u)\bigg(1-\frac{\sigma\sqrt{h}\pm\wt{\gamma}h}{\varepsilon}u\bigg)^{-Y}du=O\big(h\varepsilon^{-Y}\big).
\end{align}
Therefore, by combining \eqref{eq:LimitExpPhiPlus} and \eqref{eq:LimitExpPhiMinus}, we obtain that
\begin{align}\label{eq:LimitExpPhi}
\wt{\bE}\left(\overline{\Phi}\bigg(\frac{\varepsilon}{\sigma\sqrt{h}}\pm\frac{J_{h}}{\sigma\sqrt{h}}\bigg)\right)=O\Big(e^{-\varepsilon^{2}/(2\sigma^{2}h)}\Big)+O\big(h\varepsilon^{-Y}\big),\quad\text{as }\,h\rightarrow 0.
\end{align}
Finally, by combining \eqref{eq:Decompb1eps12}, \eqref{eq:Decompb1eps12pm}, \eqref{eq:LimitExpphi}, \eqref{eq:LimitExpJhphi}, and \eqref{eq:LimitExpPhi}, we conclude that
\begin{align}\label{eq:Limitb1eps12}
I_{1}(h)={\frac{\sqrt{2}}{\sigma\sqrt{\pi}}\cdot\frac{\varepsilon}{\sqrt{h}}\,e^{-\varepsilon^{2}/(2\sigma^{2}h)}}+O\Big(e^{-\varepsilon^{2}/(2\sigma^{2}h)}\Big)+O\big(h\varepsilon^{-Y}\big),\quad\text{as }\,h\rightarrow 0.
\end{align}

\smallskip
\noindent
\textbf{Step 1.2.} We now study the asymptotic behavior of $I_{2}(h)$, defined in \eqref{eq:Decompb1eps1}, as $h\rightarrow 0$. Let us first consider the following decomposition
\begin{align}
I_{2}(h)&=\wt{\bE}\Big(\Big(e^{-\wt{U}_{h}}-1+\wt{U}_{h}\Big)W_{1}^{2}\,{\bf 1}_{\{|\sigma\sqrt{h}W_{1}+J_{h}|>\varepsilon\}}\Big)-
\wt{\bE}\Big(\wt{U}_{h}W_{1}^{2}\,{\bf 1}_{\{|\sigma\sqrt{h}W_{1}+J_{h}|>\varepsilon\}}\Big)\\
\label{eq:Decompb1eps13} &=:I_{21}(h)-I_{22}(h).
\end{align}
The first term $I_{21}(h)$ can be bounded as follows: as $h\rightarrow 0$,
\begin{align}\label{eq:Limitb1eps131}
0\leq I_{21}(h)\leq\wt{\bE}\Big(e^{-\wt{U}_{h}}\!-1+\wt{U}_{h}\Big)=\exp\bigg(h\int_{\bR_{0}}\!\Big(e^{-\phi(x)}\!-1+\phi(x)\Big)\wt{\nu}(dx)\bigg)-1=O(h).
\end{align}
To deal with $I_{22}$, for any $t\in\bR_{+}$, we further decompose $\wt{U}_{t}$ as
\begin{align*}
\wt{U}_{t}=\int_{0}^{t}\int_{\bR_{0}}\big(\varphi(x)+\alpha_{\text{sgn}(x)}x\big)\wt{N}(ds,dx)-\int_{0}^{t}\int_{\bR_{0}}\alpha_{\text{sgn}(x)}x\wt{N}(ds,dx)=:\wt{U}^{\text{BV}}_{t}-\alpha_{+}{Z}_{t}^{+}-\alpha_{-}{Z}^{-}_{t},
\end{align*}
where the first integral is well-defined in light of Assumption \ref{assump:Funtq}-(i) \& (ii), so that
\begin{align*}
I_{22}(h)&=\wt{\bE}\Big(\wt{U}^{\text{BV}}_{h}W_{1}^{2}\,{\bf 1}_{\{|\sigma\sqrt{h}W_{1}+J_{h}|>\varepsilon\}}\Big)-\alpha_{+}\wt{\bE}\Big(Z^{+}_{h}W_{1}^{2}\,{\bf 1}_{\{|\sigma\sqrt{h}W_{1}+J_{h}|>\varepsilon\}}\Big)\\
&\quad -\alpha_{-}\wt{\bE}\Big(Z^{-}_{h}W_{1}^{2}\,{\bf 1}_{\{|\sigma\sqrt{h}W_{1}+J_{h}|>\varepsilon\}}\Big)=:I_{22}^{\text{BV}}(h)-\alpha_{+}I_{22}^{+}(h)-\alpha_{-}I_{22}^{-}(h).
\end{align*}
For the first term $I_{22}^{\text{BV}}(h)$, note that
\begin{align*}
\big|I_{22}^{\text{BV}}(h)\big|\leq\wt{\bE}\Big(\big|\wt{U}^{\text{BV}}_{h}\big|\Big)\leq 2h\int_{\bR_{0}}\big|\varphi(x)+\alpha_{\text{sgn}(x)}x\big|\,\wt{\nu}(dx),
\end{align*}
where the last integral is finite since in a neighborhood of the origin,
\begin{align*}
\big|\varphi(x)+\alpha_{\text{sgn}(x)}x\big|=\big|-\ln q(x)+\alpha_{\text{sgn}(x)}x\big|=O\Big(\big|1-q(x)+\alpha_{\text{sgn}(x)}x\big|\Big),
\end{align*}
which is integrable with respect to $\wt{\nu}(dx)$ in view of Assumption \ref{assump:Funtq}-(ii). As for the terms $I_{22}^{\pm}$, due to the self-similarity of $Z_{t}^{\pm}$ and the fact that $\varepsilon h^{-1/Y}\rightarrow\infty$ (since $Y\in(1,2)$), the monotone convergence theorem implies that $I_{22}^{\pm}(h)=o(h^{1/Y})$, as $h\rightarrow 0$. Hence, we obtain that
\begin{align}\label{eq:Limitb1eps132}
I_{22}(h)=o\big(h^{1/Y}\big),\quad\text{as }\,h\rightarrow 0.
\end{align}
By combining {\eqref{eq:Decompb1eps13}, \eqref{eq:Limitb1eps131}, and \eqref{eq:Limitb1eps132}}, we conclude that
\begin{align}\label{eq:Limitb1eps13}
I_{2}(h)=o\big(h^{1/Y}\big),\quad\text{as }\,h\rightarrow 0.
\end{align}
Finally, from \eqref{eq:Decompb1eps1}, \eqref{eq:Limitb1eps12}, and \eqref{eq:Limitb1eps13}, we obtain that
\begin{align}\label{eq:Limitb1eps1}
\bE\Big(\!W_{h}^{2}{\bf 1}_{\{|\sigma W_{h}+J_{h}|\leq\varepsilon\}}\!\Big)\!=\!h\!-\!\frac{\sqrt{2h}\,\varepsilon}{\sigma\sqrt{\pi}}e^{-\varepsilon^{2}/(2\sigma^{2}h)}\!+\!O\Big(\!he^{-\varepsilon^{2}/(2\sigma^{2}h)}\!\Big)\!+\!O\big(h^{2}\varepsilon^{-Y}\big)\!+{\!o\big(h^{1+1/Y}\big)},\quad
\end{align}
as $h\rightarrow 0$, which completes the analysis in Step 1.

\bigskip
\noindent
\textbf{Step 2.} In this step, we will study the asymptotic behavior of the second term in \eqref{eq:Decompb1eps}, as $h\rightarrow 0$. By \eqref{eq:DenTranTildePP}, \eqref{eq:DecompUpm}, and \eqref{eq:DecompJZpm}, we first have
\begin{align}
\bE\Big(J_{h}^{2}\,{\bf 1}_{\{|\sigma W_{h}+J_{h}|\leq\varepsilon\}}\Big)&=\wt{\bE}\Big(e^{-\wt{U}_{h}-\eta h}\,J_{h}^{2}\,{\bf 1}_{\{|\sigma W_{h}+J_{h}|\leq\varepsilon\}}\Big)\\
&=e^{-\eta h}\,\wt{\bE}\Big(e^{-\wt{U}_{h}}Z_{h}^{2}{\bf 1}_{\{|W_{h}+Z_{h}+\wt{\gamma}h|\leq\varepsilon\}}\Big)+2\wt{\gamma}he^{-\eta h}\,\wt{\bE}\Big(e^{-\wt{U}_{h}}Z_{h}{\bf 1}_{\{|W_{h}+Z_{h}+\wt{\gamma}h|\leq\varepsilon\}}\Big)\\
&\quad +\wt{\gamma}^{2}h^{2}e^{-\eta h}\,\wt{\bE}\Big(e^{-\wt{U}_{h}}\,{\bf 1}_{\{|W_{h}+Z_{h}+\wt{\gamma}h|\leq\varepsilon\}}\Big)\\
\label{eq:Decompb1eps2} &=:e^{-\eta h}I_{3}(h)+2\wt{\gamma}he^{-\eta h}I_{4}(h)+\wt{\gamma}^{2}h^{2}e^{-\eta h}\,\wt{\bE}\Big(e^{-\wt{U}_{h}}\,{\bf 1}_{\{|W_{h}+Z_{h}+\wt{\gamma}h|\leq\varepsilon\}}\Big).
\end{align}
Clearly,
\begin{align}\label{eq:Limitb1eps23}
\wt{\gamma}^{2}h^{2}e^{-\eta h}\,\wt{\bE}\Big(e^{-\wt{U}_{h}}\,{\bf 1}_{\{|W_{h}+Z_{h}+\wt{\gamma}h|\leq\varepsilon\}}\Big)=O\big(h^{2}\big),\quad\text{as }\,h\rightarrow 0.
\end{align}
It remains to analyze the asymptotic behavior of the first two terms in \eqref{eq:Decompb1eps2}.

\medskip
\noindent
\textbf{Step 2.1.} We begin with the analysis of $I_{3}(h)$. Clearly,
\begin{align}
I_{3}(h)&=\wt{\bE}\Big(Z_{h}^{2}\,{\bf 1}_{\{|\sigma W_{h}+Z_{h}+\wt{\gamma}h|\leq\varepsilon\}}\Big)+\wt{\bE}\Big(\Big(e^{-\wt{U}_{h}}-1\Big)Z_{h}^{2}\,{\bf 1}_{\{|\sigma W_{h}+Z_{h}+\wt{\gamma}h|\leq\varepsilon\}}\Big)\\
&=h^{2/Y}\wt{\bE}\Big(Z_{1}^{2}\,{\bf 1}_{\{|\sigma\sqrt{h}W_{1}+h^{1/Y}\!Z_{1}+\wt{\gamma}h|\leq\varepsilon\}}\Big)+\wt{\bE}\Big(\Big(e^{-\wt{U}_{h}}-1\Big)Z_{h}^{2}\,{\bf 1}_{\{|\sigma W_{h}+Z_{h}+\wt{\gamma}h|\leq\varepsilon\}}\Big)\\
\label{eq:Decompb1eps21} &=:h^{2/Y}I_{31}(h)+I_{32}(h).
\end{align}
By the symmetry of $W_{1}$, we note that
\begin{align*}
\wt{\bE}\Big(Z_{1}^{2}\,{\bf 1}_{\{-\varepsilon\leq\sigma\sqrt{h}W_{1}+h^{1/Y}\!Z_{1}+\wt{\gamma}h\leq 0\}}\Big)=\wt{\bE}\Big(Z_{1}^{2}\,{\bf 1}_{\{0\leq\sigma\sqrt{h}W_{1}-h^{1/Y}\!Z_{1}-\wt{\gamma}h\leq\varepsilon\}}\Big).
\end{align*}
In what follows, we let $h>0$ small enough so that $\varepsilon-|\wt{\gamma}|h>0$.

To study the asymptotic behavior of $I_{31}(h)$, as $h\rightarrow 0$, let us first consider
\begin{align}
E_{1}^{\pm}(h)&:=\wt{\bE}\Big(Z_{1}^{2}\,{\bf 1}_{\{0\leq\sigma\sqrt{h}W_{1}\pm h^{1/Y}Z_{1}\pm\wt{\gamma}h\leq\varepsilon,\,W_{1}\geq 0,\,\pm Z_{1}\geq 0\}}\Big)\!=\!\!\int_{0}^{\frac{\varepsilon\mp\wt{\gamma}h}{\sigma\sqrt{h}}}\!\!\bigg(\!\int_{0}^{\frac{\varepsilon\mp\wt{\gamma}h-\sigma\sqrt{h}x}{h^{1/Y}}}\!\!\!u^{2}p_{Z}(\pm u)du\!\bigg)\phi(x)dx\\
&\,\,=\frac{\varepsilon\mp\wt{\gamma}h}{\sigma\sqrt{h}}\int_{0}^{1}\bigg(\int_{0}^{\frac{(\varepsilon\mp\wt{\gamma}h)(1-\omega)}{h^{1/Y}}}u^{2}p_{Z}(\pm u)\,du\bigg)\phi\bigg(\frac{\varepsilon\mp\wt{\gamma}h}{\sigma\sqrt{h}}\omega\bigg)d\omega\\
&\,\,=\frac{C_{\pm}\big(\varepsilon\mp\wt{\gamma}h\big)}{\sigma\sqrt{h}}\int_{0}^{1}\bigg(\int_{0}^{\frac{(\varepsilon\mp\wt{\gamma}h)(1-\omega)}{h^{1/Y}}}u^{1-Y}du\bigg)\phi\bigg(\frac{\varepsilon\mp\wt{\gamma}h}{\sigma\sqrt{h}}\omega\bigg)d\omega\\
\label{eq:Decompb1eps211PlusPlus} &\,\,\quad +\frac{\varepsilon\mp\wt{\gamma}h}{\sigma\sqrt{h}}\int_{0}^{1}\bigg(\int_{0}^{\frac{(\varepsilon\mp\wt{\gamma}h)(1-\omega)}{h^{1/Y}}}u^{2}\Big(p_{Z}(\pm u)-C_{\pm}u^{-1-Y}\Big)du\bigg)\phi\bigg(\frac{\varepsilon\mp\wt{\gamma}h}{\sigma\sqrt{h}}\omega\bigg)d\omega.
\end{align}
For the first term in \eqref{eq:Decompb1eps211PlusPlus}, we have
\begin{align}
\int_{0}^{1}\bigg(\int_{0}^{\frac{(\varepsilon\mp\wt{\gamma}h)(1-\omega)}{h^{1/Y}}}\!u^{1-Y}du\bigg)\phi\bigg(\frac{\varepsilon\mp\wt{\gamma}h}{\sigma\sqrt{h}}\omega\bigg)d\omega&=\frac{\big(\varepsilon\mp\wt{\gamma}h\big)^{2-Y}}{(2-Y)h^{(2-Y)/Y}}\int_{0}^{1}(1-\omega)^{2-Y}\phi\bigg(\frac{\varepsilon\mp\wt{\gamma}h}{\sigma\sqrt{h}}\omega\bigg)d\omega\\
\label{eq:Limitb1eps211PlusPlus1} &\sim\frac{\big(\varepsilon\mp\wt{\gamma}h\big)^{2-Y}}{2(2-Y)h^{(2-Y)/Y}}\cdot\frac{\sigma\sqrt{h}}{\varepsilon\mp\wt{\gamma}h},\quad\text{as }\,h\rightarrow 0.
\end{align}
For the second term in \eqref{eq:Decompb1eps211PlusPlus}, since $Y\in(1,2)$, we first observe that, for any $z>0$,
\begin{align}\label{eq:IntPoweru}
\int_{0}^{z}u^{2}\big(u^{-Y-1}\!\wedge u^{-2Y-1}\big)du=\frac{z^{2-Y}}{2\!-\!Y}{\bf 1}_{(0,1]}(z)+\frac{1\!-\!z^{2-2Y}}{2(Y\!-\!1)}{\bf 1}_{(1,\infty)}(z)\leq\frac{Y}{2(Y\!-\!1)(2\!-\!Y)}.
\end{align}
Hence, we deduce from \eqref{eq:2ndOrderEstDenZ} that
\begin{align*}
&\int_{0}^{1}\bigg(\int_{0}^{\frac{(\varepsilon\mp\wt{\gamma}h)(1-\omega)}{h^{1/Y}}}u^{2}\big|p_{Z}(\pm u)-C_{\pm}u^{-1-Y}\big|du\bigg)\phi\bigg(\frac{\varepsilon\mp\wt{\gamma}h}{\sigma\sqrt{h}}\omega\bigg)d\omega\\ &\quad\leq\frac{\wt{K}Y}{2(Y-1)(2-Y)}\int_{0}^{1}\phi\bigg(\frac{\varepsilon\mp\wt{\gamma}h}{\sigma\sqrt{h}}\omega\bigg)d\omega=O\big(\sqrt{h}\,\varepsilon^{-1}\big),\quad\text{as }\,h\rightarrow 0.
\end{align*}
Therefore, we obtain that
\begin{align}\label{eq:Limitb1eps211PlusPlus}
E_{1}^{\pm}(h)=\frac{C_{\pm}}{2(2-Y)}\,h^{1-2/Y}\varepsilon^{2-Y}+O(1),\quad\text{as }\,h\rightarrow 0.
\end{align}
Using the same argument as above and since $\varepsilon\gg h$, we also obtain that, when $\pm\wt{\gamma}>0$, as $h\rightarrow 0$,
\begin{align}
E_{2}^{\pm}(h)&:=\wt{\bE}\Big(Z_{1}^{2}\,{\bf 1}_{\{0\leq\sigma\sqrt{h}W_{1}\pm h^{1/Y}Z_{1}\pm\wt{\gamma}h\leq\varepsilon,\,W_{1}\leq 0,\,\pm Z_{1}\leq 0\}}\Big)\\
\label{eq:Limitb1eps211MinusMinus} &\,\,=\wt{\bE}\Big(Z_{1}^{2}\,{\bf 1}_{\{0\leq\sigma\sqrt{h}W_{1}\pm h^{1/Y}Z_{1}\leq\mp\wt{\gamma}h,\,W_{1}\leq 0,\,\pm Z_{1}\leq 0\}}\Big)=O\big(h^{3-Y-2/Y}\big)+O(1).
\end{align}

Next, we consider
\begin{align}
E_{3}^{\pm}(h)&:=\wt{\bE}\Big(Z_{1}^{2}\,{\bf 1}_{\{0\leq\sigma\sqrt{h}W_{1}\pm h^{1/Y}Z_{1}\pm\wt{\gamma}h\leq\varepsilon,\,W_{1}\geq 0,\,\pm Z_{1}\leq 0\}}\Big)\\
\label{eq:Decompb1eps211PlusMinus} &\,\,=\!\int_{0}^{\frac{\varepsilon\mp\wt{\gamma}h}{\sigma\sqrt{h}}}\!\!\bigg(\!\int_{-\frac{\sigma\sqrt{h}x}{h^{1/Y}}}^{0}\!\!u^{2}p_{Z}(\pm u)du\!\bigg)\phi(x)dx\!+\!\!\int_{\frac{\varepsilon\mp\wt{\gamma}h}{\sigma\sqrt{h}}}^{\infty}\!\!\bigg(\!\int_{-\frac{\sigma\sqrt{h}x}{h^{1/Y}}}^{\frac{\varepsilon\mp\wt{\gamma}h-\sigma\sqrt{h}x}{h^{1/Y}}}\!\!\!\!u^{2}p_{Z}(\pm u)du\!\bigg)\phi(x)dx.\quad
\end{align}
By \eqref{eq:1stOrderEstDenZ}, the first term in \eqref{eq:Decompb1eps211PlusMinus} is such that
\begin{align*}
&\int_{0}^{\frac{\varepsilon\mp\wt{\gamma}h}{\sigma\sqrt{h}}}\bigg(\int_{-\frac{\sigma\sqrt{h}x}{h^{1/Y}}}^{0}u^{2}p_{Z}(\pm u)\,du\bigg)\phi(x)\,dx\leq\wt{K}\int_{0}^{\frac{\varepsilon\mp\wt{\gamma}h}{\sigma\sqrt{h}}}\bigg(\int_{0}^{\frac{\sigma\sqrt{h}x}{h^{1/Y}}}u^{1-Y}du\bigg)\phi(x)\,dx\\
&\quad =\frac{\wt{K}\sigma^{2-Y}h^{2-Y/2-2/Y}}{2-Y}\int_{0}^{\frac{\varepsilon\mp\wt{\gamma}h}{\sigma\sqrt{h}}}x^{2-Y}\phi(x)\,dx=O\big(h^{2-Y/2-2/Y}\big),\quad\text{as }\,h\rightarrow 0.
\end{align*}
Similarly, the second term in \eqref{eq:Decompb1eps211PlusMinus} can be estimated as follows:
\begin{align*}
&\int_{\frac{\varepsilon\mp\wt{\gamma}h}{\sigma\sqrt{h}}}^{\infty}\bigg(\int_{-\frac{\sigma\sqrt{h}x}{h^{1/Y}}}^{\frac{\varepsilon\mp\wt{\gamma}h-\sigma\sqrt{h}x}{h^{1/Y}}}u^{2}p_{Z}(\pm u)\,du\bigg)\phi(x)\,dx\leq\wt{K}\int_{\frac{\varepsilon\mp\wt{\gamma}h}{\sigma\sqrt{h}}}^{\infty}\bigg(\int_{\frac{\sigma\sqrt{h}x-(\varepsilon\mp\wt{\gamma}h)}{h^{1/Y}}}^{\frac{\sigma\sqrt{h}x}{h^{1/Y}}}u^{1-Y}du\bigg)\phi(x)\,dx\\
&\quad\leq\frac{\wt{K}\sigma^{2-Y}h^{2-2/Y-Y/2}}{2-Y}\int_{\frac{\varepsilon\mp\wt{\gamma}h}{\sigma\sqrt{h}}}^{\infty}x^{2-Y}\phi(x)\,dx=o\big(h^{2-2/Y-Y/2}\big),\quad\text{as }\,h\rightarrow 0.
\end{align*}
Therefore, we obtain that
\begin{align}\label{eq:Limitb1eps211PlusMinus}
E_{3}^{\pm}(h)=O\big(h^{2-Y/2-2/Y}\big),\quad\text{as }\,h\rightarrow 0.
\end{align}

To complete the analysis for $I_{3}(h)$, it remains to study
\begin{align}
E_{4}^{\pm}(h)&:=\wt{\bE}\Big(Z_{1}^{2}\,{\bf 1}_{\{0\leq\sigma\sqrt{h}W_{1}\pm h^{1/Y}Z_{1}\pm\wt{\gamma}h\leq\varepsilon,\,W_{1}\leq 0,\,\pm Z_{1}\geq 0\}}\Big)=\!\int_{-\infty}^{0}\!\bigg(\!\int_{\frac{-\sigma\sqrt{h}x\mp\wt{\gamma}h}{h^{1/Y}}}^{\frac{\varepsilon\mp\wt{\gamma}h-\sigma\sqrt{h}x}{h^{1/Y}}}\!\!u^{2}p_{Z}(\pm u)du\!\bigg)\phi(x)dx\\ &\,\,=C_{\pm}\int_{-\infty}^{0}\bigg(\int_{\frac{-\sigma\sqrt{h}x\mp\wt{\gamma}h}{h^{1/Y}}}^{\frac{\varepsilon\mp\wt{\gamma}h-\sigma\sqrt{h}x}{h^{1/Y}}}u^{1-Y}du\bigg)\phi(x)\,dx\\
\label{eq:Decompb1eps211MinusPlus} &\quad\,\,+\int_{-\infty}^{0}\bigg(\int_{\frac{-\sigma\sqrt{h}x\mp\wt{\gamma}h}{h^{1/Y}}}^{\frac{\varepsilon\mp\wt{\gamma}h-\sigma\sqrt{h}x}{h^{1/Y}}}u^{2}\Big(p_{Z}(\pm u)-C_{\pm}u^{-1-Y}\Big)du\bigg)\phi(x)\,dx.
\end{align}
For the first term in \eqref{eq:Decompb1eps211MinusPlus}, we have
\begin{align}
&C_{\pm}\int_{-\infty}^{0}\bigg(\int_{\frac{-\sigma\sqrt{h}x\mp\wt{\gamma}h}{h^{1/Y}}}^{\frac{\varepsilon\mp\wt{\gamma}h-\sigma\sqrt{h}x}{h^{1/Y}}}u^{1-Y}du\bigg)\phi(x)\,dx\\
&\quad =\frac{C_{\pm}\,\varepsilon^{2-Y}}{(2-Y)h^{2/Y-1}}\int_{-\infty}^{0}\left(\bigg(1-\frac{\sigma\sqrt{h}x\pm\wt{\gamma}h}{\varepsilon}\bigg)^{2-Y}\!-\bigg(\!-\frac{\sigma\sqrt{h}x\pm\wt{\gamma}h}{\varepsilon}\bigg)^{2-Y}\right)\!\phi(x)\,dx\\
\label{eq:Limitb1eps211MinusPlus1} &\quad\sim\frac{C_{\pm}\,h^{1-2/Y}\varepsilon^{2-Y}}{2(2-Y)},\quad\text{as }\,h\rightarrow 0.
\end{align}
For the second term in \eqref{eq:Decompb1eps211MinusPlus}, we deduce from \eqref{eq:IntPoweru} that
\begin{align*}
\int_{-\infty}^{0}\bigg(\int_{\frac{-\sigma\sqrt{h}x\mp\wt{\gamma}h}{h^{1/Y}}}^{\frac{\varepsilon\mp\wt{\gamma}h-\sigma\sqrt{h}x}{h^{1/Y}}}u^{2}\Big|p_{Z}(\pm u)-C_{\pm}u^{-1-Y}\Big|du\bigg)\phi(x)\,dx=O(1),\quad h\rightarrow 0.
\end{align*}
Therefore, we obtain that
\begin{align}\label{eq:Limitb1eps211MinusPlus}
E_{4}^{\pm}(h)=\frac{C_{\pm}}{2(2-Y)}h^{1-2/Y}\varepsilon^{2-Y}+O(1),\quad\text{as }\,h\rightarrow 0.
\end{align}
By combining \eqref{eq:Limitb1eps211PlusPlus}, \eqref{eq:Limitb1eps211MinusMinus}, \eqref{eq:Limitb1eps211PlusMinus}, and \eqref{eq:Limitb1eps211MinusPlus}, we conclude that
\begin{align}\label{eq:Limitb1eps211}
I_{31}(h)=\sum_{i=1}^{4}\big(E_{i}^{+}(h)\!+\!E_{i}^{-}(h)\big)=\frac{C_{+}\!+\!C_{-}}{2-Y}h^{1-2/Y}\varepsilon^{2-Y}+O\big(h^{2-Y/2-2/Y}\big),\quad\text{as }\,h\rightarrow 0.\quad
\end{align}

Next, we will study the asymptotic behavior of $I_{32}(h)$, as $h\rightarrow 0$. Clearly, by Cauchy-Schwarz inequality and self-similarity of $Z_{h}$ and $W_{h}$ under $\widetilde{\bP}$, we have that
\begin{align}\label{eq:Decompb1eps212}
I_{32}(h)\leq h^{2/Y}\left(\wt{\bE}\bigg(\!\Big(e^{-\wt{U}_{h}}-1\Big)^{2}\bigg)\right)^{1/2}\bigg(\wt{\bE}\Big(Z_{1}^{4}\,{\bf 1}_{\{|\sigma\sqrt{h}W_{1}+h^{1/Y}\!Z_{1}+\wt{\gamma}h|\leq\varepsilon\}}\Big)\bigg)^{1/2}.
\end{align}
By Assumption \ref{assump:Funtq}-(v) and denoting $\wt{C}_{\ell}=\int_{\bR_{0}}\big(e^{-\ell\varphi(x)}-1+\ell\varphi(x)\big)\tilde{\nu}(dx)$, $\ell=1,2$, we first have
\begin{align}\label{eq:Limitb1eps2121}
\wt{\bE}\bigg(\Big(e^{-\wt{U}_{h}}-1\Big)^{2}\bigg)=e^{\wt{C}_{2}h}-2e^{\wt{C}_{1}h}+1\sim\big(\wt{C}_{2}-2\wt{C}_{1}\big)h,\quad\text{as }\,h\rightarrow 0.
\end{align}
The analysis of the asymptotic behavior, as $h\rightarrow 0$, of the second factor in \eqref{eq:Decompb1eps212} is similar to that of $I_{31}(h)$. More precisely, we first consider
\begin{align*}
F_{1}^{\pm}(h)&:=\wt{\bE}\Big(Z_{1}^{4}\,{\bf 1}_{\{0\leq\sigma\sqrt{h}W_{1}\pm h^{1/Y}Z_{1}\pm\wt{\gamma}h\leq\varepsilon,\,W_{1}\geq 0,\,\pm Z_{1}\geq 0\}}\Big)\\
&\,\,=\frac{C_{\pm}\big(\varepsilon\mp\wt{\gamma}h\big)}{\sigma\sqrt{h}}\int_{0}^{1}\bigg(\int_{0}^{\frac{(\varepsilon\mp\wt{\gamma}h)(1-\omega)}{h^{1/Y}}}u^{3-Y}du\bigg)\phi\bigg(\frac{\varepsilon\mp\wt{\gamma}h}{\sigma\sqrt{h}}\omega\bigg)d\omega\\
&\,\,\quad +\frac{\varepsilon\mp\wt{\gamma}h}{\sigma\sqrt{h}}\int_{0}^{1}\bigg(\int_{0}^{\frac{(\varepsilon\mp\wt{\gamma}h)(1-\omega)}{h^{1/Y}}}u^{4}\Big(p_{Z}(\pm u)-C_{\pm}u^{-1-Y}\Big)du\bigg)\phi\bigg(\frac{\varepsilon\mp\wt{\gamma}h}{\sigma\sqrt{h}}\omega\bigg)d\omega.
\end{align*}
A similar argument as in \eqref{eq:Limitb1eps211PlusPlus1} shows that
\begin{align*}
\frac{C_{\pm}\big(\varepsilon\mp\wt{\gamma}h\big)}{\sigma\sqrt{h}}\int_{0}^{1}\bigg(\int_{0}^{\frac{(\varepsilon\mp\wt{\gamma}h)(1-\omega)}{h^{1/Y}}}u^{3-Y}du\bigg)\phi\bigg(\frac{\varepsilon\mp\wt{\gamma}h}{\sigma\sqrt{h}}\omega\bigg)d\omega\sim\frac{C_{\pm}}{2(4-Y)}h^{1-4/Y}\varepsilon^{4-Y},\quad\text{as }\,h\rightarrow 0,
\end{align*}
and by \eqref{eq:2ndOrderEstDenZ},
\begin{align*}
&\frac{\varepsilon\mp\wt{\gamma}h}{\sigma\sqrt{h}}\int_{0}^{1}\bigg(\int_{0}^{\frac{(\varepsilon\mp\wt{\gamma}h)(1-\omega)}{h^{1/Y}}}u^{4}\big|p_{Z}(u)-Cu^{-1-Y}\big|du\bigg)\phi\bigg(\frac{\varepsilon\mp\wt{\gamma}h}{\sigma\sqrt{h}}\omega\bigg)d\omega\\ &\quad\leq\frac{\varepsilon\!\mp\!\wt{\gamma}h}{\sigma\sqrt{h}}\cdot\frac{\wt{K}\big(\varepsilon\mp\wt{\gamma}h\big)^{4-2Y}}{2(2-Y)h^{(4-2Y)/Y}}\!\int_{0}^{1}(1-\omega)^{4-2Y}\phi\bigg(\frac{\varepsilon\!\mp\!\wt{\gamma}h}{\sigma\sqrt{h}}\omega\bigg)d\omega=O\big(h^{2-4/Y}\varepsilon^{4-2Y}\big),\quad\text{as }\,h\rightarrow 0.
\end{align*}
Hence, we obtain that
\begin{align}\label{eq:Limitb1eps2122PlusPlus}
F_{1}^{\pm}(h)=\frac{C_{\pm}}{2(4-Y)}h^{1-4/Y}\varepsilon^{4-Y}+O\big(h^{2-4/Y}\varepsilon^{{4}-2Y}\big),\quad\text{as }\,h\rightarrow 0.
\end{align}
Using the same argument as above and since $\varepsilon\gg h$, we also obtain that, when $\pm\wt{\gamma}>0$, as $h\rightarrow 0$,
\begin{align}
F_{2}^{\pm}(h)&:=\wt{\bE}\Big(Z_{1}^{4}\,{\bf 1}_{\{0\leq\sigma\sqrt{h}W_{1}\pm h^{1/Y}Z_{1}\pm\wt{\gamma}h\leq\varepsilon,\,W_{1}\leq 0,\,\pm Z_{1}\leq 0\}}\Big)\\
\label{eq:Limitb1eps2122MinusMinus} &\,\,=\wt{\bE}\Big(Z_{1}^{4}\,{\bf 1}_{\{0\leq\sigma\sqrt{h}W_{1}\pm h^{1/Y}Z_{1}\leq\mp\wt{\gamma}h,\,W_{1}\geq 0,\,\pm Z_{1}\leq 0\}}\Big)=O\big(h^{6-4/Y-2Y}\big).
\end{align}
Moreover, using arguments similar to those for $E_{3}^{\pm}(h)$, we deduce that
\begin{align}\label{eq:Limitb1eps2122PlusMinus}
F_{3}^{\pm}(h):=\wt{\bE}\Big(Z_{1}^{4}\,{\bf 1}_{\{0\leq\sigma\sqrt{h}W_{1}\pm h^{1/Y}Z_{1}\pm\wt{\gamma}h\leq\varepsilon,\,W_{1}\geq 0,\,\pm Z_{1}\leq 0\}}\Big)=O\big(h^{3-4/Y-Y/2}\big),\quad\text{as }\,h\rightarrow 0.\quad
\end{align}
Finally, we consider
\begin{align*}
F_{4}^{\pm}(h)&:=\wt{\bE}\Big(Z_{1}^{4}\,{\bf 1}_{\{0\leq\sigma\sqrt{h}W_{1}\pm h^{1/Y}Z_{1}\pm\wt{\gamma}h\leq\varepsilon,\,W_{1}\leq 0,\,\pm Z_{1}\geq 0\}}\Big)=\!\int_{-\infty}^{0}\!\bigg(\!\int_{\frac{-\sigma\sqrt{h}x\mp\wt{\gamma}h}{h^{1/Y}}}^{\frac{\varepsilon\mp\wt{\gamma}h-\sigma\sqrt{h}x}{h^{1/Y}}}\!\!u^{4}p_{Z}(\pm u)du\!\bigg)\phi(x)dx\\ &\,\,=C_{\pm}\!\int_{-\infty}^{0}\!\!\!\bigg(\!\int_{\frac{-\sigma\sqrt{h}x\mp\wt{\gamma}h}{h^{1/Y}}}^{\frac{\varepsilon\mp\wt{\gamma}h-\sigma\sqrt{h}x}{h^{1/Y}}}\!\!\!u^{3-Y}\!du\!\bigg)\phi(x)dx\!+\!\!\int_{-\infty}^{0}\!\!\left(\!\int_{\frac{-\sigma\sqrt{h}x\mp\wt{\gamma}h}{h^{1/Y}}}^{\frac{\varepsilon\mp\wt{\gamma}h-\sigma\sqrt{h}x}{h^{1/Y}}}\!\!\!u^{4}\bigg(\!p_{Z}(\pm u)\!-\!\frac{C_{\pm}}{u^{1+Y}}\!\bigg)du\!\right)\!\phi(x)dx.
\end{align*}
A similar argument as in \eqref{eq:Limitb1eps211MinusPlus1} shows that
\begin{align*}
C_{\pm}\int_{-\infty}^{0}\bigg(\int_{\frac{-\sigma\sqrt{h}x\mp\wt{\gamma}h}{h^{1/Y}}}^{\frac{\varepsilon\mp\wt{\gamma}h-\sigma\sqrt{h}x}{h^{1/Y}}}u^{3-Y}du\bigg)\phi(x)\,dx\sim\frac{C_{\pm}}{2(4-Y)}h^{1-4/Y}\varepsilon^{4-Y},\quad\text{as }\,h\rightarrow 0,
\end{align*}
and by \eqref{eq:2ndOrderEstDenZ},
\begin{align*}
&\int_{-\infty}^{0}\!\bigg(\int_{\frac{-\sigma\sqrt{h}x\mp\wt{\gamma}h}{h^{1/Y}}}^{\frac{\varepsilon\mp\wt{\gamma}h-\sigma\sqrt{h}x}{h^{1/Y}}}\!u^{4}\Big|p_{Z}(\pm u)-C_{\pm}u^{-1-Y}\Big|du\bigg)\phi(x)\,dx\leq\wt{K}\int_{-\infty}^{0}\!\bigg(\int_{\frac{-\sigma\sqrt{h}x\mp\wt{\gamma}h}{h^{1/Y}}}^{\frac{\varepsilon\mp\wt{\gamma}h-\sigma\sqrt{h}x}{h^{1/Y}}}\!u^{3-2Y}du\bigg)\phi(x)\,dx\\
&\quad =\frac{\wt{K}\varepsilon^{4-2Y}}{2(2\!-\!Y)h^{4/Y-2}}\int_{-\infty}^{0}\!\left(\!\bigg(\!1\!-\!\frac{\sigma\sqrt{h}x\!\pm\!\wt{\gamma}h}{\varepsilon}\bigg)^{4-2Y}\!\!\!-\!\bigg(\!\!-\!\frac{\sigma\sqrt{h}x\!\pm\!\wt{\gamma}h}{\varepsilon}\bigg)^{4-2Y}\right)\!\phi(x)\,dx=O\bigg(\frac{\varepsilon^{4-2Y}}{h^{4/Y-2}}\bigg).
\end{align*}
Hence, we obtain that
\begin{align}\label{eq:Limitb1eps2122MinusPlus}
F_{4}^{\pm}(h)=\frac{C_{\pm}}{2(4-Y)}h^{1-4/Y}\varepsilon^{4-Y}+O\big(h^{2-4/Y}\varepsilon^{4-2Y}\big),\quad\text{as }\,h\rightarrow 0.
\end{align}
Combining \eqref{eq:Limitb1eps2122PlusPlus}, \eqref{eq:Limitb1eps2122MinusMinus}, \eqref{eq:Limitb1eps2122PlusMinus}, and \eqref{eq:Limitb1eps2122MinusPlus}, leads to
\begin{align}
&\wt{\bE}\Big(Z_{1}^{4}\,{\bf 1}_{\{|\sigma\sqrt{h}W_{1}+h^{1/Y}\!Z_{1}+\wt{\gamma}h|\leq\varepsilon\}}\Big)=\sum_{i=1}^{4}\big(F_{i}^{+}(h)+F_{i}^{-}(h)\big)\\
\label{eq:Limitb1eps2122} &\quad =\frac{{\big(C_{+}+C_{-}\big)}h^{1-4/Y}\varepsilon^{4-Y}}{4-Y}+O\big(h^{2-4/Y}\varepsilon^{4-2Y}\big)+O\big(h^{3-4/Y-Y/2}\big),\quad\text{as }\,h\rightarrow 0.
\end{align}
Therefore, by combining \eqref{eq:Decompb1eps212}, \eqref{eq:Limitb1eps2121}, and \eqref{eq:Limitb1eps2122}, we have
\begin{align}\label{eq:Limitb1eps212}
I_{32}(h)=O\big(h\,\varepsilon^{2-Y/2}\big)+O\big(h^{3/2}\varepsilon^{2-Y}\big)+O\big(h^{2-Y/4}\big),\quad\text{as }\,h\rightarrow 0.
\end{align}
Finally, by combining \eqref{eq:Decompb1eps21}, \eqref{eq:Limitb1eps211}, and \eqref{eq:Limitb1eps212}, we obtain that
\begin{align}\label{eq:Limitb1eps21}
I_{3}(h)=\frac{C_{+}+C_{-}}{2-Y}\,h\varepsilon^{2-Y}+O\big(h\varepsilon^{2-Y/2}\big)+O\big(h^{2-Y/2}\big),\quad\text{as }\,h\rightarrow 0.
\end{align}

\smallskip
\noindent
\textbf{Step 2.2.} In this step, we will investigate the asymptotic behavior of $I_{4}(h)$, as $h\rightarrow 0$. Note that
\begin{align*}
I_{4}(h)&=h^{1/Y}\wt{\bE}\Big(Z_{1}{\bf 1}_{\{|\sigma\sqrt{h}W_{1}+h^{1/Y}\!Z_{1}+\wt{\gamma}h|\leq\varepsilon\}}\Big)+\wt{\bE}\bigg(\Big(e^{-\wt{U}_{h}}\!-1\Big)Z_{h}{\bf 1}_{\{|\sigma W_{h}+Z_{h}+\wt{\gamma}h|\leq\varepsilon\}}\bigg)\\
&\,\,\leq h^{1/Y}\bigg(\wt{\bE}\Big(Z_{1}^{2}\,{\bf 1}_{\{|\sigma\sqrt{h}W_{1}+h^{1/Y}\!Z_{1}+\wt{\gamma}h|\leq\varepsilon\}}\Big)\bigg)^{1/2}\left(1+\left(\wt{\bE}\bigg(\Big(e^{-\wt{U}_{h}}-1\Big)^{2}\bigg)\right)^{1/2}\right)\\
&\,\,=O\left(h^{1/Y}\bigg(\wt{\bE}\Big(Z_{1}^{2}\,{\bf 1}_{\{|\sigma\sqrt{h}W_{1}+h^{1/Y}\!Z_{1}+\wt{\gamma}h|\leq\varepsilon\}}\Big)\bigg)^{1/2}\right),\quad h\rightarrow 0,
\end{align*}
where the second inequality above follows from Cauchy-Schwarz inequality. Therefore, by \eqref{eq:Decompb1eps21} and \eqref{eq:Limitb1eps211}, we obtain that
\begin{align}\label{eq:Limitb1eps22}
I_{4}(h)=O\big(\sqrt{h}\,\varepsilon^{1-Y/2}\big)+O\big(h^{1-Y/4}\big),\quad\text{as }\,h\rightarrow 0.
\end{align}

Finally, by combining \eqref{eq:Decompb1eps2}, \eqref{eq:Limitb1eps23}, \eqref{eq:Limitb1eps21}, and \eqref{eq:Limitb1eps22}, we conclude that
\begin{align}\label{eq:Limitb1eps2}
\bE\Big(J_{h}^{2}\,{\bf 1}_{\{|\sigma W_{h}+J_{h}|\leq\varepsilon\}}\Big)=\frac{C_{+}+C_{-}}{2-Y}\,h\varepsilon^{2-Y}+O\big(h\varepsilon^{2-Y/2}\big)+O\big(h^{2-Y/2}\big),\quad\text{as }\,h\rightarrow 0,
\end{align}
which completes the analysis of Step 2.

\bigskip
\noindent
\textbf{Step 3.} In this last step, we will study the asymptotic behavior of the third term in \eqref{eq:Decompb1eps}, as $h\rightarrow 0$. By \eqref{eq:DenTranTildePP} and \eqref{eq:DecompUpm}, we first decompose it as
\begin{align}
&\bE\Big(W_{h}J_{h}{\bf 1}_{\{|\sigma W_{h}+J_{h}|\leq\varepsilon\}}\Big)=\wt{\bE}\Big(e^{-\wt{U}_{h}-\eta h}\,W_{h}J_{h}{\bf 1}_{\{|\sigma W_{h}+J_{h}|\leq\varepsilon\}}\Big)\\
&\quad =\sqrt{h}\,e^{-\eta h}\,\wt{\bE}\Big(W_{1}J_{h}{\bf 1}_{\{|\sigma\sqrt{h}W_{1}+J_{h}|\leq\varepsilon\}}\Big)+\sqrt{h}\,e^{-\eta h}\,\wt{\bE}\Big(\Big(e^{-\wt{U}_{h}}-1\Big)W_{1}J_{h}{\bf 1}_{\{|\sigma\sqrt{h}W_{1}+J_{h}|\leq\varepsilon\}}\Big)\\
\label{eq:Decompb1eps3} &\quad =:e^{-\eta h}\sqrt{h}\,I_{5}(h)+e^{-\eta h}\sqrt{h}\,I_{6}(h).
\end{align}
For $I_{5}(h)$, by conditioning on $J_{h}$, and using the fact that, for any $x_{1},x_{2}\in\bR$ with $x_{1}<x_{2}$,
\begin{align*}
\wt{\bE}\Big(W_{1}{\bf 1}_{\{W_{1}\in [x_{1},x_{2}]\}}\Big)=\phi(x_{1})-\phi(x_{2}),
\end{align*}
we obtain from \eqref{eq:LimitExpJhphi} that, as $h\rightarrow 0$,
\begin{align}\label{eq:Limitb1eps31}
I_{5}(h)=\wt{\bE}\left(J_{h}\bigg(\phi\bigg(\frac{\varepsilon+J_{h}}{\sigma\sqrt{h}}\bigg)-\phi\bigg(\frac{\varepsilon-J_{h}}{\sigma\sqrt{h}}\bigg)\bigg)\right)=O\Big(h\,e^{-\varepsilon^{2}/(2\sigma^{2}h)}\Big)+O\big(h^{3/2}\varepsilon^{-Y}\big).
\end{align}
As for $I_{6}(h)$, by Cauchy-Schwarz inequality, \eqref{eq:DecompJZpm}, \eqref{eq:Decompb1eps21}, \eqref{eq:Limitb1eps211}, and \eqref{eq:Limitb1eps2121}, we obtain that
\begin{align}
\big|I_{6}(h)\big|&\leq\bigg(\wt{\bE}\bigg(\Big(e^{-\wt{U}_{h}}-1\Big)^{2}\bigg)\bigg)^{1/2}\bigg(\wt{\bE}\Big(J_{h}^{2}\,{\bf 1}_{\{|\sigma\sqrt{h}W_{1}+J_{h}|\leq\varepsilon\}}\Big)\bigg)^{1/2}\\
&\leq\bigg(\wt{\bE}\bigg(\Big(e^{-\wt{U}_{h}}-1\Big)^{2}\bigg)\bigg)^{1/2}\bigg(\wt{\bE}\Big(2\big(Z_{h}^{2}+\wt{\gamma}^{2}h^{2}\big){\bf 1}_{\{|\sigma\sqrt{h}W_{1}+J_{h}|\leq\varepsilon\}}\Big)\bigg)^{1/2}\\
\label{eq:Limitb1eps32} &=O\big(h\varepsilon^{1-Y/2}\big),\quad\text{as }\,h\rightarrow 0.
\end{align}
Therefore, by combining \eqref{eq:Decompb1eps3}, \eqref{eq:Limitb1eps31}, and \eqref{eq:Limitb1eps32}, we obtain that, as $h\rightarrow 0$,
\begin{align}\label{eq:Limitb1eps3}
\bE\Big(W_{h}J_{h}{\bf 1}_{\{|\sigma W_{h}+J_{h}|\leq\varepsilon\}}\Big)&=O\Big(h^{3/2}e^{-\varepsilon^{2}/(2\sigma^{2}h)}\Big)+O\big(h^{2}\varepsilon^{-Y}\big)+O\big(h^{3/2}\varepsilon^{1-Y/2}\big),
\end{align}
which completes the analysis in Step 3.

\bigskip
Finally, by combining \eqref{eq:Decompb1eps}, \eqref{eq:Limitb1eps1}, \eqref{eq:Limitb1eps2}, and \eqref{eq:Limitb1eps3}, we conclude that, as $h\rightarrow 0$,
\begin{align}
\bE\big(b_{1}(\varepsilon)\big)\!=\!\sigma^{2}h\!-\!{\frac{\sigma\varepsilon\sqrt{2h}}{\sqrt{\pi}}e^{-\varepsilon^{2}/(2\sigma^{2}h)}\!+\!\frac{C_{+}\!\!+\!C_{-}}{2-Y}h\varepsilon^{2-Y}}\!\!+\!O\Big(he^{-\varepsilon^{2}/(2\sigma^{2}h)}\Big)\!+\!O\big(h\varepsilon^{2-Y/2}\big)\!+\!O\big(h^{2-Y/2}\big),
\end{align}
which completes the proof of {\Blue the theorem.}

\bibliographystyle{plain}

\begin{thebibliography}{99}

\bibitem{AitSahaliaJacod:2009}
Y. A\"{i}t-Sahalia and J. Jacod.
\newblock{Estimating the Degree of Activity of Jumps in High Frequency Data}.
\newblock {\em Ann. Stat.}, 37(5A):2202$-$2244, 2009.

\bibitem{Applebaum:2004}
D. Applebaum.
\newblock {\em L\'{e}vy Processes and Stochastic Calculus}, 2nd Ed..
\newblock {\em Cambridge Stud. Adv. Math.}, 116, Cambridge University Press, Cambridge, U.K., 2004.

\bibitem{Belomestny:2010}
D. Belomestny.
\newblock {Spectral Estimation of the Fractional Order of a L\'{e}vy Process}.
\newblock {\em Ann. Stat.}, 38(1):317$-$351, 2010.

\bibitem{Bull:2016}
A. D. Bull.
\newblock {Near-Optimal Estimation of Jump Activity in Semimartingales}.
\newblock {\em Ann. Stat.}, 44(1):58$-$86, 2016.

\bibitem{CarrGemanMadanYor:2002}
P. Carr, H. Geman, D. B. Madan, and M. Yor.
\newblock {The Fine Structure of Asset Returns: An Empirical Investigation}.
\newblock {\em J. Bus.}, 75(2):305$-$332, 2002.

\bibitem{ContTankov:2004}
R. Cont and P. Tankov.
\newblock {\em Financial Modelling with Jump Processes}.
\newblock {\em Chapman \& Hall/CRC Financ. Math. Ser.}, Chapman \& Hall/CRC, Boca Raton, FL, U.S.A., 2004.

\bibitem{FigueroaLopez:2012}
J. E. Figueroa-L\'{o}pez.
\newblock {Statistical Estimation of L\'{e}vy-Type Stochastic Volatility Models}.
\newblock {\em Ann. Finance}, 8(2):309$-$335, 2012.

\bibitem{FigueroaLopezGongHoudre:2016}
J. E. Figueroa-L\'{o}pez, R. Gong, and C. Houdr\'{e}.
\newblock {High-Order Short-Time Expansions for ATM Option Prices of Exponential L\'{e}vy Models}.
\newblock {\em Math. Financ.}, 26(3):516$-$557, 2016.

\bibitem{FigueroaLopezGongHoudre:2017}
J. E. Figueroa-L\'{o}pez, R. Gong, and C. Houdr\'{e}.
\newblock {Third-Order Short-Time Expansions for Close-to-the-Money Option Prices under the CGMY Model}.
\newblock {\em Appl. Math. Financ.}, 24(6):547$-$574, 2017.

\bibitem{FigueroaLopezMancini:2019}
J. E. Figueroa-L\'{o}pez and C. Mancini.
\newblock {Optimum Thresholding Using Mean and Conditional Mean Square Error}.
\newblock {\em J. Econom.}, 208(1):179$-$210, 2019.

\bibitem{FigueroaLopezOlafsson:2019(1)}
J. E. Figueroa-L\'{o}pez and S. \'{O}lafsson.
\newblock {Short-Time Expansions for Close-to-the-Money Options under a L\'{e}vy Jump Model with Stochastic Volatility}.
\newblock {\em Financ. Stoch.}, 20(1):219$-$265, 2016.

\bibitem{FigueroaLopezOlafsson:2019(2)}
J. E. Figueroa-L\'{o}pez and S. \'{O}lafsson.
\newblock {Short-Time Asymptotics for the Implied Volatility Skew under a Stochastic Volatility Model with L\'{e}vy Jumps}.
\newblock {\em Financ. Stoch.}, 20(4):973$-$1020, 2016.

\bibitem{JacodTodorov:2014}
J. Jacod and V. Todorov.
\newblock {Efficient Estimation of Integrated Volatility in Presence of Infinite Variation Jumps}.
\newblock {\em Ann. Stat.}, 42(3):1029$-$1069, 2014.

\bibitem{Kawai:2010}
R. Kawai.
\newblock {On Sequential Calibration for an Asset Price Model with Piecewise L\'{e}vy Processes}.
\newblock {\em IAENG Int. J. Appl. Math.}, 40(4):239$-$246, 2010.

\bibitem{KyprianouSchoutensWilmott:2005}
A. Kyprianou, W. Schoutens, and P. Wilmott.
\newblock {\em Exotic Option Pricing and Advanced L\'{e}vy Models}.
\newblock {John Wiley \& Sons Ltd.}, ChiChester, England, 2005.

\bibitem{Masuda}
H. Masuda. Parametric estimation of L\'evy processes. In L\'evy Matters IV, p. 179-286. Springer.
.
\bibitem{Mies:2019}
F. Mies.
\newblock {Rate-Optimal Estimation of the Blumenthal-Getoor Index of a L\'{e}vy Process}.
\newblock {\em Electronic Journal of Statistics}, 14(2):4165$-$4206, 2020.

\bibitem{Reiss:2013}
M. Rei{\ss}.
\newblock {Testing the Characteristics of a L\'{e}vy Process}.
\newblock {\em Stoch. Proc. Appl.}, 123(7):2808$-$2828, 2013.

\bibitem{Rosinski:2007}
J. Rosi\'{n}ski.
\newblock {Tempering Stable Processes}.
\newblock {\em Stoch. Proc. Appl.}, 117(6):677$-$707, 2007.

\bibitem{Sato:1999}
K. Sato.
\newblock {\em L\'{e}vy Processes and Infinitely Divisible Distributions}.
\newblock {\em Cambridge Stud. Adv. Math.}, 68, Cambridge University Press, Cambridge, U.K., 1999.

\bibitem{Tankov:2010}
P. Tankov.
\newblock {Pricing and Hedging in Exponential L\'{e}vy Models: Review of Recent Results}.
\newblock {{\em Paris-Princeton Lectures in Mathematical Finance 2010} (R. Carmona, E. \c{C}inlar, I. Ekeland, E. Jouini, J. A. Scheinkman, and N. Touzi (eds.))}, {\em Lect. Notes Math.}, 2003, 319$-$359, 2010.

\bibitem{FigueroaQiKuffner}
Q. Wang, J.E. Figueroa-L\'opez, and T. Kuffner.
\newblock {Bayesian Inference on Volatility in the Presence of Infinite Jump Activity and Microstructure Noise.} 
\newblock {\em Electronic Journal Of Statistics.}, 15(1):506$-$553, 2021.

\bibitem{ZhangMyklandAitSahalia:2005}
L. Zhang, P. A. Mykland, and Y. A\"{i}t-Sahalia.
\newblock {A Tale of Two Time Scales: Determining Integrated Volatility with Noisy High-Frequency Data}.
\newblock {\em J. Am. Stat. Assoc.}, 100(472):1394$-$1411, 2005.







\end{thebibliography}

\end{document}